\documentclass[lettersize,onecolumn]{IEEEtran}

\usepackage[utf8]{inputenc}

\usepackage{blindtext}
\usepackage{titlesec}
\usepackage{thmtools}
\usepackage{thm-restate}

\AtBeginDocument{%
  }

\makeatletter
\def\endthebibliography{%
	\def\@noitemerr{\@latex@warning{Empty `thebibliography' environment}}%
	\endlist
}
\makeatother

\newcommand{\rootpath}{TIT/}
\makeatletter
\providecommand*{\input@path}{}
\edef\input@path{{\rootpath}\input@path}%
\makeatother

\usepackage{cite}
\usepackage{amsmath,amssymb,amsfonts}
\usepackage{algorithmic}
\usepackage{graphicx}
\usepackage{textcomp}
\usepackage{xcolor}
\usepackage[ruled,vlined]{algorithm2e}
\usepackage{algorithmic}
\usepackage{makecell}
\usepackage{diagbox}
\usepackage{multirow}
\usepackage{wrapfig}

\usepackage{xcolor}

\usepackage{nicefrac}
\usepackage{hyperref}
\hypersetup{
    colorlinks=true,
    linkcolor=blue,
    filecolor=magenta,      
    urlcolor=cyan
}
\usepackage{mathtools}
\usepackage{amsmath,amsfonts}

\usepackage{amssymb}
\usepackage[ruled,vlined]{algorithm2e}
\usepackage{algorithmic}
\usepackage{makecell}
\usepackage{diagbox}
\usepackage{multirow}
\usepackage{textgreek}
\usepackage{bbm}
\usepackage{comment}

\usepackage{balance}

\usepackage{caption}
\usepackage{subcaption}

\usepackage{amsthm}
\newtheorem{theorem}{Theorem}

\newtheorem{proposition}{Proposition}
\newtheorem{property}{Property}
\newtheorem{lemma}{Lemma}

\usepackage[nameinlink,capitalize]{cleveref}

\crefname{section}{\S}{\S}
\Crefname{section}{\S}{\S}
\crefname{appendix}{App.}{Apps.}
\Crefname{appendix}{App.}{Apps.}
\crefname{theorem}{Thm.}{Thms.}
\Crefname{theorem}{Thm.}{Thms.}
\crefname{proposition}{Prop.}{Props.}
\Crefname{proposition}{Prop.}{Props.}
\crefname{algorithm}{Alg.}{Algs.}
\Crefname{algorithm}{Alg.}{Algs.}
\crefname{assumption}{Asm.}{Asms.}
\Crefname{assumption}{Asm.}{Asms.}
\crefname{mechanism}{Mech.}{Mechs.}
\Crefname{mechanism}{Mech.}{Mechs.}
\crefname{property}{Property}{Properties}
\Crefname{property}{Property}{Properties}
\newtheorem{definition}{Definition}

\usepackage{mfirstuc}
\usepackage[T1]{fontenc}

\newcounter{packednmbr}

\newcommand{\cameraready}[1]{#1}
\newcommand{\camerareadydelete}[1]{}

\newcommand{\bra}[1]{\left( #1 \right)}
\newcommand{\brb}[1]{\left[ #1 \right]}

\newcommand{\brc}[1]{\left\{ #1 \right\}}
\newcommand{\brd}[1]{\left| #1 \right|}

\DeclarePairedDelimiter\abs{\lvert}{\rvert}

\newcommand{\floor}[1]{\left\lfloor {#1} \right\rfloor}
\newcommand{\ceil}[1]{\left\lceil {#1} \right\rceil}

\newcommand{\privname}{statistic maximal leakage\xspace}
\newcommand{\PrivName}{Statistic Maximal Leakage\xspace}
\newcommand{\Privname}{Statistic maximal leakage\xspace}

\newcommand{\rr}{RR}
\newcommand{\qm}{QM}

\newcommand{\datamechanism}{data release mechanism}

\newcommand{\datamechanisms}{data release mechanisms}

\newcommand{\privacynotation}{\Pi}
\newcommand{\distortionnotation}{\Delta}

\newcommand{\minentropy}{\mathcal{L}_{MEL}}
\newcommand{\maxl}{\mathcal{L}_{ML}}
\newcommand{\sml}{\privacynotation_{\mech, \secretnotation}}
\newcommand{\smldistortion}{\distortionnotation_{\mech}}

\newcommand{\rvprivatenotation}{X}
\newcommand{\rvreleasenotation}{Y}

\newcommand{\rvprivatewithparam}[1]{\rvprivatenotation_{#1}}
\newcommand{\rvreleasewithparam}[1]{\rvreleasenotation_{#1}}
\newcommand{\rvparamnotation}{\theta}

\newcommand{\releaservparamnotation}{\theta'}

\newcommand{\distributionnotation}{\nu}
\newcommand{\distributionof}[1]{\distributionnotation_{#1}}
\newcommand{\privatedistribution}{\distributionof{\rvprivatewithparam{\rvparamnotation}}}
\newcommand{\releasedistribution}{\distributionof{\rvreleasewithparam{\releaservparamnotation}}}

\newcommand{\paramdistribution}{\mathbb{P}_{\Theta}}

\newcommand{\pdfnotation}{f}
\newcommand{\pdf}[1]{\pdfnotation_{#1}}
\newcommand{\pdfof}[2]{\pdfnotation_{#1}\bra{#2}}

\newcommand{\secretrv}{g}
\newcommand{\secretvalueset}{\mathbf{G}}
\newcommand{\paramset}{\mathbf{\Theta}}
\newcommand{\secretset}{\mathbf{\Theta}_{\secretrv}}
\newcommand{\secretsetof}[1]{\mathbf{\Theta}_{\secretrv_{#1}}}
\newcommand{\secretsetofabbr}[1]{\mathbf{\Theta}_{#1}}

\newcommand{\releaseparamset}{\mathbf{\Theta}'}
\newcommand{\releaseparamsetall}{\releaseparamset^{(*)}}
\newcommand{\releaseparamsetori}{\releaseparamset^{(0)}}

\newcommand{\releasesecretsetofabbr}[1]{\mathbf{\Theta'}_{{#1}}}

\newcommand{\releaseparametervector}{\boldsymbol{\theta'}}

\newcommand{\postprocessparamset}{\mathbf{\Theta}''}

\newcommand{\mech}{\mathcal{M}}
\newcommand{\mechquan}{\mathcal{M}_\text{QM}}
\newcommand{\mechrr}{\mathcal{M}_\text{RR}}

\newcommand{\privacyquan}{\privacynotation_{\mechquan,\secretnotation}}
\newcommand{\privacyrr}{\privacynotation_{\mechrr,\secretnotation}}

\newcommand{\distortionquan}{\distortionnotation_{\mechquan}}
\newcommand{\distortionrr}{\distortionnotation_{\mechrr}}

\newcommand{\rrratio}{r}

\newcommand{\midpoint}[1]{\text{R}\bra{#1}}

\newcommand{\secretnotation}{\mathfrak{g}}
\newcommand{\secretof}[1]{\secretnotation\bra{#1}}

\newcommand{\secretofparam}{\secretof{\rvparamnotation}}

\newcommand{\distanceof}[2]{D_{\text{TV}}\bra{#1\|#2}}

\newcommand{\distance}{D_{\text{TV}}}

\newcommand{\probnotation}{\mathbb{P}}
\newcommand{\probof}[1]{\probnotation\bra{#1}}

\newcommand{\support}[1]{\text{Supp}\bra{#1}}

\newcommand{\mechnum}{m}
\newcommand{\indexi}{i}

\newcommand{\intervallen}{I}

\newcommand{\adaptivemech}[1]{\probnotation_{\Theta'^{(#1)} | \Theta, \brc{\mech_k, \Theta'^{(k)}}_{k \in [#1-1]}}}

\newcommand{\samplenumori}{n}
\newcommand{\samplenum}{\tau}
\newcommand{\categorynum}{d}

\newcommand{\categoryactual}{d^*}
\newcommand{\categoryleast}{\hat{d}^*_0}
\newcommand{\categorynoise}{\hat{d}^*_1}
\newcommand{\categoryestimate}{\hat{d}^*}

\newcommand{\combinationset}{\mathbf{\Gamma}}

\newcommand{\combinationsetall}{\mathbf{\Gamma}^*}
\newcommand{\combinationsetsingle}{\tilde{\mathbf{\Gamma}}}
\newcommand{\combinationsetestimate}{\hat{\mathbf{\Gamma}}^*}
\newcommand{\combinationsetpre}{\hat{\mathbf{\Gamma}}^*_0}
\newcommand{\combinationsetnoise}{\hat{\mathbf{\Gamma}}^*_1}

\newcommand{\dataset}{\mathbf{D}}

\newcommand{\parameterdistribution}{\theta}

\newcommand{\mass}[1]{m_{#1}}

\newcommand{\secretnum}{s}

\newcommand{\boundgaprr}{\Delta_{\mechrr, \samplenum}}
\newcommand{\boundgapquan}{\Delta_{\mechquan, \samplenum}}

\newcommand{\mincost}{\mathcal{C}_{\mech,\secretnotation}}

\newcommand{\smlinf}{\privacynotation^{\text{inf}}_{\mech, \secretnotation}}
\newcommand{\smlworse}{\tilde{\privacynotation}_{\mech, \secretnotation}}
\newcommand{\smlworseinf}{\tilde{\privacynotation}^{\text{inf}}_{\mech, \secretnotation}}
\newcommand{\ipmetric}{\mathcal{L}_{DP}}

\newcommand{\ip}{\mu}

\newcommand{\node}[1]{N_{#1}}

\newcommand{\greedy}{MaxL}

\declaretheorem[name=Proposition,numberwithin=section]{proposition}
\declaretheorem[name=Lemma,numberwithin=section]{lemma}

\crefformat{section}{\S#2#1#3}

\AtBeginEnvironment{appendices}{\crefalias{section}{appendix}}
\AtBeginEnvironment{appendices}{\crefalias{subsection}{appendix}}
\AtBeginEnvironment{appendices}{\crefalias{subsubsection}{appendix}}

\begin{document}

\title{Statistic Maximal Leakage
}

\author{%
\IEEEauthorblockN{Shuaiqi Wang \IEEEauthorrefmark{1},
Zinan Lin \IEEEauthorrefmark{2},
Giulia Fanti \IEEEauthorrefmark{1}\\
}
\IEEEauthorblockA{\IEEEauthorrefmark{1}%
Carnegie Mellon University, Pittsburgh, PA, USA,
\{shuaiqiw, gfanti\}@andrew.cmu.edu}\\
\IEEEauthorblockA{\IEEEauthorrefmark{2}%
Microsoft Research, 
Redmond, WA, USA,
zinanlin@microsoft.com}
\thanks{This paper was presented in part at ISIT 2024 \cite{wang2024statistic}.}
}

\date{}

\maketitle
\pagenumbering{arabic} 

\begin{abstract}
    We introduce a privacy measure called \emph{\privname} that quantifies how much a privacy mechanism leaks about a specific secret, relative to the adversary's prior information about that secret. 
    \Privname is an extension of the well-known  \emph{maximal leakage}. Unlike maximal leakage, 
    which protects an arbitrary, unknown secret, 
    \privname protects a single, known secret.
    We show that \privname satisfies composition and post-processing properties. Additionally, we show how to efficiently compute it in the special case of deterministic data release mechanisms. We analyze two important mechanisms under \privname: the quantization mechanism and randomized response. We show theoretically and empirically that the quantization mechanism achieves better privacy-utility tradeoffs in the settings we study. 
\end{abstract}


\section{Introduction}
\label{sec:intro}
Information disclosure while preventing privacy leakage is a central problem in the privacy and information theory literature. 
That is, how can we release a (realization of a) random variable without leaking correlated secret information?
Since the seminal work of Yamamoto \cite{yamamoto1983source}, many papers have studied this problem~\cite{smith2009foundations,alvim2012measuring,wang2019privacy,makhdoumi2014information,zamani2022bounds,braun2009quantitative,issa2019operational,du2017principal,shkel2020secrecy,zamani2023privacy,huang2024efficient,shkel2021compression,zamani2023cache,zamani2023private}. %
{We model a \camerareadydelete{a (possibly multi-dimensional) random variable $X$}\cameraready{data holder with data drawn from a distribution parameterized by a random variable  $\Theta$.}} %
{The data holder also knows a secret, represented by a discrete random variable $G$, which is computed as a function of the underlying data distribution; specifically, $G=\secretnotation(\Theta)$. }
The data holder's goal is to release \camerareadydelete{$Y$}\cameraready{$\Theta'$}, %
a perturbed version of \camerareadydelete{$X$}\cameraready{$\Theta$}, while optimizing the tradeoff between the \emph{leakage} about $G$ and the \emph{utility} of \camerareadydelete{$Y$}\cameraready{$\Theta'$}. %

We seek a privacy measure that satisfies three properties: 
\begin{enumerate}
	\item  \emph{Prior-independence:}
	The measure should not depend on any party's prior over the input data\camerareadydelete{$X$}\cameraready{, up to determining the support of the distribution}.
	This arises because priors may be difficult to obtain in practice; measures that require such knowledge may admit mechanisms that are fragile to prior mis-specification \cite{lin2023summary,golden2019consequences,richardson1997some,kleijn2006misspecification}. %
	\item \emph{Secret-awareness:}
	We assume the secret function $\secretnotation$ is known, which describes how to obtain the secret $G$ from \cameraready{the input data}.
	We want a measure that depends explicitly on $\secretnotation$.
	This is primarily for efficiency reasons; by utilizing known information, we may be able to add less noise or perturbation to our data than a measure that is secret-agnostic.
	\item \emph{Composition and post-processing:}
	Informally, composition provides a bound on how a privacy measure degrades when one or more mechanisms are applied sequentially to the same data. For example, in differential privacy (DP), the privacy parameter $\epsilon$ degrades additively when a mechanism is applied multiple times to the same dataset \cite{dwork2014algorithmic}.
	Post-processing  guarantees that if one applies an arbitrary (possibly random) function to the output of a privacy mechanism, as long as the function does not depend on the original data, the privacy measure in question does not degrade. Composition and post-processing are  useful properties exhibited by DP, which have contributed to its widespread usage in practical settings such as machine learning pipelines \cite{abadi2016deep}.
\end{enumerate}

Today there exist measures that satisfy all three of these properties. %
To the best of our knowledge, these measures are all \cameraready{inspired by} differential privacy. Examples include attribute privacy \cite{zhang2022attribute}, distribution privacy \cite{kawamoto2019local}, and distribution inference \cite{suri2021formalizing}. 
Due to their assumptions and formulation, they require large amounts of noise in practice \cite{lin2023summary,wangguarding}.
{We discuss this tradeoff further in \cref{sec:related_work}, \cref{section:formulation}, and \cref{sec:empirical}.}

In this work, we study an information-theoretic privacy measure that satisfies the above three properties. 
For a special class of \camerareadydelete{random variables $X$}\cameraready{data} with distributions drawn from a parametric family, and parameter vectors drawn from a finite set, we propose a privacy measure inspired by maximal leakage \cite{braun2009quantitative,issa2019operational}, which we call 
\textit{\privname}.
Our contributions in this work are as follows: 
\begin{itemize}
    \item \emph{Formulation and properties.} We formulate the \privname  privacy measure, and show that it trivially satisfies the first two properties. We also show that it also satisfies composition and post-processing. 
Note that although composition and post-processing have been previously proved for an 
extension of maximal leakage called \emph{pointwise maximal leakage} \cite{saeidian2022pointwise},
their result and proofs do not imply that \privname also satisfies these properties.
    \item \emph{Computation.} We next study how to compute \privname. 
{We show that in general, computing \privname is NP-hard. }
However, for the class of deterministic data release mechanisms, we show that \privname can be computed in polynomial time by solving a maximum flow problem over a graph whose construction depends on the mechanism parameters.
    \item \emph{Mechanisms.} We next analyze two natural mechanisms that have been studied widely in the privacy literature: the quantization mechanism \cite{lin2023summary,farokhi2019development} and randomized response \cite{chaudhuri2020randomized,kairouz2016extremal}.
    Quantization-based mechanisms have been shown to achieve (near)-optimal privacy-utility tradeoffs in several privacy frameworks, including summary statistic privacy \cite{lin2023summary} and non-randomized privacy \cite{farokhi2019development}. %
    Randomized response is a widely-adopted mechanism \cite{chaudhuri2020randomized} that achieves optimal  privacy-utility tradeoffs for differentially-private data collection \cite{wang2016using}. 
    We show that both mechanisms satisfy non-trivial \privname guarantees. 
    Further, under general tabular datasets, we analyze their privacy-utility tradeoffs; for this analysis, we define utility as a worst-case total variation distance between the original and the released data distributions. %
    Our results show that for a tabular data release problem that aims to hide one of the marginal values, the quantization mechanism achieves a better privacy-utility tradeoff than randomized response for practical privacy parameters. 
    \item \emph{Empirical evaluation.} Finally, we apply the quantization mechanism %
to a real-world tabular dataset. We show that when instantiated with an appropriate \privname parameter, the quantization mechanism effectively protects the secret while still ensuring high utility of the released data. {Moreover, we illustrate empirically that the quantization mechanism we propose for satisfying \privname incurs lower utility cost than mechanisms designed for other privacy frameworks: attribute privacy \cite{zhang2022attribute} and maximal leakage \cite{issa2019operational}.}

\end{itemize}

\section{Related Work}
\label{sec:related_work}

There have been many privacy measures proposed for capturing leakage of sensitive information during information disclosure. We divide them into \emph{information-theoretic measures} and \emph{indistinguishability-based measures}.

\subsection{Information-Theoretic Measures}
\paragraph{Maximal leakage-based measures}
While there are various of metrics designed to measure the privacy leakage, \emph{maximal leakage} is the one that is the most related to our designed statistic maximal leakage. Maximal leakage, introduced by Issa et al. \cite{issa2019operational}, is an information-theoretic measure designed to capture the worst-case leakage of sensitive information in a data release mechanism. %
It is defined using the Markov chain $G-\Theta-\Theta'-\hat{G}$, where %
$\hat{G}$ represents the adversary's guess of secret $G$. 
The measure is expressed as:

\[
\maxl = \sup_{G-\Theta-\Theta'-\hat{G}} \log \frac{\probof{G=\hat{G}}}{ \max_{g}\mathbb{P}_G(g) },
\]
where the $\sup$ is over $G$ and $\hat{G}$, i.e., consider the worst-case secret and the strongest attack strategy.
This formulation evaluates the ratio of probabilities of correctly guessing the secret with and without observing the released data, thereby quantifying the increase in guessing probability facilitated by the data release.

Several variants and extensions of maximal leakage have been proposed to address different privacy scenarios via gain functions \cite{liao2019tunable,saeidian2022pointwise,gilani2023alpha,kurri2022operational}, which penalize different values of $\probof{G=\hat{G}}$ differentaly. %
This framework can adapt to different threat models by modifying the evaluation of adversarial success, thus offering a tunable balance between privacy and utility.
For example, maximal \((\alpha, \beta)\)-leakage \cite{gilani2023alpha} generalizes maximal leakage by introducing parameters \(\alpha\) and \(\beta\), which adjust the sensitivity of leakage to different probabilities of adversarial success. By tuning \(\alpha\) and \(\beta\), the metric can be converted to several known leakage measures, such as maximal leakage and local differential privacy.
Pointwise maximal leakage, discussed by Saeidian et al. \cite{saeidian2022pointwise}, examines the leakage given certain released data rather than averaging over the entire distribution, preventing disproportionate exposure of sensitive information.
Binary maximal leakage \cite{cung2024binary} focuses on the scenario where only binary secret functions are assumed to be of interest to the adversary, and demonstrates that restricting the set of secrets to be hidden allows for a better utility. 

Maximal leakage and its variants \cite{issa2019operational,saeidian2022pointwise,liao2019tunable,gilani2023alpha,kurri2022operational,cung2024binary}
assume the exact secret is unknown \emph{a priori}, and hence are worst-case metrics over all secrets (or secrets 
within certain type). Most of these works also require knowledge of a prior distribution of the data, expect for maximal $\bra{\alpha,\beta}$-leakage \cite{gilani2023alpha} and binary maximal leakage \cite{cung2024binary}. Besides, although maximal leakage and some of its variant \cite{gilani2023alpha,cung2024binary} analyze the composition property, it holds only when successive outputs from the mechanism(s) are conditionally independent, conditioned on the input data. The only exception is pointwise maximal leakage \cite{saeidian2022pointwise}, which has the property of adaptive composition in additive form.

Besides maximal leakage, several metrics are proposed to quantify the information leakage from $\Theta$ to $\Theta'$ %
based on mutual information \cite{makhdoumi2014information,calmon2015fundamental,rassouli2021perfect,zamani2022bounds,biswas2022privic,oya2017back,clark2002quantitative,clark2007static,malacaria2007assessing}, $f$-divergences \cite{wang2019privacy,rassouli2019optimal}, or min-entropy \cite{smith2009foundations,alvim2012measuring,alvim2014additive,asoodeh2017privacy,asoodeh2018estimation,bordenabe2016correlated,jurado2023analyzing,jin2019security,alvim2023novel}, which are detailed as follows. 

\paragraph{Mutual information-based measures} 
Privacy Funnel \cite{makhdoumi2014information,yamamoto1983source} measures the privacy leakage by the mutual information between the secret $G$ and released data. It minimizes the leakage of secret while ensuring the released data retains a certain level of utility from the original data, which is captured by the mutual information between $\Theta$ and $\Theta'$. 
Similarly, rate-distortion theory, originally from source coding \cite{shannon1959coding}, models privacy problems by minimizing mutual information between $\Theta$ and $\Theta'$ subject to a distortion constraint, trading off privacy for utility \cite{biswas2022privic,oya2017back}. %
However, in addition to requiring knowledge of a prior distribution of input data, as discussed in \cite{issa2019operational,lin2023summary}, mutual information does not align with practical notions of privacy and utility, as it does not account for the probability of correct guesses by an adversary.

\paragraph{Min-entropy-based measures} Several works measure information leakage with min-entropy \cite{alvim2012measuring,alvim2014additive,asoodeh2017privacy,asoodeh2018estimation}, the probability of guessing a secret correctly, which provides a more direct measure of privacy risk. For example, summary statistic privacy \cite{lin2023summary,wang2024guarding} quantifies the privacy loss by assessing the worst-case probability the adversary successfully guesses the secret within a tolerance range.  However, those measures require the knowledge of the prior distribution of $\Theta$. Although \cite{braun2009quantitative} extends the measure to be prior-independent by considering the probability of guessing the secret under a worst-case prior over $\Theta$, it considers on the privacy loss of the entire random variable $\Theta$, rather than a specific secret; moreover, no composition and post-processing properties are provided.

\paragraph{Quantitative Information Flow} Quantitative Information Flow (QIF) measures how much information about a secret is leaked through observable outputs. Originating in the works of Denning \cite{denning1982cryptography} and Gray \cite{gray1992toward}, QIF has evolved to encompass a range of privacy metrics. Early QIF frameworks used mutual information to quantify leakage \cite{clark2002quantitative,clark2007static,malacaria2007assessing}, while Smith \cite{smith2009foundations} adopted min-entropy as the metric after recognizing the shortcomings of mutual information as a privacy metric. Further developments introduced generalized metrics like $g$-leakage, which incorporate gain functions to model scenarios involving partial or multiple guesses \cite{alvim2012measuring}. This framework has been applied to various domains, such as local DP and cyber-attack defenses \cite{jurado2023analyzing, alvim2023novel, jin2019security}. However, these studies treat the whole data as the secret and few analyze post-processing and composition properties.

\paragraph{Non-stochastic information theoretic
methods} 
Non-stochastic information theoretic
methods have also been adopted to assess privacy leakage \cite{bhaskar2011noiseless,farokhi2019development,farokhi2021noiseless,farokhi2021non}, and most works in this area focus on the exploration of noiseless privacy-preserving mechanisms, e.g., quantization-based mechanisms. For example, Farokhi et al. \cite{farokhi2019development, farokhi2021noiseless} evaluate utility as the worst-case deviation between input and output data and measure the privacy by {maximin information} \cite{nair2013nonstochastic}, which is defined through a unique taxicab partitioning of feasible input-output pairs. \cite{farokhi2019development} identifies quantization as an optimal privacy mechanism within deterministic piecewise differentiable policies, i.e., policies quantizing the input dataset into bins and the output is differentiable for each bin, that maximize privacy under a utility constraint. However, similar to Quantitative Information Flow, these works consider the whole data as the secret and do not prove post-processing or composition properties.

\subsection{Indistinguishability-Based Measures}
Orthogonal to information-theoretic approaches, several studies measure the privacy leakage of statistical properties of the data based on a measure of indistinguishability over candidate inputs \cite{zhang2022attribute,kawamoto2019local,suri2021formalizing,suridissecting}; these techniques draw inspiration from \emph{differential privacy} (DP) \cite{dwork2014algorithmic}, one of the most widely adopted privacy metrics. Roughly, differential privacy quantifies the privacy of a data release mechanism by measuring the influence of the participation of an individual record (or a group of records) on the final output \cite{dwork2014algorithmic}. However, as discussed in \cite{lin2023summary,wang2024guarding}, DP is designed to measure individual-level privacy and cannot directly quantify information leakage of statistical secrets. As an example, if we aim to protect the mean of a dataset of scalar values and apply a typical local DP mechanism \cite{bebensee2019local} that adds zero-mean Gaussian noise to each record, this prevents the attacker from inferring the inclusion of individual records in the dataset but still allows an unbiased estimate of the dataset's mean. As the number of records increases, the mean of released dataset converges to the original mean value \cite{lin2023summary}. 
Simple variants of this mechanism that add different noise scales to different features are also shown to fail in \cite{lin2023summary}.

Distribution privacy \cite{kawamoto2019local} and distribution inference \cite{suri2021formalizing,suridissecting} tackle this challenge by ensuring that for any two input distributions with different parameters, the output distributions remain indistinguishable (up to multiplicative and additive factors), thus protecting sensitive distributional properties. For instance, attribute privacy \cite{zhang2022attribute} adopts an indistinguishability definition based on the pufferfish framework, while restricting the class of distributions that should be indistinguishable; it is particularly useful for scenarios where only a statistic of the dataset is shared. Both attribute and distribution privacy mechanisms can be very noisy, as they must obfuscate between entire distributions. 
This is challenging because two random variables may have similar secret values (e.g., means), while having arbitrarily different underlying distributions; hence, the amount of noise required to make these distributions indistinguishable can be substantial. 
We show an example of this privacy-utility tradeoff in \cref{sec:empirical}.

\section{Notation and Problem Setting}
\label{section:formulation}

{We generally use uppercase Greek and Roman letters (e.g., $\Theta, G$) to denote random variables, and lowercase Greek letters (e.g., $\theta$) to denote their realization. Sets are denoted by \textbf{boldface} uppercase Greek and Roman letters.} 
\cameraready{A data holder has data drawn from a distribution parameterized by a parameter vector $\theta$.}
\camerareadydelete{$X$'s}\cameraready{The} parameter $\theta$ is itself a realization of a random variable $\Theta \in \paramset$ belonging to a finite set
$\paramset$.\footnote{\cameraready{
We extend the formulation to infinite, continuous parameter sets in \cref{sec:conclusion}.}}
$\paramdistribution$ represents the prior distribution of the parameter random variable $\Theta$, 
and can equivalently be viewed as the prior over the input data\camerareadydelete{$X_{\theta}$}. We use $\distributionof{}$ to denote distribution measures.

The data holder aims to protect a secret $g=\secretnotation (\theta)$ where $\secretnotation$ is a function that is fixed and known.
$g$ is a realization of random variable $G \in \secretvalueset \triangleq \brc{\secretrv_1, \secretrv_2,\cdots, \secretrv_{\secretnum}}$ (i.e., the secret can take $s$ values).
We use $\secretset$ to represent the set of original parameters whose secret values are $\secretrv$, i.e., $\secretset = \brc{\theta\in \mathbf{\Theta} | \secretofparam = \secretrv}$. %
In \cref{sec:mech_analysis}, we use $\secretsetofabbr{i}$ to represent $\mathbf{\Theta}_{\secretrv_i}, \forall i\in \brb{\secretnum}$, for convenience.

The data holder releases \camerareadydelete{$X_{\theta}$}\cameraready{data} via a data release mechanism $\mech=\mathbb{P}_{\Theta'|\Theta}$, which maps input parameter $\theta$ to \camerareadydelete{an}\cameraready{a (possibly random)} output parameter $\cameraready{\Theta'} \in \releaseparamset$ (in general, $\paramset \neq \releaseparamset$).
We use $\mech\bra{\theta}$ to denote the \cameraready{random} distribution parameter \camerareadydelete{$\theta'$}\cameraready{$\Theta'$} output by mechanism $\mech$ with input $\theta$.
\camerareadydelete{The data holder releases variable $Y_{\theta'}$ parameterized by $\theta'$.} 
Given \camerareadydelete{$Y_{\theta'}$}{the realization of $\Theta'$, which we denote by $\theta'$}, the attacker outputs a (possibly random) estimate of the secret, $\hat{G}$. \cameraready{We assume the attacker knows the prior distribution of the data and the data release mechanism $\mech$, and has infinite computational power.} The overall data sharing and attacker guessing process can be formulated as a Markov chain \camerareadydelete{$G-X_{\theta}-Y_{\theta'}-\hat{G}$}\cameraready{$G-\Theta-\Theta'-\hat{G}$}.

\subsection{Utility}
To analyze the privacy-utility tradeoff for a mechanism $\mech$, we 
define our utility measure as the distortion of $\mech$ as the expected total variation (TV) 
distance between the original and released \camerareadydelete{datasets}\cameraready{data distributions, represented by $X_{\Theta}$ and $Y_{\Theta'}$ respectively}, %
under the worst-case prior: %
\begin{align}
    \smldistortion = \sup_{\paramdistribution} \cameraready{\mathbb{E}_{\Theta, \Theta'=\mech\bra{\Theta}}[\distanceof{\distributionof{X_{\Theta}}} {\distributionof{Y_{\Theta'}}}]},
    \label{eq:utility}
\end{align}
where $\distance$ is the total \cameraready{variation} distance.

{Note that \cref{eq:utility} is computed for a worst-case prior, but it is average-case over the realizations of the data and the output of the data release mechanism. We have chosen to model utility in this way because we want the utility measure to apply to a data release mechanism, regardless of input. Hence, the measure should not depend on  $\paramdistribution$. Two natural options are to either consider the worst-case distribution $\paramdistribution$, or the average-case. As we do not model a prior over $\paramdistribution$, we use the former.
We subsequently take an expectation over $\Theta$ and $\Theta'$, as is common in many works studying the utility of privacy mechanisms, both in information theory \cite{oya2017back,asoodeh2017privacy} and  in differential privacy \cite{alvim2012differential,murakami2019utility}.
}

We then provide \cref{lemma:distortion}, which allows us to analyze the distortion of $\mech$ as the expected
TV distance between the original and released datasets under a worst-case input.
\begin{restatable}{lemma}{distortionmeasure}
\label{lemma:distortion}
The distortion measure $\smldistortion$ can be rewritten as 
$
\smldistortion
=\sup_{\theta} \mathbb{E}_{\Theta'=\mech\bra{\theta}}[\distanceof{\privatedistribution} {\releasedistribution}].
$
\end{restatable}
(Proof in \cref{sec:proof_distortion})
Since our utility measure considers a worst-case prior distribution, the distortion of mechanisms proposed for attribute privacy \cite{zhang2022attribute}, distribution privacy \cite{kawamoto2019local}, or distribution inference \cite{suri2021formalizing} can reach the trivial upper bound $1$ 
on distortion,\footnote{We illustrate the case for attribute privacy in \cref{sec:empirical}.} i.e., the worst possible utility. Our goal is to understand whether there exist data release mechanisms that achieve meaningful utility while also satisfying a privacy guarantee with the desired properties from \cref{sec:intro}. %

\section{\PrivName{} }
\label{sec:formulation}

{We next present our proposed privacy measure, discuss its properties, and explain how to compute it.}
\Privname (SML) measures the largest increase an adversary can gain in their guess of $g$; it is a property of a data release mechanism $\mathcal M$ and a secret mapping $\secretnotation$. We define it as follows:
\begin{align}
    \sml = \sup_{\paramdistribution, \mathbb{P}_{\hat{G}|\Theta'}}\log \frac{\probof{\hat{G}=G}}{\sup_{\secretrv\in \secretvalueset} \mathbb{P}_G\bra{\secretrv}},
    \label{eq:privmetric}
\end{align}
where the supremum is taken over all prior distributions $\paramdistribution$ over the distribution parameter $\theta$ and attack strategies $\mathbb{P}_{\hat{G}|\Theta'}$. 
The probability in the numerator is over the attacker's randomized estimator, the mechanism, and the secret.
Note that 
the secret function $\secretnotation$ and the data release mechanism $\mech$ are fixed and assumed to be known in this optimization.

SML bears some similarities with maximal leakage \cite{issa2019operational} and worst-case min-entropy leakage \cite{braun2009quantitative}, though they do not simultaneously satisfy all three desirable properties we propose. Adopting our notation, maximal leakage $\maxl$ is defined as follows:
\begin{align*}
\maxl = \sup_{\mathbb{P}_{G|\Theta}, \mathbb{P}_{\hat{G}|\Theta'}} \log\frac{\probof{\hat{G}=G}}{\sup_{\secretrv\in \secretvalueset} \mathbb{P}_G\bra{\secretrv}}.
\end{align*}
\normalsize
Maximal leakage requires a known prior distribution over input data {$\mathbb{P}_{\Theta}$} and assumes the secret function $\secretofparam$ is unknown. 
{Hence, it optimizes over all $\mathbb{P}_{G|\Theta}$.}
Worst-case min-entropy leakage $\minentropy$ treats the entire input data distribution as the secret to protect; for the Markov chain \camerareadydelete{$X_{\theta}-Y_{\theta'}-\hat{X}_{\theta}$}\cameraready{$\Theta-\Theta'-\hat{\Theta}$}, it is defined as
\begin{align*}
\minentropy = \sup_{\paramdistribution, \mathbb{P}_{\camerareadydelete{\hat{X}_{\theta}}\cameraready{\hat{\Theta}}|\Theta'}} \log\frac{\probof{\camerareadydelete{\hat{X}_{\theta}={X}_{\theta}}\cameraready{\hat{\Theta}=\Theta}}}{\sup_{\theta\in \mathbf{\Theta}} \mathbb{P}_\Theta\bra{\theta}}.
\end{align*}
\normalsize
Under a fixed prior, maximal leakage can have a smaller value than SML since SML %
considers the worst-case prior.
However, one can construct $\mech$ and $\secretnotation$ for which maximal leakage, with a worst-case prior, achieves its largest possible value (i.e., $\min\brc{\log |\Theta'|, \log |\Theta|}$)  
while SML is 0.
The following property shows that SML is upper- and lower-bounded by both maximal leakage with a worst-case prior and worst-case min-entropy leakage, by up to an additive factor that depends on the secret and the parametric family. (Proof in \cref{sec:proof_metric_compare}.)

\begin{restatable}[Relation to Maximal Leakage and Min-Entropy Leakage]{property}{relation}
\label{property:metric_compare}
\Privname{} $\sml$ satifies
\begin{align*}
\minentropy- \sup_{\secretrv\in \secretvalueset}\log|\mathbf{\Theta}_{\secretrv}|~ \leq ~&\sml ~\leq ~\minentropy,\\
\sup_{\paramdistribution}\maxl-\sup_{\secretrv\in \secretvalueset}\log|\mathbf{\Theta}_{\secretrv}| ~\leq~ &\sml ~\leq~ \sup_{\paramdistribution}\maxl.
\end{align*}
\normalsize
\end{restatable}

\begin{figure}[tbp]
\centering
\includegraphics[width=0.8\linewidth]{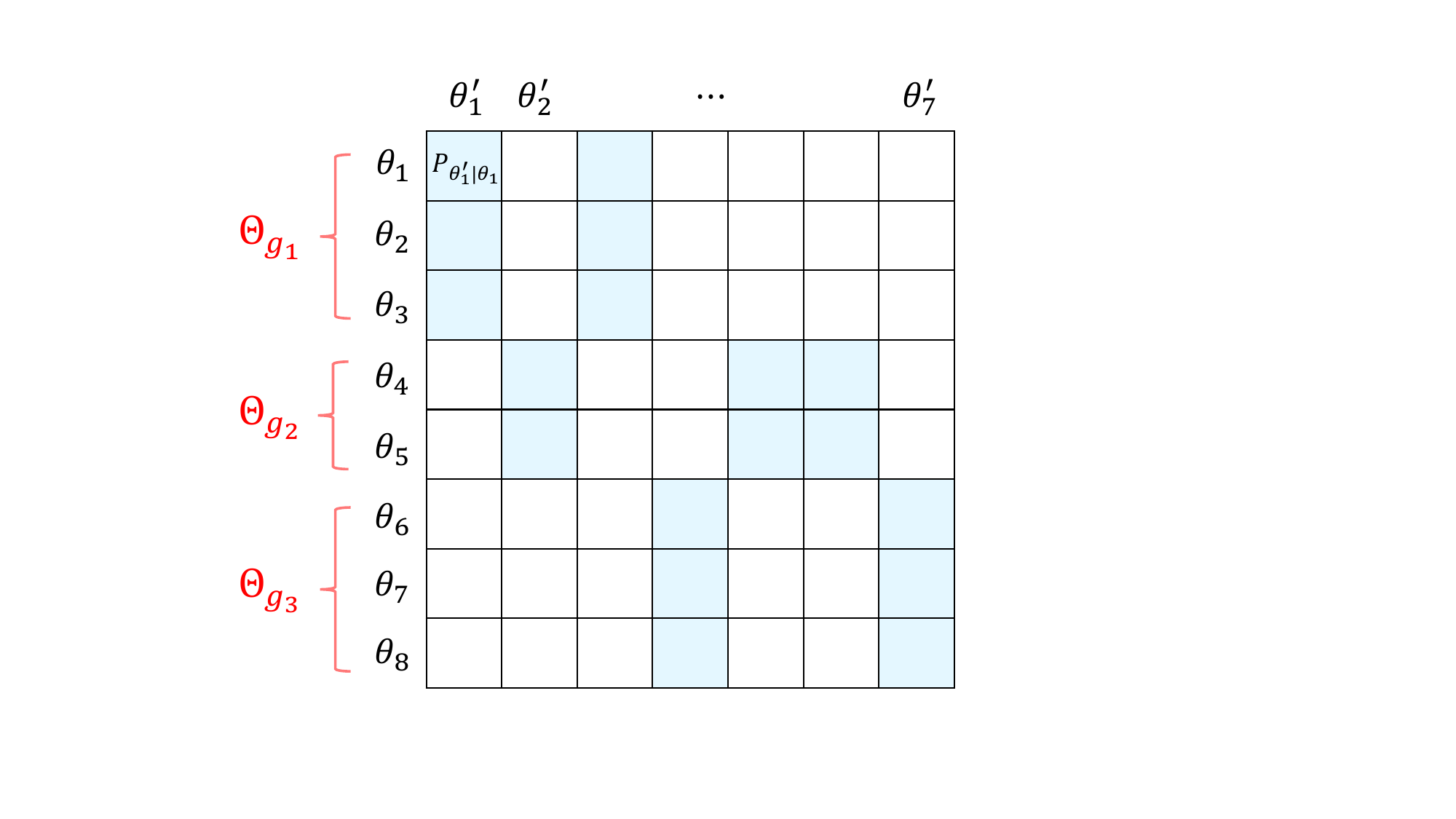}
\caption{
Given a mechanism $\mech=\mathbb{P}_{\Theta'|\Theta}$, the left subfigure shows a policy matrix. For each column $j$, the red outlined region indicates rows of parameters with secret $\secretrv$ maximizing $\mathbb{P}_{\Theta'|\Theta}\bra{\theta'_j|\theta_\secretrv}$. The blue cell lies in the row of $\theta_\secretrv$. When the mechanism $\mech$ is deterministic, SML calculation can be converted to a min-cost flow problem (right). The constructed directed graph contains three columns of nodes (representing $G, \Theta, \Theta'$ respectively) between the source and sink nodes. The capacity of all edges are $1$, and only the edges between nodes in $\Theta$ and $\Theta'$ columns have non-zero cost ($-\mathbb{P}_{\Theta'|\Theta}\bra{\theta'_k|\theta_j}$ between $\theta_j$ and $\theta'_k$). Edges are annotated as: Capacity (Cost).
}
\label{fig:calculation}
\end{figure}

\subsection{Computation of \PrivName}
{We next discuss how to compute \privname.}
Computing SML is more convenient under an alternative form, which shows that we only need to search over a restricted class of priors that assign nonzero probability mass to at most one parameter $\theta \in \secretset$, for each secret $g\in \secretvalueset$.
Under such a prior, for a fixed $g \in\secretvalueset$, there is at most one  $\theta \in \paramset$ such that $\mathbb{P}_{\Theta|G}\bra{\theta|\secretrv}>0$; we use $\theta_\secretrv$ to denote this value.%
\begin{restatable}{proposition}{smlexpression}
\label{prop:sml_calculation}
\Privname{} satisfies
\begin{align*}
\sml = \sup_{\mathbb{P}_{\Theta|G}\in\brc{0,1}}\log \sum_{\theta'\in\releaseparamset} \sup_{\secretrv\in \secretvalueset} \mathbb{P}_{\Theta'|\Theta}\bra{\theta'|\theta_\secretrv}.
\end{align*}
\end{restatable}

\noindent(Proof in \cref{section:proof_sml_calculation}.)
Based on \cref{prop:sml_calculation}, we explain how to compute SML. For concreteness, \cref{fig:calculation} (left) illustrates an example with $\paramset=\brc{\theta_1, \theta_2, \theta_3}, \releaseparamset=\brc{\theta'_1, \theta'_2}$, and $\mathbf{\Theta}_{\secretrv_1}=\brc{\theta_1, \theta_2}, \mathbf{\Theta}_{\secretrv_2}=\brc{\theta_3}$. 
Given a mechanism $\mech=\mathbb{P}_{\Theta'|\Theta}$, we can construct a policy matrix where the value in the $i$-th row and $j$-th column is $\mathbb{P}_{\Theta'|\Theta}\bra{\theta'_j|\theta_i}$. First, fix a prior $\paramdistribution$ such that $\mathbb{P}_{\Theta|G}\in\brc{0,1}$. 
\cref{fig:calculation} illustrates a case where the prior satisfies $\theta_{\secretrv_1}=\theta_1$ and $\theta_{\secretrv_2}=\theta_3$ ($\theta_\secretrv$ is defined above \cref{prop:sml_calculation}%
).
Next, fix a column $\theta'_j$ in the policy matrix. 
We can now find a secret value $\tilde \secretrv\in\secretvalueset$ that maximizes $\mathbb{P}_{\Theta'|\Theta}\bra{\theta'_j|\theta_{\tilde \secretrv}}$---i.e., $\tilde \secretrv$ is  the maximum likelihood secret for an observed output $\theta'_j$. 
For each column, the red outline denotes the input parameters in $\boldsymbol \Theta_{\tilde g}$.
Our example mechanism satisfies $\arg\sup_{\tilde \secretrv}\mathbb{P}_{\Theta'|\Theta}\bra{\theta'_1|\theta_{\tilde \secretrv}}=\secretrv_1, \arg\sup_{\tilde \secretrv}\mathbb{P}_{\Theta'|\Theta}\bra{\theta'_2|\theta_{\tilde \secretrv}}=\secretrv_2$. 
For each column, the blue square is the intersection of the red region with the row of $\theta_{\tilde \secretrv}$. 
We finally sum the likelihoods of all the blue squares. 
Our goal is to find the worst-case prior and calculate the maximum value of $\log \sum_{\theta'\in\releaseparamset} \sup_{\tilde \secretrv\in \secretvalueset} \mathbb{P}_{\Theta'|\Theta}\bra{\theta'|\theta_{\tilde \secretrv}}$.
Worst-case, this can be done in time exponential in the number of input parameters $|\paramset|$ by enumerating all feasible $P_{\Theta|G}$.

\subsection{Computation for Deterministic Mechanisms: Min-Cost Flow}
When the mechanism $\mech$ is \emph{deterministic}, i.e., $\mathbb{P}_{\Theta'|\Theta}\in\brc{0,1}$, SML calculation process can be converted to a min-cost flow problem \cite{ford2015flows}. Given a directed graph where each edge is assigned a capacity and a cost, the min-cost flow problem aims to design a network flow satisfying the capacity constraint of each edge, while achieving the minimum cost. %
The final cost of the min-cost flow has a one-to-one correspondence to the SML of the underlying problem. 

To construct the network, we start with a source and a sink node, and create three columns of nodes between them. 
The first $G$-column contains all potential secret values ($\secretrv_1$, $\secretrv_2$ in \cref{fig:calculation}). The capacity of the edge between the source and each node in the $G$-column is $1$ and the cost is $0$. %
The second $\Theta$-column contains all possible input parameter values ($\theta_1$, $\theta_2$, $\theta_3$ in \cref{fig:calculation}). 
There is an edge between node ${\secretrv_i}$ and ${\theta_j}$ iff $\theta_j \in \mathbf{\Theta}_{\secretrv_i}$. The capacity of each edge is $1$ and the cost is $0$. 
The third $\Theta'$-column contains all possible released parameter values ($\theta'_1$, $\theta'_2$ in \cref{fig:calculation}).
This column is fully connected with the second column. The capacity of the edge between $\theta_j$ and $\theta'_k$ is $1$ and the cost is $-\mathbb{P}_{\Theta'|\Theta}\bra{\theta'_k|\theta_j}$. All nodes in $\Theta'$-column are also connected with the sink node with edge capacity $1$ and cost $0$. Concretely, the network construction steps are detailed in \cref{alg:network_construct}.

\begin{algorithm}[htpb]
    \LinesNumbered
	\BlankLine
	\SetKwInOut{Input}{Input}
\caption{Network construction under deterministic mechanism.}
\label{alg:network_construct}
\Input{parameter value set $\paramset$, secret function $\secretnotation$, secret value set $\secretvalueset$, released parameter value set $\releaseparamset$, data release mechanism $\mech=\mathbb{P}_{\Theta'|\Theta}$}
	\BlankLine
Construct a source node $\node{\text{src}}$ and a sink node $\node{\text{sink}}$\;
\textbf{for} each $\secretrv\in\secretvalueset$: construct a node $\node{\secretrv}$ connected 
to $\node{\text{src}}$ with edge capacity $1$ and cost $0$\;
\textbf{for} each $\secretrv\in\secretvalueset$ and $\theta\in \secretset$: construct a node $\node{\theta}$ connected to $\node{\secretrv}$ with edge capacity $1$ and cost $0$\;
\textbf{for} each $\theta'\in\releaseparamset$ and $\theta\in \paramset$: construct a node $\node{\theta'}$ connected to $\node{\theta}$ with edge capacity $1$ and cost $-\mathbb{P}_{\Theta'|\Theta}\bra{\theta'|\theta}$, and connected to $\node{\text{sink}}$ with edge capacity $1$ and cost $0$.
\end{algorithm}

From \cref{prop:sml_calculation}, we know that 
among all the distributions we are optimizing over,
there is only one 
$\theta_\secretrv$ satisfying $\mathbb{P}_{\Theta|G}\bra{\theta_\secretrv|\secretrv}=1>0, \forall \secretrv\in\secretvalueset$. 
For any \emph{deterministic} mechanism, there is only one $\theta'\in\releaseparamset$ satisfying $\mathbb{P}_{\Theta'|\Theta}\bra{\theta'|\theta}=1>0, \forall \theta\in\paramset$. %
Finally, for each $\theta'\in\Theta'$, we can only select one $\tilde \secretrv\in \secretvalueset$ to calculate $\mathbb{P}_{\Theta'|\Theta}\bra{\theta'|\theta_{\tilde \secretrv}}$ for the SML calculation, based on \cref{prop:sml_calculation}. Hence, we set the capacity of all edges %
as $1$. It is known that there exists a min-cost network flow such that the flow of each edge is either $1$ or $0$ \cite{ford2015flows}. In that case, for all nodes $\theta\in\secretset$ in $\Theta$ column, only one can accept one unit of flow from $\secretrv$, and this node is $\theta_\secretrv$. For each 
node in the $\Theta'$ column, it can also only accept one unit of flow from $\theta_\secretrv, \forall \secretrv\in\secretvalueset$. Therefore, the min-cost flow problem under our constructed network in \cref{alg:network_construct} shares the same objective as \cref{prop:sml_calculation}. %
Importantly, the min-cost flow problem can be solved efficiently in polynomial time \cameraready{in $\brd{\paramset}\cdot\brd{\releaseparamset}$} \cite{ford2015flows}.

The following proposition shows that under the network constructed through \cref{alg:network_construct}, the final cost of the min-cost flow has a one-to-one correspondence to the SML of the underlying problem.

\begin{restatable}[SML computation, deterministic  mechanism]{proposition}{smlcompdet}
\label{prop:min-cost-flow}
Given a deterministic mechanism $\mech$ and a secret mapping $\secretnotation$, the negative log cost of the min-cost flow under the network constructed through \cref{alg:network_construct} is equal to the SML $\sml$.
\end{restatable}
(Proof in \cref{section:proof_min-cost-flow})

Intuitively, we can refer to our example, where we allocate 1 unit of flow from the source to each of $\secretrv_1$ and $\secretrv_2$. %
For $\secretrv_1$, the full flow goes either to $\theta_1$ or $\theta_2$; the selected node is dubbed $\theta_{\secretrv_1}$ under $\mathbb{P}_{\Theta'|\Theta}$.
The flows from $\theta_{\secretrv_1}$ and $\theta_3$ then go to $\theta'_1$ and $\theta'_2$, respectively. 
This is because $\theta'_1$ has $g_1$ as its ML secret, and $\theta'_2$ has $g_2$ as its ML secret under $\mathbb{P}_{\Theta'|\Theta}$.
Finally, the flows merge to the sink.
The log of the negative total cost of the flow is the SML.

\subsection{Computational hardness for general mechanisms}
Although the computation of SML is efficient when the mechanism is deterministic, the computation is NP-hard in general, although there exist efficient approximation algorithms, as shown in \cref{prop:np-hard}.

\begin{restatable}[Hardness of SML computation]{theorem}{hardness}
\label{prop:np-hard}
Given a mechanism $\mech$ and a secret mapping $\secretnotation$, the computation of SML $\sml$ is NP-hard. There exists an approximation algorithm for SML calculation with approximation ratio $1+\rho$ ($\rho>0$) and running time polynomial in $\brd{\Theta}\cdot\brd{\Theta'}$ and $1/\rho$.
\end{restatable}
(Proof in \cref{proof:np-hard})
We prove the NP-hardness of SML computation by reducing it to the 3-set cover decision problem, {for which NP-completeness is known \cite{karp2010reducibility}.} To approximately compute SML, we convert the computation to an edge cost flow problem, {for which there exist practical (poly-time) approximation algorithms \cite{krumke1999flow}. However, we do not use these approximation methods in the remainder of the paper; although our proposed mechanisms are not deterministic, they possess a special structure that enables us to determine their privacy guarantees in closed-form.}

\subsection{Properties of \PrivName{}}

We next show that SML satisfies two natural desired properties: adaptive composition and post-processing. 
Adaptive composition bounds the total leakage of releasing multiple results from one or more possibly adaptive mechanisms applied sequentially.

\begin{restatable}[Adaptive Composition]{theorem}{composition}
\label{thm:composition}
    Suppose a data holder sequentially applies $\mechnum$ mechanisms $\mech_1,\ldots, \mech_m$, where $\forall i\in [m]$, the $i$th mechanism is a function of the input data $\theta$ and all of the previous outputs, which we denote as $\theta'^{(1)}, \ldots \theta'^{(i)}$. That is, $\mech_i(\theta, \theta'^{(1)}, \ldots, \theta'^{(i-1)})=\theta'^{(i)}$. 
    Suppose $\forall i\in [m]$, mechanism  $\mech_i$ satisfies a \privname guarantee with respect to $\secretnotation$ of $\privacynotation_{\mech_\indexi,\secretnotation}$.
    Let $\boldsymbol{\mech}=\mech_1 \circ \mech_2\circ  \ldots \circ \mech_m$ denote the composition of these adaptively chosen mechanisms. 
    The SML with respect to an adversary that can see all intermediate outputs $\theta'^{(1)}, \theta'^{(2)}, \ldots, \theta'^{(m)}$ can be bounded as
    $\privacynotation_{\boldsymbol{\mech},\secretnotation} \leq \sum_{\indexi \in [\mechnum]} \privacynotation_{\mech_\indexi,\secretnotation}$.
\end{restatable}

\noindent(Proof in \cref{section:proof_composition}.) \cref{thm:composition} shows that \privname{} degrades additively when one or more possibly adaptive mechanisms are applied multiple times sequentially to the original dataset.
This additive result is similar in form to analogous composition results for other privacy metrics, including pointwise maximal leakage \cite[Thm. 13]{saeidian2022pointwise} and differential privacy \cite[Thm. III.1]{dwork2010boosting}. In particular, we note that the result for pointwise maximal leakage does not imply ours, nor vice versa. 
Pointwise maximal leakage requires knowledge of the prior distribution of the input data and assumes the secret function is unknown; the composition property for the worst-case secret function or prior does not imply the composition property for an arbitrary secret or prior.

\begin{restatable}[Post-Processing]{theorem}{post}
\label{thm:post-processing}
Let $\mech$ be a \datamechanism{} whose SML is $\sml$. Let $\widetilde{\mech}$ be an arbitrary (possibly randomized) mechanism defined by $\mathbb{P}_{\Theta''|\Theta'}$. Then the SML of $\widetilde{\mech}\circ\mech$ is $\privacynotation_{\widetilde{\mech}\circ\mech,\secretnotation}\leq\sml$.
\end{restatable}

\noindent(Proof in \cref{section:proof_post-processing}.) \cref{thm:post-processing} shows that applying an arbitrary (possibly randomized) mechanism to the output of a mechanism that satisfies \privname  will not degrade \privname{}.
\section{Mechanism Design for Tabular Data}
\label{sec:mech_analysis}

We next study two natural mechanisms for releasing \emph{tabular} data {with attributes that take values from a finite set} under SML: randomized response and the quantization mechanism.
{Both have been widely used in practice \cite{lin2023summary,farokhi2019development,chaudhuri2020randomized,kairouz2016extremal}.}
Our goal is to understand (a) if each of these satisfies a SML guarantee, and (b) if so, which one has a better privacy-utility tradeoff?

\subsection{Setup}
\label{sec:tabular_setup}
The data holder holds a tabular dataset $\dataset$ with $\samplenumori$ rows and $c$ columns, i.e., $\samplenumori$ samples with $c$ attributes for each sample. Let $\combinationset$ be the set of combinations of attributes for samples existing in the original dataset $\dataset$, where $\categorynum \triangleq \abs{\combinationset}$. 
For example, suppose our dataset has binary columns ``Above age 18?'' and ``Registered to vote in the U.S?'' 
and only includes samples with attribute values ``(Yes, Yes)'' and ``(Yes, No)''; then $d=2$ {and the number of columns $c=2$.}%

We assume there is some unknown true  feasible set of attribute combinations $\combinationsetall$, where $\categoryactual \triangleq \brd{\combinationsetall}$ %
and $\combinationset\subseteq\combinationsetall$. 
\cameraready{In our example, $\categoryactual=3$ because in the U.S., voters must be at least 18 years of age, {so the only feasible combinations of attributes are (Above 18, Registered), (Above 18, Not registered), (Under 18, Not registered)}.}
We use $\combinationsetestimate$ to denote the data holder's estimate of the true support $\combinationsetall$ (e.g., this could be obtained from public data),
where $\categoryestimate \triangleq \brd{\combinationsetestimate}$.
\cameraready{In our voting example, $\combinationsetall$ can be accurately estimated based on public information, i.e., $\combinationsetestimate=\combinationsetall$. {We start by analyzing this special case in \cref{sec:tradeoffs}, but in general, they need not be the same set, and we subsequently analyze robustness to misspecification of the feasible attribute set in \cref{sec:robustness}}.} %

The data release mechanism is designed such that samples with attribute combinations in $\combinationset\cup\combinationsetestimate$ may exist in the released dataset $\dataset'$. %
Let $\combinationsetpre$ and $\combinationsetnoise$ be the feasible and infeasible attribute combinations within the estimated set $\combinationsetestimate$ respectively, %
i.e., $\combinationsetpre=\combinationsetestimate\cap\combinationsetall$ and $\combinationsetnoise=\combinationsetestimate\setminus\combinationsetall$. Let $\categoryleast=\brd{\combinationsetpre}$ and $\categorynoise=\brd{\combinationsetnoise}$, so we have $\categoryestimate=\categoryleast+\categorynoise$.
We illustrate the relation between $\combinationset$, $\combinationsetall$, $\combinationsetestimate$, $\combinationsetpre$ and $\combinationsetnoise$ in \cref{fig:venn}.
\begin{figure}[htbp]
\centering
\includegraphics[width=0.7\linewidth]{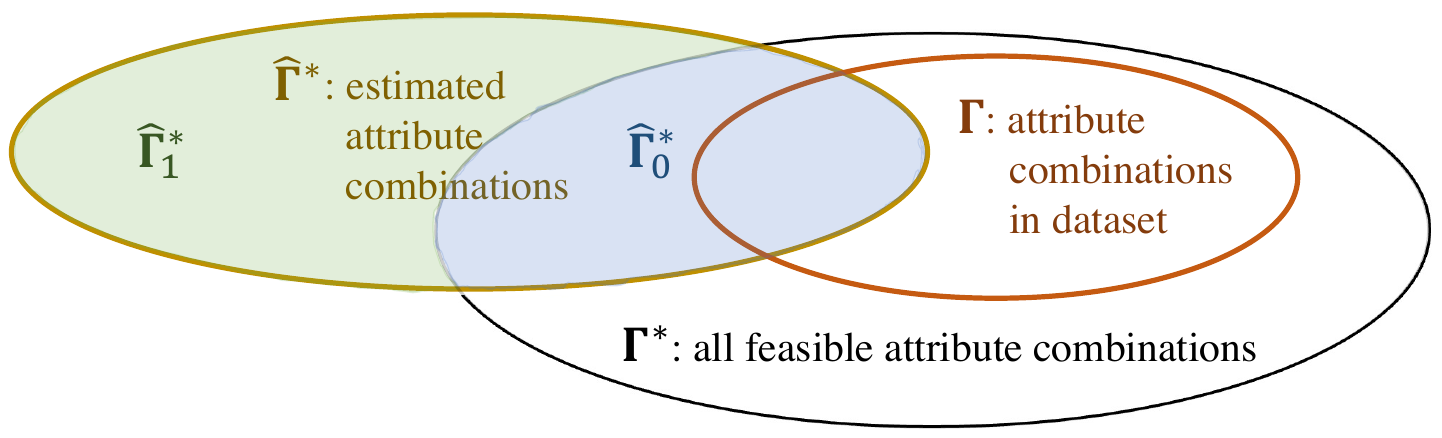}
\caption{Relation between $\textcolor{red}{\combinationset}$, $\combinationsetall$, $\textcolor{brown}{\combinationsetestimate}$, $\textcolor{blue}{\combinationsetpre}$ and $\textcolor{teal}{\combinationsetnoise}$. {$\combinationsetall$  is the set depicted in black, containing all feasible attribute combinations. $\textcolor{red}{\combinationset}$ in red is the set of attribute combinations existing in the dataset $\dataset$, and $\textcolor{red}{\combinationset}\subseteq \combinationsetall$. $\textcolor{brown}{\combinationsetestimate}$ in dark yellow is the data holder's estimate of $\combinationsetall$;  it may contain both feasible attribute combinations, as shown in the blue sub-region with notation $\textcolor{blue}{\combinationsetpre}$, and infeasible attribute combinations, as shown in the green sub-region with notation $\textcolor{teal}{\combinationsetnoise}$.%
}
}
\label{fig:venn}
\end{figure}

{In the remainder of this paper, we model a tabular dataset as a one-dimensional histogram (which has a one-to-one correspondence with an underlying parameter vector $\theta$).}
In particular, suppose the released dataset $\dataset'$ has the same size as the original dataset $\dataset$. $\dataset$ and $\dataset'$ can be represented by \camerareadydelete{$X_\theta$ and $Y_{\theta'}$ with}\cameraready{categorical distribution parameters $\theta$ and $\theta'$ respectively, {where each category corresponds to  {the fraction of records (out of $n$) with a given combination of attribute values}, e.g., (Yes, Yes) in our voting example.
{So for example, suppose our input dataset contained four records, as listed in \cref{fig:tabular} below. Then we would generate from this dataset a histogram with $\brd{\combinationset\cup\combinationsetestimate}$ bins, where each bin contains the fraction of records with a specific combination of attribute values.}
{This representation of tabular, categorical data captures correlations across columns by modeling the empirical joint distribution.} 
Finally, let the data holder's precision level of the categorical distribution be $\samplenum$, i.e., the probability mass of any category is a value within set $\brc{0, 1/\samplenum, 2/\samplenum, \cdots, 1}$. We assume $\samplenum\leq 
\samplenumori$, where $\samplenumori$ is the number of samples in the dataset.}}

\begin{figure}[htbp]
\centering
\includegraphics[width=0.7\linewidth]{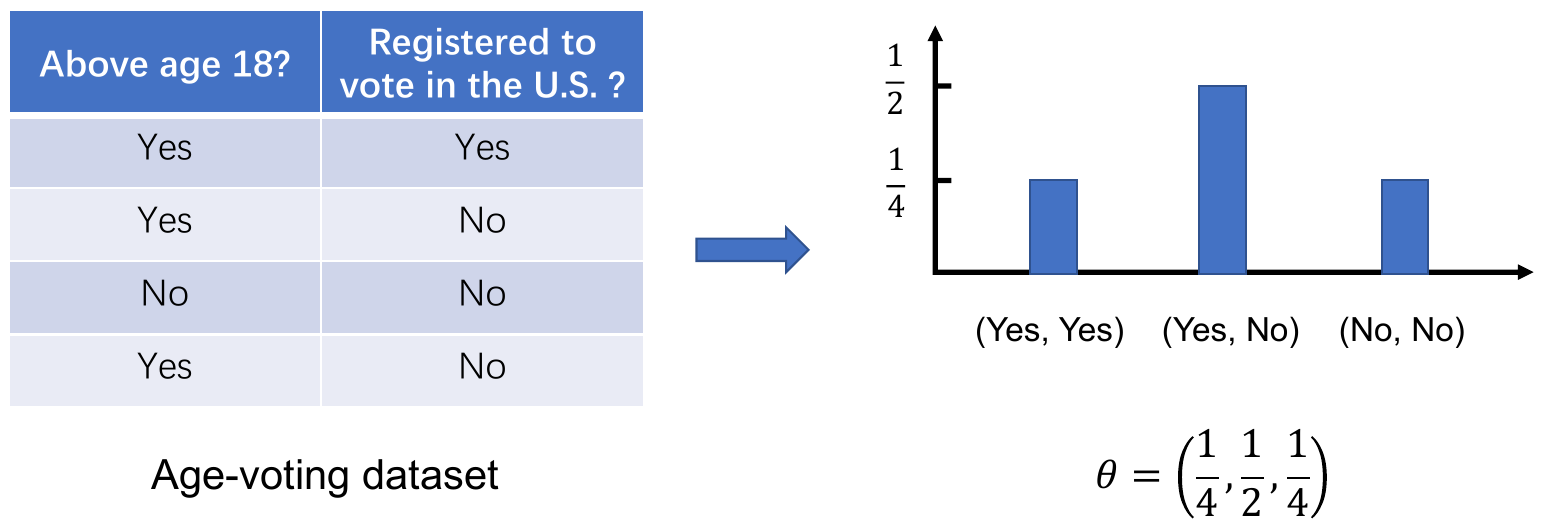}
\caption{For a tabular dataset, we can generate a categorical distribution with each category corresponding to the fraction of records with a given combination of attribute values. $\theta$ is the parameter of the constructed categorical distribution.}
\label{fig:tabular}
\end{figure}

Suppose the set of secret values admits a total ordering, and without loss of generality, suppose that 
$g_1<g_2<\cdots<g_{\secretnum}$. 
For convenience, we use $\secretsetofabbr{i}$ to represent $\secretsetof{i}$, the parameter set with secret value $g_i$, $\forall i\in \brb{\secretnum}$. 

\subsection{Mechanism definition}
\paragraph{Randomized Response (\rr{})} 
\rr{} is a widely-used mechanism in the DP literature for discrete distributions \cite{kairouz2016extremal}. 
We first consider a \rr{} mechanism $\mechrr$ that outputs the original distribution parameter with some  probability, and otherwise releases a different  distribution's parameters uniformly at random {over the parameter set}. 
Specifically, $\mechrr$ can be written as follows, {$\forall \theta \in \paramset{}$}:
\begin{align*} 
\probof{\mechrr\bra{\theta}=\theta'} =
\begin{cases}
\frac{e^{\epsilon}}{\brd{\releaseparamset}+e^{\epsilon}-1}, & \theta' = \theta,\\
\frac{1}{\brd{\releaseparamset}+e^{\epsilon}-1}, & \theta'\in \releaseparamset\setminus \brc{\theta}.
\end{cases}
\end{align*}

\paragraph{Quantization Mechanism (\qm{})}
\qm{} has also been adopted in privacy-preserving mechanism design \cite{lin2023summary,issa2019operational,farokhi2019development}.  This mechanism partitions the secret values into subsets of size $\intervallen$ and uniformly releases a distribution parameter with the secret as the median index of the corresponding bin. 
Specifically, 
$\mechquan$ can be written as  $\forall k \in \brc{0, 1, \cdots, \left\lceil \frac{\secretnum}{\intervallen}\right\rceil-1}, j\in \brb{\intervallen}, \parameterdistribution \in \secretsetofabbr{k\intervallen+j}:$
\begin{align*}
\mechquan\bra{\theta} \sim \text{Unif}\bra{\releasesecretsetofabbr{\midpoint{k}}},
\end{align*}
\normalsize
where $\midpoint{k} = \floor{\bra{k+\frac{1}{2}}\intervallen}+1$, and $\releasesecretsetofabbr{\midpoint{k}}$ represents the released parameter set with secret value $\secretrv_{\midpoint{k}}$.
{More generally, if there is not an ordering of the secret, QM randomly at uniform releases a distribution parameter whose secret is contained in the corresponding bin. In the following analysis, we assume that the secret is the fraction of a category (i.e., one of the PMF values).}

\subsection{Privacy-Utility Tradeoffs}
\label{sec:tradeoffs}
When $\combinationsetestimate=\combinationsetall$, i.e., the data holder correctly estimates the feasible attribute combination set, we analyze and compare the privacy-distortion tradeoffs between \rr{} and \qm{} as follows. Note that in our analysis of \qm{}, we constrain the secret function $\secretnotation$ to be the PMF value for a
specific category (e.g., {``the fraction of people above age 18 registered to vote in the U.S.''} %
) {to simplify the analysis}.

\begin{restatable}{theorem}{privacydistortionrr}
\label{prop:mech_tradeoff}
\emph{(Privacy and Distortion of Randomized Response)} %
\cameraready{For any secret function $\secretnotation$,} the SML and distortion of \rr{} are:
\begin{align*}
\privacyrr=\log\frac{1+\secretnum\rrratio}{1+\rrratio}, \quad \distortionrr = \frac{\categoryleast-1}{\categoryleast\bra{1+\rrratio}}.
\end{align*}
where $\rrratio \triangleq \frac{e^{\epsilon}-1}{\brd{\releaseparamset}}=\frac{e^{\epsilon}-1}{\binom{\samplenum+\categoryleast-1}{\categoryleast-1}}$.

\noindent
\emph{(Privacy and Distortion of Quantization Mechanism)}  %
The privacy of \qm{} is $\privacyquan = \log\ceil{\frac{\secretnum}{\intervallen}}$. When secret is the fraction of a category, the distortion of \qm{} is 
\begin{align*}
\distortionquan = \frac{1}{2}+\frac{\categoryleast\floor{\frac{\intervallen}{2}}-\samplenum}{2\samplenum\bra{\categoryleast-1}}.
\end{align*}

\noindent
\emph{(Mechanism Comparison)}
When secret is the fraction of a category, for any non-trivial privacy budget $T<\log\secretnum$, if $\privacyquan=\privacyrr\leq T$, we have 
$
\lim_{\samplenum\rightarrow\infty}\frac{\distortionrr}{\distortionquan} \geq 1.
$
\end{restatable}

\noindent(Proof in \cref{sec:proof_mech_tradeoff}.) From \cref{prop:mech_tradeoff}, we know that {when the precision level is high enough}, {for the same SML guarantee,} \qm{} \emph{performs at least as well as} \rr{} as long as SML does not achieve its upper bound $\log\secretnum$. Intuitively, this is because the output space of \rr{} covers the full support of $\Theta'$, while the output space of \qm{} is significantly reduced.

\subsubsection{Robustness to mis-specification of feasible attribute combinations}
\label{sec:robustness}
When \camerareadydelete{$\combinationsetpre\subsetneq \combinationsetall$}\cameraready{$\combinationsetestimate\not= \combinationsetall$}, i.e., the data holder only partially knows the feasible attribute combination set, we provide a robustness result for the privacy of the mechanisms. 

\begin{definition}[Robustness to support mismatch]
Consider a tabular dataset $\mathcal{D}$ with attribute combination set $\combinationset$.
For any mechanism $\mech$, let $\sml$ be the SML of $\mech$ if its released dataset contains samples with  attribute combinations in $\combinationsetestimate\cup\combinationset$, 
and  $\sml^*$ be the SML of $\mech$ if its released dataset only contains samples with attribute combinations in $\combinationsetall$. The mechanism $\mech$ is $\gamma$-robust if for any $\combinationset$,
$
\privacynotation_\mech - \privacynotation^*_\mech \leq \gamma\bra{\categoryactual-\categoryleast}.
$
\end{definition}

\begin{restatable}{proposition}{robustnesstabular}
\label{prop:robustness_tabular}
Consider a dataset $\mathcal{D}$ with $\samplenum$ samples. %
Regardless of the value of $\categorynoise$,
\rr{} is $\log 3$-robust if its hyperparameter $\epsilon$ satisfies $e^{\epsilon}-1 
\leq \binom{\samplenum+\categoryestimate-1}{\categoryestimate-1} / \secretnum$, and %
\qm{} with any interval length $\intervallen$ is $1$-robust when the secret is the fraction of a category.
\end{restatable}

\noindent(Proof in \cref{section:proof_robustness_tabular}.) \cref{prop:robustness_tabular} indicates that \rr{} is robust to support mismatch %
when it satisfies certain privacy constraints, and \qm{} is robust under certain secret types.

We also analyze the upper and lower bounds on privacy and distortion for both mechanisms when $\combinationsetestimate\not= \combinationsetall$, as well as the tightness of the privacy bounds. {To improve the readability of the paper, the details are deferred to \cref{section:proof_bounds}.}
\section{Empirical Evaluation}
\label{sec:empirical}

In this section, we conduct empirical experiments on a real world dataset to analyze and compare the privacy-distortion tradeoffs of our proposed Randomized Response (\rr{}) and Quantization Mechanism (\qm{}), and evaluate the performance of Quantization Mechanism under downstream tasks. We first introduce the experimental settings. 

\subsection{Experimental Setting}
\paragraph{Dataset}
We conduct our empirical evaluation on the Census Income dataset \cite{misc_adult_2}, which collects information from 48,842 individuals ($\samplenumori = 48,842$) about their income, education level, age, gender, and more. We set the distribution precision level as $\samplenum=\samplenumori$. The dataset contains 22,381 unique attribute combinations ($\categorynum=22,381$), all of which we assume to be contained in the estimated attribute combination set $\combinationsetestimate$, i.e., $\combinationset \subseteq \combinationsetestimate$. As described in \cref{sec:tabular_setup}, the category distribution parameter $\theta$ can be constructed with each category corresponding to a certain attribute combination, e.g., (white, male, age 32, master's degree, etc), and the value being set as the fraction of samples with this attribute combination.

\paragraph{Baselines}
Besides the Quantization Mechanism and Randomized Response proposed, we also include and compare with the mechanism of Wu et al. \cite{wu2020optimal}, which satisfies a maximal leakage guarantee. For ease of reference, we call this mechanism \greedy{}. {It is known to exhibit a privacy-utility tradeoff with bounded sub-optimality under a class of utility functions that does not include distance measures like total variation distance. However, we do not know of other mechanisms that satisfy a maximal leakage guarantee and can be applied to our dataset.} 

In short, for a given distribution parameter set $\paramset$ and a released distribution parameter set $\releaseparamset$, let $\mathbf{A}$ be a subset of $\releaseparamset$, $u$ be the cost function (total variation distance in our experiment), $f$ be a mapping function that satisfies $f\bra{\theta} = \arg\min_{\theta'\in\mathbf{A}}u\bra{\theta,\theta'}$, and $\mathcal{U}_{\paramdistribution}\bra{\mathbf{A}}=\mathbb{E}_{\Theta}\brb{u(\theta, f\bra{\theta})}$ be the expected cost under the prior distribution $\paramdistribution$. \greedy{} works by initializing $\mathbf{A}=\emptyset$, and iteratively adding an element $\theta'\in\releaseparamset\setminus\mathbf{A}$ to $\mathbf{A}$, where $\theta'$ satisfies $\theta'=\arg\min_{\theta'\in\releaseparamset\setminus\mathbf{A}}\mathcal{U}_{\paramdistribution}\bra{\mathbf{A}\cup \theta'}$. MaxL terminates when $\mathbf{A}=\releaseparamset$ or $\min_{\theta'\in\releaseparamset\setminus\mathbf{A}}\mathcal{U}_{\paramdistribution}\bra{\mathbf{A}\cup \theta'}=\mathcal{U}_{\paramdistribution}\bra{\mathbf{A}}$, and design the data release mechanism $\mech_{\text{MaxL}}$ according to the mapping function $f$.

Since we assume the prior distribution $\paramdistribution$ is unknown in our setting, we replace the expected cost $\mathcal{U}_{\paramdistribution}\bra{\mathbf{A}}$ in MaxL by $\max_{\paramdistribution}\mathcal{U}_{\paramdistribution}\bra{\mathbf{A}}$. When the cost function is set as total variation distance, the performance of MaxL is sensitive to the choice of the released parameter set $\releaseparamset$. For example, if $\releaseparamset=\paramset$, we can easily check that MaxL will always release the original distribution parameter, i.e., $\mech_{\text{MaxL}}\bra{\theta} = \theta$, resulting in no privacy protection. To ensure a fair comparison, we design the released parameter set $\releaseparamset$ based on a similar idea to the Quantization Mechanism, that is, we partition the secret values into subsets of size $\intervallen$, and for each subset, we uniformly select a distribution parameter with the secret within the subset and add it to $\releaseparamset$. We vary the interval length $\intervallen$ to change the privacy guarantee of MaxL.

\paragraph{Evaluation Metrics}
Due to the above results showing that attribute privacy has a poor privacy-utility tradeoff under our setting, 
we use SML to measure the privacy performance in the remainder of this paper. To measure the utility of the data release mechanism, we adopt the worst-case expected total variation distance from \cref{eq:utility}. 
We quantify downstream utility by conducting a classification task on the downstream data; namely, we learn a classifier predicting whether an individual has high income or not based on a random forest trained on the released (perturbed) data and compute ROC curves.

\subsection{{Comparison to Attribute Privacy}}
\label{sec:inf-privacy}

\begin{figure}[t]
\centering
\begin{subfigure}{0.46\textwidth}
         \centering
\includegraphics[width=\linewidth]{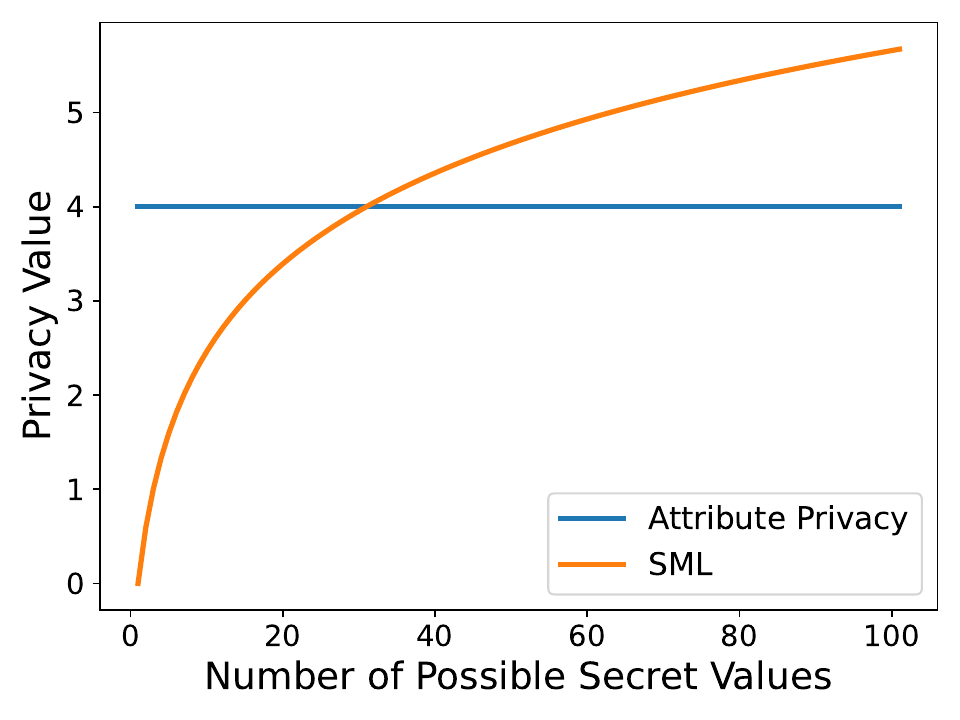}
\caption{Privacy value of \rr{} with $\epsilon=4$ v.s. Number of possible secret values}
\label{exp:secret_num}
\end{subfigure}
\begin{subfigure}{0.46\textwidth}
         \centering
        \includegraphics[width=1\linewidth]{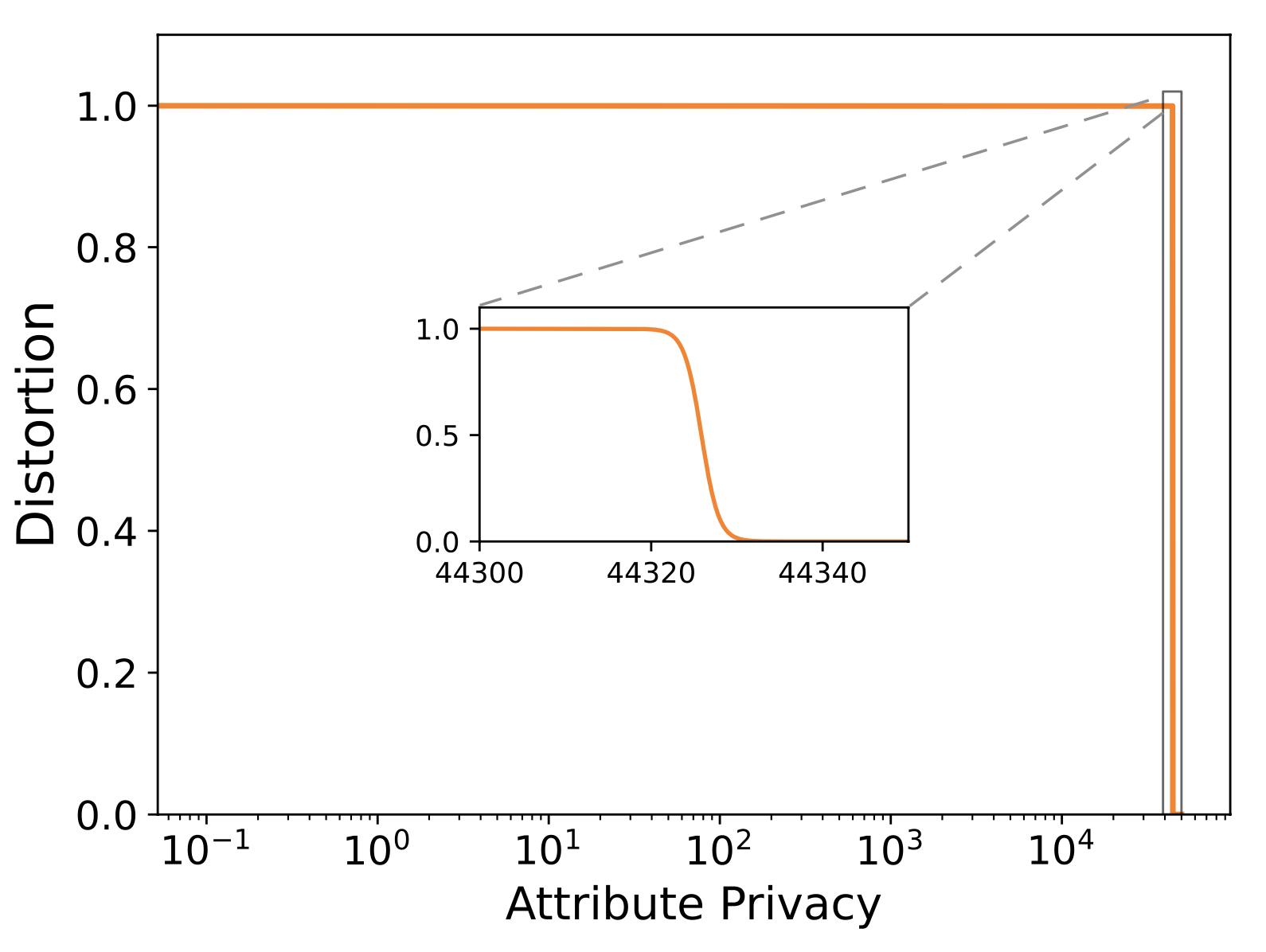}
         \caption{Distortion v.s. Attribute Privacy for \rr{} with $\epsilon=4$\\ \ }
         \label{exp:ap_distortion}
     \end{subfigure}
\caption{Under attribute privacy, the privacy level of \rr{} does not vary with the number of possible secrets, indicating its inability to capture the influence of the secret value space on the information leakage (\cref{exp:secret_num}); the distortion of RR always achieve its trivial upper bound 1 until the attribute privacy value is
larger than 44,320, indicating that attribute privacy is a much more conservative metric (\cref{exp:ap_distortion}). 
}
\label{fig:IP_SML}
\end{figure}

We first compare the privacy (and utility) of a mechanism (RR with hyperparameter $\epsilon=4$) under the privacy measures SML and a differential privacy (DP)-based measure; namely, we consider attribute privacy (AP) \cite{zhang2022attribute}.
We chose RR because it satisfies both AP and SML, whereas QM does not satisfy a finite AP guarantee for any interval size {less than the full range of secrets}. Additionally, the mechanism proposed in \cite{zhang2022attribute} applies only to continuous query results from the dataset, rather than to the entire dataset.
A fair comparison between SML and AP is challenging because these two measures have different definitions, and hence the numeric values of SML and attribute/inferential privacy cannot be directly compared. We thus run two experiments: 
(1) We fix the dataset and mechanism, and study how the values of SML and AP \emph{change} with the number of
possible secret values. (2) We evaluate the privacy-distortion tradeoff for the RR mechanism under the AP privacy measure.
First, in \cref{exp:secret_num}, we vary $s$, the number of possible secret values the original dataset can take, by quantizing the secret quantity to different granularities (i.e., the smaller the quantization interval, the greater the number of possible secret values).%
As the number of possible secret values increases, the value of SML increases, i.e., the privacy measure degrades. This aligns with our intuition: for a given dataset, under the same Randomized Response mechanism, a larger number of possible secret values results in more uncertainty regarding the true secret value prior to data release. However, under randomized response, the information the adversary obtains after data release remains the same---specifically, the probability of RR outputting the original data remains the same regardless of the number of secrets. Therefore, we intuitively expect that RR with more secrets should leak more information, and we observe that SML has this property. %
In contrast, attribute privacy only measures whether the mechanism outputs are different under different secrets, and only depends on the mechanism itself; hence, the attribute privacy level of \rr{} does not change with the number of possible secrets. Furthermore, we evaluate the privacy-distortion tradeoff of \rr{} under attribute privacy in \cref{exp:ap_distortion}, where the distortion metric is the worst-case total variation distance from \cref{eq:utility}. \cref{exp:ap_distortion} shows that the distortion of \rr{} achieves its trivial upper bound of $1$ until the attribute privacy value is larger than $44,320$, which is impractically large. 
This suggests that for our utility metric (which is a worst-case metric, and hence may be conservative) attribute privacy incurs a poor privacy-utility tradeoff for reasonable privacy parameters.

\subsection{Privacy-Distortion Tradeoffs}

\begin{wrapfigure}{r}{0.5\textwidth}
    \centering
        \includegraphics[width=\linewidth]{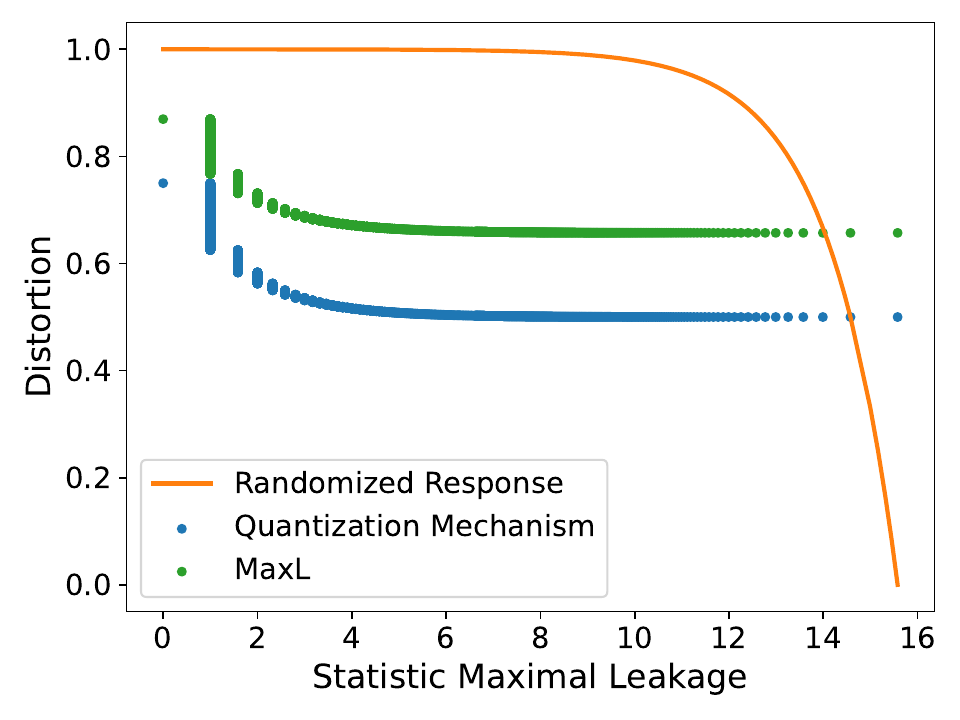}
\caption{Privacy-utility trade-offs of \rr{}, \qm{}, and \greedy{} when the secret is the fraction of an arbitrary category. %
}
\label{fig:trade-offs}
\end{wrapfigure}

In this section, we compare the privacy-distortion tradeoffs between \qm{}, \rr{}, and \greedy{} with the assumption that $\combinationset=\combinationsetall=\combinationsetestimate$, i.e., all feasible attribute combinations are contained in the dataset, and there is no estimate error for the attribute combinations. We defer the privacy-distortion tradeoff analysis under the case where $\combinationset\subseteq\combinationsetall\not=\combinationsetestimate$ in \cref{app:exp_more_res}.

We consider the secret as the fraction of an arbitrary category---{namely, (white, male, age 32, Master's degree)}---and compare the privacy-utility tradeoffs between \rr{}, \qm{}, and \greedy{} in \cref{fig:trade-offs}. To obtain different privacy guarantees, we vary the hyperparameter $\epsilon$ in \rr{} and vary the interval length $\intervallen$ in \qm{} and \greedy{}. The SML for this setting takes values in  $[0,16]$, where the upper bound arises from the limited size of the secret value space. 
From \cref{fig:trade-offs}, we observe that for the same privacy guarantee, the distortion of \rr{} is almost twice as large as that of \qm{} under most SML levels, indicating better performance of \qm{}, which is in line with the theoretical insight in \cref{prop:mech_tradeoff}. For MaxL, which can be regarded as a variant of \qm{} in our setting, it also has a lower distortion than \rr{} at most SML levels, while showing a worse privacy-distortion performance than \qm{}.
{This intuitively makes sense, because MaxL is designed to prevent leakage of \emph{any} secret, whereas QM is only protecting the specified secret.}

\subsection{Evaluation of the Quantization Mechanism under Downstream Tasks}

We next analyze the downstream utility of \qm{} %
by conducting a classification task to predict whether an individual has high income or not based on a random forest model trained on the released data. %
We set the secret as the difference between the proportion of white and non-white high-income people (>\$50k/yr) among their own race groups. The SML for this setting takes values in  $[0,4.09]$. 

\subsubsection{Perfectly-Calibrated Attribute Sets}
We first consider the case where $\combinationset=\combinationsetall=\combinationsetestimate$, and vary the quantization set size $\intervallen$ %
to achieve different levels of SML from $0$ to $2$. %
For each SML level, we conduct the experiment $20$ times with independent mechanism outputs, show the averaged ROC curve of the random forests trained on corresponding released datasets. We then compare this with the performance of the random forest trained on the original dataset in \cref{fig:roc_quan}.
We observe that as SML increases (weaker privacy), AUC (area under the ROC curve) increases, indicating the improvement of the downstream task utility. When SML is as little as $2$, the AUC is close to its upper bound ($0.89$ in raw dataset). When SML is $0$ (perfect privacy), the utility drops mildly ($0.14$ AUC) compared to the utility of the original dataset. Note that a perfect privacy-preserving mechanism can still achieve reasonable utility on this task since we only aim to protect a secret of the dataset, rather than the whole data. 

\begin{figure}[htbp]
\centering
\includegraphics[width=0.5\linewidth]{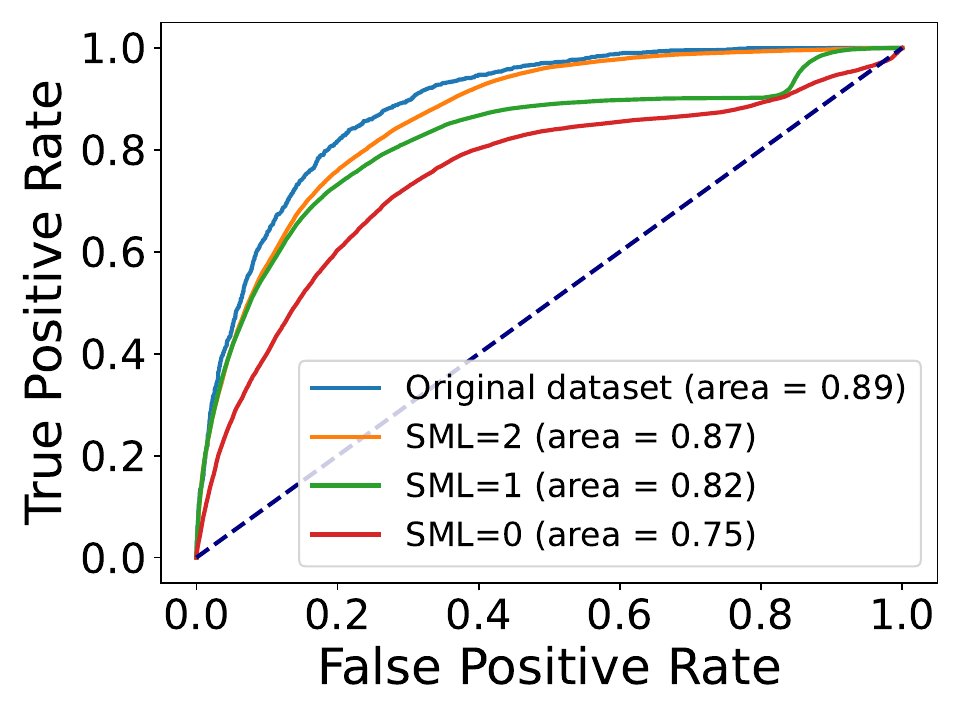}
\caption{Comparison of ROC curves of random forests under \qm{} with different SML.}
\label{fig:roc_quan}
\end{figure}

\begin{figure}[t]
    \centering
    \begin{subfigure}
    {0.48\textwidth}
         \centering
        \includegraphics[width=1\linewidth]{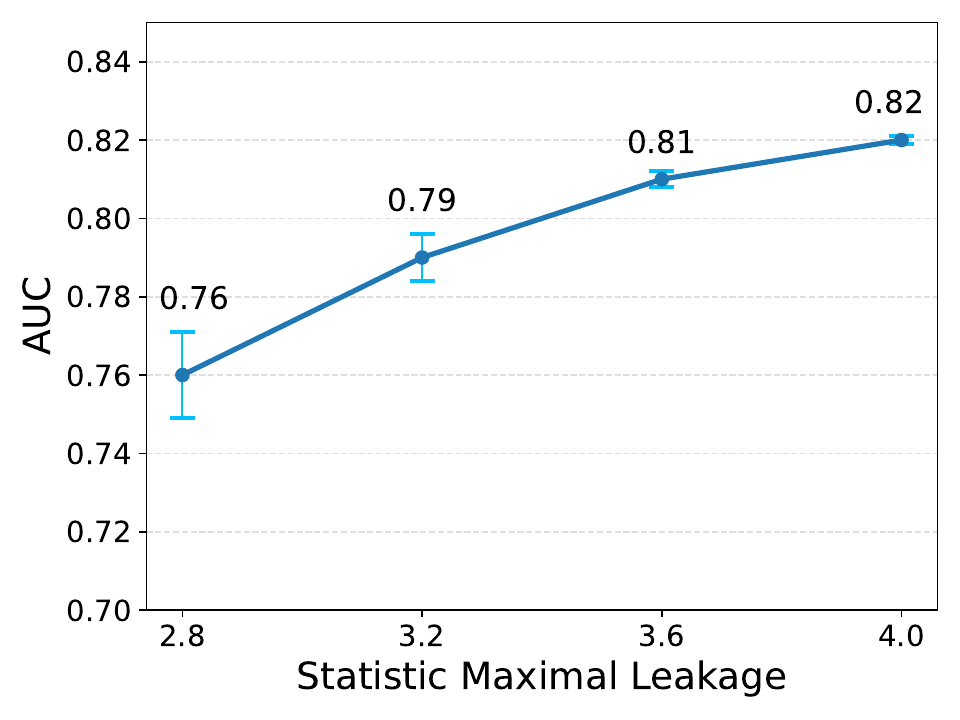}
         \caption{$\categorynoise=0, \categoryactual-\categoryleast=5$.}
         \label{fig:roc_vary_budget}
     \end{subfigure}
     \begin{subfigure}{0.48\textwidth}
         \centering
        \includegraphics[width=1\linewidth]{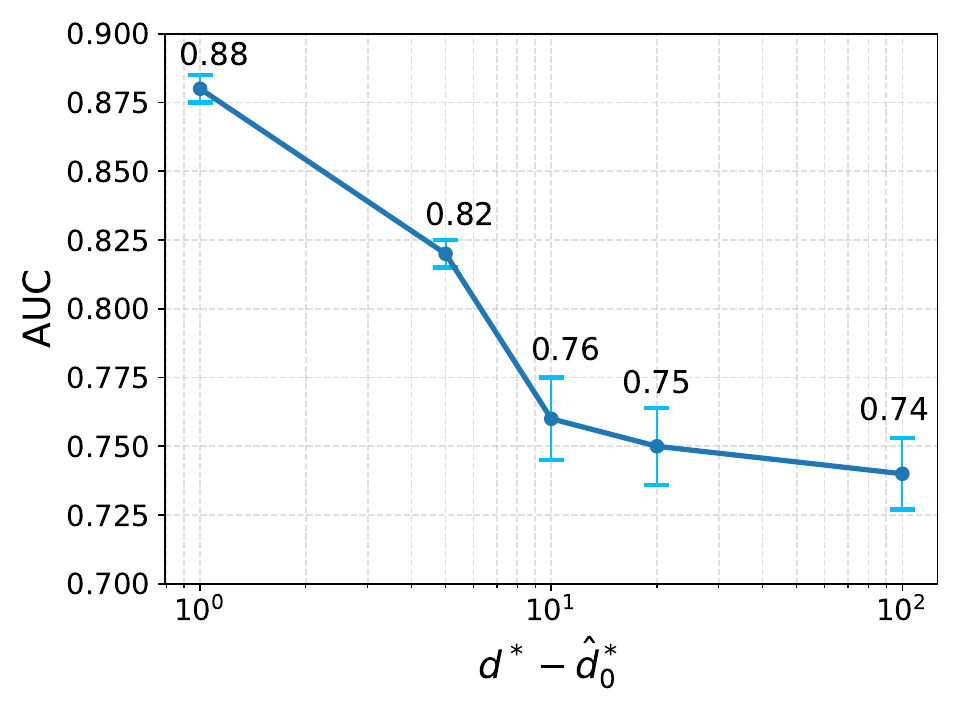}
         \caption{SML$\ =4$, $\categorynoise=0$.}
         \label{fig:roc_vary_miss}
     \end{subfigure}

     \begin{subfigure}{0.32\textwidth}
         \centering
        \includegraphics[width=1\linewidth]{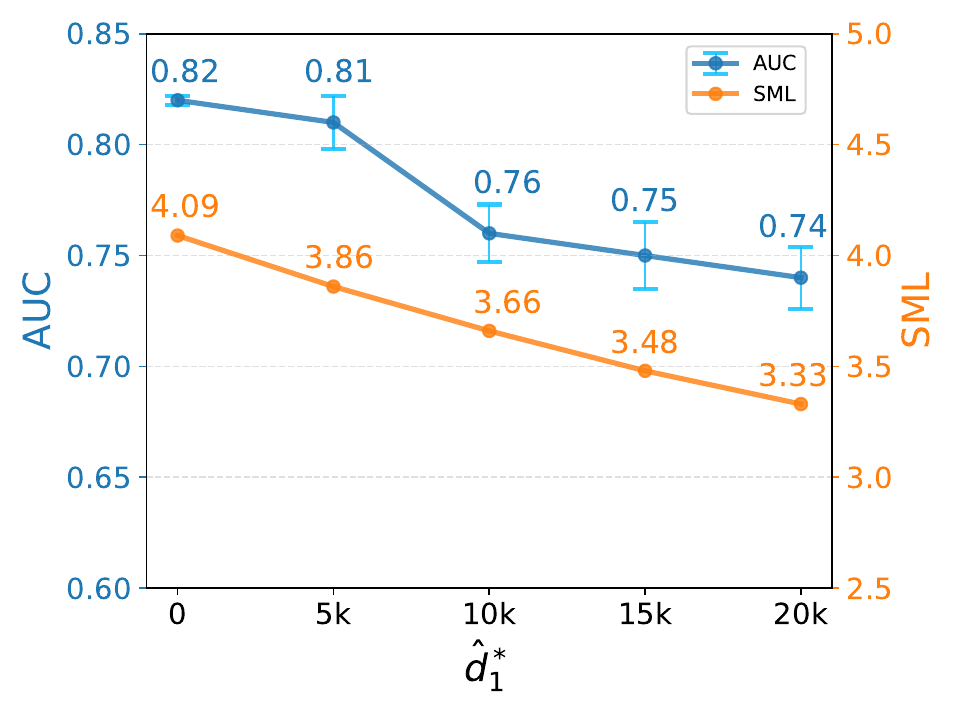}
         \caption{$\categoryactual-\categoryleast=5$.}
         \label{fig:roc_vary_d'_5}
     \end{subfigure}
     \begin{subfigure}{0.32\textwidth}
         \centering
        \includegraphics[width=1\linewidth]{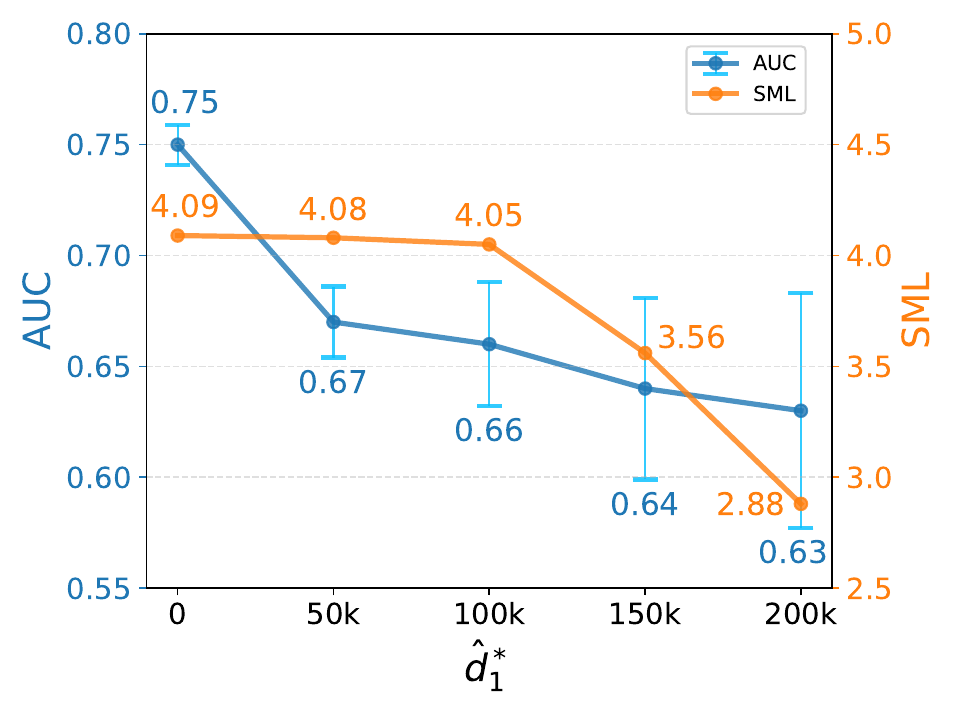}
         \caption{$\categoryactual-\categoryleast=20$.}
         \label{fig:roc_vary_d'_20}
     \end{subfigure}
     \begin{subfigure}{0.32\textwidth}
         \centering
        \includegraphics[width=1\linewidth]{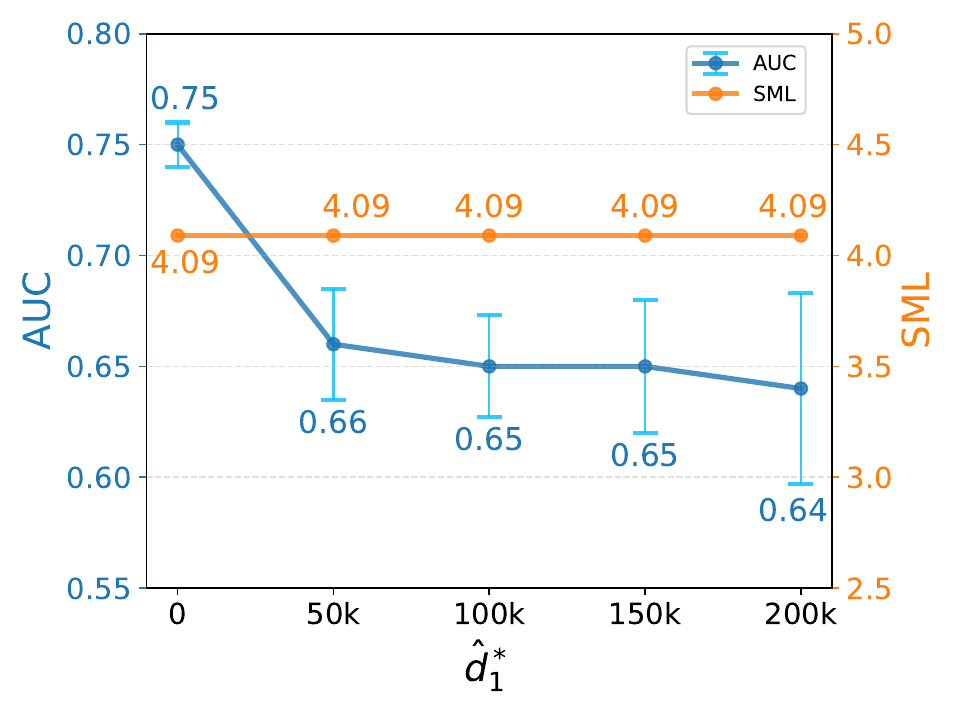}
         \caption{$\categoryactual-\categoryleast=100$.}
         \label{fig:roc_vary_d'_100}
     \end{subfigure}
     \caption{Comparisons of ROC curves of random forests trained on released data by quantization mechanism with different SML, missed feasible attribute combinations $\bra{\categoryactual-\categoryleast}$, or invalid attribute combinations inserted $\bra{\categorynoise}$.%
     }
\label{fig:roc_quan_miss}
\end{figure}

\subsubsection{Mis-calibrated Attribute Sets}
We then focus on the case where $\combinationset\subseteq\combinationsetall\not=\combinationsetestimate$ in \cref{fig:roc_quan_miss}. Under each setting, we conduct 20 independent experiments, show the average AUC with standard deviation of the random forests trained on the released datasets.

In \cref{fig:roc_vary_budget}, we compare the AUC of random forests with different SML levels when the data holder misses $5$ feasible attribute combinations in their pre-held attribute set, i.e., $\categoryactual-\categoryleast=5$. We consider the case where the data holder does not introduce invalid attribute combinations, i.e., $\categorynoise=0$, and vary the SML from $2.8$ to $4$. We observe that compared to \cref{fig:roc_quan}, the performance of the random forests is worse even under more relaxed privacy constraints (AUC = $0.82$ when SML = 4 v.s. AUC = $0.82$ when SML = 1 in \cref{fig:roc_quan}).

In \cref{fig:roc_vary_miss}, we fix the privacy budget as SML = 4, and compare the utility when the data holder has different prior knowledge of feasible attribute combinations, i.e., different number of missing feasible combinations $\categoryactual-\categoryleast$. We vary $\categoryactual-\categoryleast$ from $0$ to $100$, and, similar to the settings of \cref{fig:roc_vary_budget}, set $\categorynoise=0$. As we can observe, with the increasing number of missing feasible attribute combinations, the utility first drops sharply -- AUC drops from $0.87$ to $0.76$ as $\categoryactual-\categoryleast$ increases from $0$ to $10$, and then drops mildly -- AUC drops from $0.76$ to $0.74$ as $\categoryactual-\categoryleast$ increases from $10$ to $100$.

In \cref{fig:roc_vary_d'_5,fig:roc_vary_d'_20,fig:roc_vary_d'_100}, we fix the data holder's pre-held attribute combination sets with $\categoryactual-\categoryleast=5, 20$ and $100$ respectively, and compare the utility under different numbers of introduced invalid attribute combinations, i.e., $\categorynoise$. We observe in \cref{fig:roc_vary_d'_5} that as $\categorynoise$ increases, the privacy upper bounds drops linearly while the utility only degrades slightly when $\categorynoise$ increases from $0$ to 5k, or from 10k to 20k. When the number of feasible attribute combinations the data holder misses increases to $20$ in \cref{fig:roc_vary_d'_20}, the privacy upper bounds drops from $4.09$ to $2.88$ in a concave curve as $\categorynoise$ increases, from 50k to 200k, while the utility only suffers a slight degradation as $\categorynoise$ increases from 50k to 200k. Those results indicate that introducing more invalid attribute combinations results in better (lower) SML while not severely harming utility. We also observe that as the number of missing feasible attribute combinations increases, more invalid attribute combinations are needed to achieve a given improvement in SML. However, as we can observe in \cref{fig:roc_vary_d'_100}, when $\categoryactual-\categoryleast=100$, %
the SML of the mechanism always achieves its trivial upper bound whatever $\categorynoise$ is, which indicates that with more missing feasible attribute combinations, adding invalid attribute combinations eventually cannot help improve the privacy performance of the mechanism. In \cref{sec:curve_explain}, we show through an approximation argument why the privacy expressions decay in the patterns observed in \cref{fig:roc_vary_d'_5,fig:roc_vary_d'_20,fig:roc_vary_d'_100}. %
Roughly, in the approximation regime where $\categorynoise\ll \categoryleast$ and $\categoryactual-\categoryleast \ll \samplenum+\categoryleast$, there exists a positive constant $\alpha = \frac{\log \secretnum - \log \ceil{ \frac{\secretnum}{\intervallen}}}{\log \brb{1+{\samplenum}/\bra{\samplenum+\categoryleast}}}$ such that with increasing $\categorynoise$, the privacy value drops linearly when $0\leq \categoryactual-\categoryleast< \alpha$, and always achieves its trivial upper bound $\log \secretnum$ when $\categoryactual-\categoryleast \geq \alpha$.

\section{Conclusion and Discussion}
\label{sec:conclusion}

In this work, we propose \emph{\privname{}}, a prior-independent and secret-specific privacy measure quantifying the information leakage of a given, known secret. We show that SML satisfies post-processing and composition properties. Although the calculation of statistic maximal leakage is NP-hard in general, under deterministic mechanisms, we show an efficient \privname{} calculation process. %
We analyze and compare the privacy-utility tradeoffs of randomized response and the quantization mechanism.
Both theoretical and empirical results suggest a better privacy-utility tradeoff for the quantization mechanism. %
However, there still remains several interesting future problems.

\paragraph{Design of Alternative Metrics}
The operational meaning of \privname{} is to measure the information gain in an attacker's knowledge about the most likely secret value after data release, under worst-case prior and attack strategy. However, in some scenarios, especially those related to social justice and fairness, protecting the most unlikely posterior secret value can also be crucial. Additionally, while \privname{} assesses the attack success probability over the randomness of the released data, a more conservative approach would measure this probability under the worst-case output. %
In \cref{sec:metric_comparison}, we take a first step towards alternative metrics under those scenarios and analyze their connections with the indistinguishability-based measure inspired by differential privacy.

\paragraph{Extension to Continuous Parameter Space}
The definition and analysis of statistic maximal leakage  in \cref{sec:formulation} are under the discrete parameter space. When extending to the continuous distribution parameter and secret spaces, statistic maximal leakage can be written as
\begin{align}
    \sml = \sup_{\pdf{\Theta}, \pdf{\hat{G}|\Theta'}}\log \frac{\pdfnotation\bra{\hat{G}=G}}{\sup_{\secretrv\in \secretvalueset} \pdfof{G}{\secretrv}},
    \label{eqn:sml_continuous}
\end{align}
where $\pdfnotation$ denotes the probability density function. Similar to \cref{prop:sml_calculation}, we provide an equivalent expression of \cref{eqn:sml_continuous} under continuous parameter space in \cref{section:proof_sml_calculation_continuous}. 
However,  efficiently calculating or estimate \privname{} over a continuous parameter space is an open problem, as is designing data release mechanisms with desirable privacy-utility tradeoffs. 

For instance, we study tabular datasets with categorical values in this work. This is a very limited class of data. Designing mechanisms that satisfy an SML guarantee for more general classes of data is an important and interesting question. 

\paragraph{Relation Between Secret and Utility} In our experiments, we (perhaps surprisingly) observed nonzero utility even at a perfect SML guarantee of zero. This is likely due to our choice of secret, which was correlated with the selected downstream task very loosely. An interesting question is to characterize conditions on the secret and the downstream task to explain how SML will affect utility.

\section*{Acknowledgments}
The authors gratefully acknowledge the support of NSF grants CIF-1705007, CCF-2338772, and RINGS-2148359, as well as support from the Sloan Foundation, Intel, J.P. Morgan Chase, Siemens,
Bosch, and Cisco. This material is based upon work supported in part by the U.S. Army Research Office
and the U.S. Army Futures Command under Contract No. W911NF20D0002. 

\clearpage
\newpage


\newpage
\begin{appendices}
\section{Proofs of \cref{lemma:distortion,property:metric_compare,prop:sml_calculation}}

\subsection{Proof of \cref{lemma:distortion}}
\label{sec:proof_distortion}

\distortionmeasure*

\begin{proof}

We show that the proposed expression is both an upper bound and a lower bound on the original distortion measure.

{Recall that the original distortion measure is defined as 
$$
\smldistortion = \sup_{\paramdistribution} \cameraready{\mathbb{E}_{\Theta, \Theta'=\mech\bra{\Theta}}[\distanceof{\distributionof{X_{\Theta}}} {\distributionof{Y_{\Theta'}}}]}.
$$
}
Let $\paramdistribution^*=\arg\sup_{\paramdistribution} \mathbb{E}_{\Theta, \Theta'=\mech\bra{\Theta}}[\distanceof{\privatedistribution} {\releasedistribution}]$, and for convenience, we define $\mathcal P_\theta \triangleq \mathbb P^*_{\Theta}(\theta)$.
We can get that $\smldistortion$ can be upper bounded by
\begin{align*}
\smldistortion &= \sup_{\paramdistribution} \mathbb{E}_{\Theta, \Theta'=\mech\bra{\theta}}[\distanceof{\privatedistribution} {\releasedistribution}]\\
&= \sum_{\theta\in\paramset}\mathcal{P}_\theta\cdot\mathbb{E}_{\Theta'=\mech\bra{\theta}}[\distanceof{\privatedistribution} {\releasedistribution}]\\
&\leq\sup_{\theta} \mathbb{E}_{\Theta'=\mech\bra{\theta}}[\distanceof{\privatedistribution} {\releasedistribution}].
\end{align*}

To get the lower bound of $\smldistortion$, we construct a prior such that %
$\mathbb{P}_{\Theta}\bra{\theta}=1$ when $\theta = \arg\sup_{\theta} \mathbb{E}_{\Theta'=\mech\bra{\theta}}[\distanceof{\privatedistribution} {\releasedistribution}]$, and we have
\begin{align*}
\smldistortion &= \sup_{\paramdistribution} \mathbb{E}_{\Theta, \Theta'=\mech\bra{\theta}}[\distanceof{\privatedistribution} {\releasedistribution}]\\
&\geq\sup_{\theta} \mathbb{E}_{\Theta'=\mech\bra{\theta}}[\distanceof{\privatedistribution} {\releasedistribution}].
\end{align*}

Therefore,
$
\smldistortion =\sup_{\theta} \mathbb{E}_{\Theta'=\mech\bra{\theta}}[\distanceof{\privatedistribution} {\releasedistribution}].
$
\end{proof}

\subsection{Proof of \cref{property:metric_compare}}
\label{sec:proof_metric_compare}

\relation*

\begin{proof}

For the worst-case min-entropy leakage $\minentropy$, it can be derived as 
\begin{align*}
\minentropy &= \sup_{\paramdistribution}\log \frac{\sum_{\theta'\in\releaseparamset} \sup_{\theta\in \mathbf{\Theta}} \mathbb{P}_{\Theta\Theta'}\bra{\theta,\theta'}}{\sup_{\theta\in \mathbf{\Theta}} \mathbb{P}_\Theta\bra{\theta}}\\
&= \sup_{\paramdistribution}\log \frac{\sum_{\theta'\in\releaseparamset} \sup_{\theta\in \mathbf{\Theta}} \mathbb{P}_{\Theta'|\Theta}\bra{\theta'|\theta}\mathbb{P}_\Theta\bra{\theta}}{\sup_{\theta\in \mathbf{\Theta}} \mathbb{P}_\Theta\bra{\theta}}\\
    & \leq \sup_{\paramdistribution}\log \frac{\sum_{\theta'\in\releaseparamset} \sup_{\theta\in \mathbf{\Theta}} \mathbb{P}_{\Theta'|\Theta}\bra{\theta'|\theta}\cdot\sup_{\theta\in \mathbf{\Theta}}\mathbb{P}_\Theta\bra{\theta}}{\sup_{\theta\in \mathbf{\Theta}} \mathbb{P}_\Theta\bra{\theta}}\\
    &= \log \sum_{\theta'\in\releaseparamset} \sup_{\theta\in \mathbf{\Theta}}\mathbb{P}_{\Theta'|\Theta}\bra{\theta'|\theta}.
\end{align*}
This upper bound can be achieved when $\paramdistribution$ satisfies 
\begin{align*}
\forall \theta'\in\releaseparamset: \  \arg\sup_{\theta\in \mathbf{\Theta}} \mathbb{P}_{\Theta'|\Theta}\bra{\theta'|\theta}\cap \arg\sup_{\theta\in \mathbf{\Theta}}\mathbb{P}_\Theta\bra{\theta} \not= \emptyset.
\end{align*}

Thus,
$\minentropy=\log\sum_{\theta'\in\releaseparamset} \sup_{\theta\in \mathbf{\Theta}}\mathbb{P}_{\Theta'|\Theta}\bra{\theta'|\theta}$.

For \privname $\sml$, we defer the derivations to \cref{section:proof_sml_calculation}, from which we have
\begin{align}
\sml
&= \sup_{\paramdistribution}\log \sum_{\theta'\in\releaseparamset} \sup_{\secretrv\in \secretvalueset} \frac{\sum_{\theta\in \mathbf{\Theta}_{\secretrv}}\mathbb{P}_{\Theta'|\Theta}\bra{\theta'|\theta}\cdot\mathbb{P}_\Theta\bra{\theta}}{\sum_{\theta\in\mathbf{\Theta}_{\secretrv}}\mathbb{P}_\Theta\bra{\theta}}
\nonumber \\
&\geq \log
\sum_{\theta'\in\releaseparamset} \sup_{\secretrv\in \secretvalueset} \frac{\sum_{\theta\in \mathbf{\Theta}_{\secretrv}}\mathbb{P}_{\Theta'|\Theta}\bra{\theta'|\theta}}{|\mathbf{\Theta}_{\secretrv}|}
\label{eq:sml_ineq_1} \\
&\geq \log\frac{ \sum_{\theta'\in\releaseparamset} \sup_{\secretrv\in \secretvalueset} \sup_{\theta\in \mathbf{\Theta}_{\secretrv}}\mathbb{P}_{\Theta'|\Theta}\bra{\theta'|\theta}}{\sup_{\secretrv\in \secretvalueset}|\mathbf{\Theta}_{\secretrv}|}
\nonumber \\
&= \log\frac{ \sum_{\theta'\in\releaseparamset} \sup_{\theta\in \mathbf{\Theta}}\mathbb{P}_{\Theta'|\Theta}\bra{\theta'|\theta}}{\sup_{\secretrv\in \secretvalueset}|\mathbf{\Theta}_{\secretrv}|}
\label{eq:sml_ineq_2} \\
&= \minentropy-\sup_{\secretrv\in \secretvalueset}\log|\mathbf{\Theta}_{\secretrv}|. \nonumber
\end{align}
Equality is achieved in \eqref{eq:sml_ineq_1} %
when $\paramdistribution$ satisfies
$
\forall \secretrv\in\secretvalueset, \forall \theta_1, \theta_2 \in \mathbf{\Theta}_{\secretrv}: ~  \mathbb{P}_{\Theta}\bra{\theta_1} =  \mathbb{P}_{\Theta}\bra{\theta_2},
$ 
and \eqref{eq:sml_ineq_2} %
holds because $\cup_{\secretrv \in \secretvalueset} \mathbf{\Theta}_{\secretrv} = \mathbf{\Theta}$.
Additionally, we have
\begin{align*}
\sml 
&= \sup_{\paramdistribution} \log\sum_{\theta'\in\releaseparamset} \sup_{\secretrv\in \secretvalueset} \frac{\sum_{\theta\in \mathbf{\Theta}_{\secretrv}}\mathbb{P}_{\Theta'|\Theta}\bra{\theta'|\theta}\cdot\mathbb{P}_\Theta\bra{\theta}}{\sum_{\theta\in\mathbf{\Theta}_{\secretrv}}\mathbb{P}_\Theta\bra{\theta}}
\\
&\leq \sup_{\paramdistribution}\log \sum_{\theta'\in\releaseparamset} \sup_{\secretrv\in \secretvalueset} \frac{\sup_{\theta\in \mathbf{\Theta}_{\secretrv}}\mathbb{P}_{\Theta'|\Theta}\bra{\theta'|\theta}\cdot\sum_{\theta\in \mathbf{\Theta}_{\secretrv}}\mathbb{P}_\Theta\bra{\theta}}{\sum_{\theta\in\mathbf{\Theta}_{\secretrv}}\mathbb{P}_\Theta\bra{\theta}}
\\
&= \log \sum_{\theta'\in\releaseparamset} \sup_{\theta\in \mathbf{\Theta}}\mathbb{P}_{\Theta'|\Theta}\bra{\theta'|\theta}
\\
&= \minentropy.
\end{align*}
\normalsize

This gives
\begin{equation}
\begin{aligned}
\label{eqn:sml_minentropy}
\minentropy-\sup_{\secretrv\in \secretvalueset}\log|\mathbf{\Theta}_{\secretrv}| \leq \sml \leq \minentropy.
\end{aligned}
\end{equation}

For maximal leakage $\maxl$, from \cite[Thm.1]{issa2019operational}, we know that
\vspace{-2mm}
\begin{align*}
\maxl = \log\sum_{\theta'\in\releaseparamset} \sup_{\substack{\theta\in \mathbf{\Theta}:\\ \mathbb{P}_{\Theta}\bra{\theta}>0}}\mathbb{P}_{\Theta'|\Theta}\bra{\theta'|\theta}.
\end{align*}
Therefore, we have 
\begin{align*}
\sup_{\paramdistribution}\maxl &= \sup_{\paramdistribution}\log\sum_{\theta'\in\releaseparamset} \sup_{\substack{\theta\in \mathbf{\Theta}:\\ \mathbb{P}_{\Theta}\bra{\theta}>0}}\mathbb{P}_{\Theta'|\Theta}\bra{\theta'|\theta}\\
&= \log\sum_{\theta'\in\releaseparamset} \sup_{\theta\in \mathbf{\Theta}}\mathbb{P}_{\Theta'|\Theta}\bra{\theta'|\theta}\\
&=\minentropy.
\end{align*}
Based on \cref{eqn:sml_minentropy}, we can get that
\begin{align*}
\sup_{\paramdistribution}\maxl-\sup_{\secretrv\in \secretvalueset}\log|\mathbf{\Theta}_{\secretrv}| \leq \sml \leq \sup_{\paramdistribution}\maxl.
\end{align*}

\end{proof}

\subsection{Proof of \cref{prop:sml_calculation}}
\label{section:proof_sml_calculation}

\smlexpression*

\begin{proof}
We first analyze the upper bound of \privname under a fixed prior. Let $\secretset$ be the set of distribution parameters whose secret value is $\secretrv$, i.e, $\secretset = \brc{\theta\in \mathbf{\Theta} | \secretofparam = \secretrv}$.
For \privname $\sml$, we have $\sml = \sup_{\paramdistribution}\log \frac{\sum_{\theta'\in\releaseparamset} \sup_{\secretrv\in \secretvalueset} \mathbb{P}_{G\Theta'}\bra{\secretrv,\theta'}}{\sup_{\secretrv\in \secretvalueset} \mathbb{P}_G\bra{\secretrv}}$. Under a fixed prior distribution $\paramdistribution$, we can get that 
\begin{equation}
\label{eqn:pi}
\begin{aligned}
\log&  \frac{\sum_{\theta'\in\releaseparamset} \sup_{\secretrv\in \secretvalueset} \mathbb{P}_{G\Theta'}\bra{\secretrv,\theta'}}{\sup_{\secretrv\in \secretvalueset} \mathbb{P}_G\bra{\secretrv}}\\
&= \log \frac{\sum_{\theta'\in\releaseparamset} \sup_{\secretrv\in \secretvalueset} \mathbb{P}_{\Theta'|G}\bra{\theta'|\secretrv}\mathbb{P}_G\bra{\secretrv}}{\sup_{\secretrv\in \secretvalueset} \mathbb{P}_G\bra{\secretrv}}\\
& \leq \log \frac{\sum_{\theta'\in\releaseparamset} \sup_{\secretrv\in \secretvalueset}  \mathbb{P}_{\Theta'|G}\bra{\theta'|\secretrv}\cdot\sup_{\secretrv\in \secretvalueset}\mathbb{P}_G\bra{\secretrv}}{\sup_{\secretrv\in \secretvalueset} \mathbb{P}_G\bra{\secretrv}}\\
&= \log \sum_{\theta'\in\releaseparamset} \sup_{\secretrv\in \secretvalueset}\mathbb{P}_{\Theta'|G}\bra{\theta'|\secretrv}\\
&= \log\sum_{\theta'\in\releaseparamset} \sup_{\secretrv\in \secretvalueset}
\sum_{\theta\in \mathbf{\Theta}}
\mathbb{P}_{\Theta'|G\Theta}\bra{\theta'|\secretrv, \theta}\cdot\mathbb{P}_{\Theta|G}\bra{\theta|\secretrv}\\
&= \log \hspace{-1mm} \sum_{\theta'\in\releaseparamset} \sup_{\secretrv\in \secretvalueset}
\sum_{\theta\in \mathbf{\Theta}}
\mathbb{P}_{\Theta'|G\Theta}\bra{\theta'|\secretrv, \theta} \frac{\mathbb{P}_{G|\Theta}\bra{\secretrv|\theta}\cdot\mathbb{P}_\Theta\bra{\theta}}{\sum_{\theta\in\mathbf{\Theta}}\mathbb{P}_{G|\Theta}\bra{\secretrv|\theta}\cdot\mathbb{P}_\Theta\bra{\theta}}
\\
&= \log \sum_{\theta'\in\releaseparamset} \sup_{\secretrv\in \secretvalueset} \frac{\sum_{\theta\in \mathbf{\Theta}_{\secretrv}}\mathbb{P}_{\Theta'|\Theta}\bra{\theta'|\theta}\cdot\mathbb{P}_\Theta\bra{\theta}}{\sum_{\theta\in\mathbf{\Theta}_{\secretrv}}\mathbb{P}_\Theta\bra{\theta}}\\
&=\log \sum_{\theta'\in\releaseparamset} \sup_{\secretrv\in \secretvalueset} \mathbb{E}_{\Theta}\brb{\mathbb{P}_{\Theta'|\Theta}\bra{\theta'|\theta}\big|\theta\in \secretset}\\
&\triangleq V_{\paramdistribution}.
\end{aligned}
\end{equation}
\normalsize
This upper bound can be achieved when $\paramdistribution$ satisfies 
\begin{equation}
\label{condition:pri1}
\begin{aligned}
    \forall \theta'\in\releaseparamset: \  \arg\sup_{\secretrv\in \secretvalueset} \mathbb{P}_{\Theta'|G}\bra{\theta'|\secretrv} \cap \arg\sup_{\secretrv\in \secretvalueset}\mathbb{P}_G\bra{\secretrv} \not= \emptyset,
\end{aligned}
\end{equation}
where $\mathbb{P}_{\Theta'|G}\bra{\theta'|\secretrv}=\mathbb{E}\brb{\mathbb{P}_{\Theta'|\Theta}\bra{\theta'|\theta}\big| \theta\in\secretset}$, $\mathbb{P}_{G}\bra{\secretrv}=\sum_{\theta\in\secretset}\mathbb{P}_{\Theta}\bra{\theta}$.

Let ${\paramdistribution^+}\in \arg\sup_{\paramdistribution} V_{\paramdistribution}$ and $\releaseparamset(\secretrv)$ be the set of released parameters satisfying $\secretrv=\arg\sup_{\secretrv\in \secretvalueset} \mathbb{E}_{\Theta}\brb{\mathbb{P}_{\Theta'|\Theta}\bra{\theta'|\theta}\big|\theta\in \secretset}$
under ${\paramdistribution^+}$. %
We can construct a prior $\paramdistribution^*$ as follows: for any $\secretrv\in\secretvalueset$, select a distribution parameter $\theta_\secretrv\in\paramset$ such that $\theta_\secretrv \in \arg\sup_{\theta\in\secretset}  \sum_{\theta'\in\releaseparamset(\secretrv)}\mathbb{P}_{\Theta'|\Theta}\bra{\theta'|\theta}$\footnote{Let 
$\arg\sup_{\theta\in\secretset}  \sum_{\theta'\in\releaseparamset(\secretrv)}\mathbb{P}_{\Theta'|\Theta}\bra{\theta'|\theta} = \secretset$ when $\releaseparamset(\secretrv)=\emptyset$.}, and let $\mathbb{P}_{\Theta}\bra{\theta_\secretrv} = \frac{1}{\brd{\secretvalueset}}$ and $\mathbb{P}_{\Theta}\bra{\theta} = 0, \forall \theta \in \secretset, \theta\not= \theta_\secretrv$.

We can get that
\begin{align*}
V_{{\paramdistribution^+}} &= \log \sum_{\theta'\in\releaseparamset} \sup_{\secretrv\in \secretvalueset} \mathbb{E}_{\Theta}\brb{\mathbb{P}_{\Theta'|\Theta}\bra{\theta'|\theta}\big|\theta\in \secretset}\\
&= \log \sum_{\secretrv\in\secretvalueset}\sum_{\theta'\in\releaseparamset(\secretrv)}\mathbb{E}_{\Theta}\brb{\mathbb{P}_{\Theta'|\Theta}\bra{\theta'|\theta}\big|\theta\in \secretset}\\
&= \log \sum_{\secretrv\in\secretvalueset}\mathbb{E}_{\Theta}\brb{\sum_{\theta'\in\releaseparamset(\secretrv)}\mathbb{P}_{\Theta'|\Theta}\bra{\theta'|\theta}\bigg|\theta\in \secretset}\\
&\leq \log \sum_{\secretrv\in\secretvalueset} \sup_{\theta\in\secretset}\sum_{\theta'\in\releaseparamset(\secretrv)}\mathbb{P}_{\Theta'|\Theta}\bra{\theta'|\theta}\\
&= V_{\paramdistribution^*}.
\end{align*}
Since ${\paramdistribution^+}\in \arg\sup_{\paramdistribution} V_{\paramdistribution}$, we can get that $\paramdistribution^*\in \arg\sup_{\paramdistribution} V_{\paramdistribution}$. Since $\paramdistribution^*$ satisfies condition \ref{condition:pri1} and $\mathbb{P}_{\Theta|G}\bra{\theta|\secretrv}\in\brc{0,1}, \forall \secretrv\in\secretvalueset, \theta\in\paramset$, we have
\begin{equation}
\begin{aligned}
\label{eqn:sml_variations}
\sml &= \sup_{\paramdistribution}\log \frac{\sum_{\theta'\in\releaseparamset} \sup_{\secretrv\in \secretvalueset} \mathbb{P}_{G\Theta'}\bra{\secretrv,\theta'}}{\sup_{\secretrv\in \secretvalueset} \mathbb{P}_G\bra{\secretrv}} \\
&= \sup_{\paramdistribution}\log \sum_{\theta'\in\releaseparamset} \sup_{\secretrv\in \secretvalueset}\mathbb{P}_{\Theta'|G}\bra{\theta'|\secretrv}\\
&= \sup_{\paramdistribution}\log \sum_{\theta'\in\releaseparamset} \sup_{\secretrv\in \secretvalueset} \frac{\sum_{\theta\in \mathbf{\Theta}_{\secretrv}}\mathbb{P}_{\Theta'|\Theta}\bra{\theta'|\theta}\cdot\mathbb{P}_\Theta\bra{\theta}}{\sum_{\theta\in\mathbf{\Theta}_{\secretrv}}\mathbb{P}_\Theta\bra{\theta}}\\
&= \sup_{\paramdistribution}\log \sum_{\theta'\in\releaseparamset} \sup_{\secretrv\in \secretvalueset} \mathbb{E}_{\Theta}\brb{\mathbb{P}_{\Theta'|\Theta}\bra{\theta'|\theta}\big|\theta\in \secretset}\\
&= \sup_{\paramdistribution:\mathbb{P}_{\Theta|G}\in\brc{0,1}}\log \sum_{\theta'\in\releaseparamset} \sup_{\secretrv\in \secretvalueset} \mathbb{P}_{\Theta'|\Theta}\bra{\theta'|\theta_\secretrv}\\
&= \sup_{\mathbb{P}_{\Theta|G}\in\brc{0,1}}\log \sum_{\theta'\in\releaseparamset} \sup_{\secretrv\in \secretvalueset} \mathbb{P}_{\Theta'|\Theta}\bra{\theta'|\theta_\secretrv},
\end{aligned}
\end{equation}
where $\theta_\secretrv$ satisfies $\mathbb{P}_{\Theta|G}\bra{\theta_\secretrv|\secretrv}=1, \forall \secretrv\in\secretvalueset,$ under a prior with $\mathbb{P}_{\Theta|G}\in\brc{0,1}$.
\end{proof}

\section{Proof of \cref{prop:min-cost-flow}}
\label{section:proof_min-cost-flow}
\smlcompdet*

\begin{proof}
Given a deterministic mechanism $\mech$ and a secret mapping $\secretnotation$, let $\mincost$ be the total cost of the min-cost flow under the network we construct.

We first prove that $\log\bra{-\mincost}\geq \sml$ as follows. Let 
$$
\mathbb{P}^*_{\Theta|G} = \arg\sup_{\mathbb{P}_{\Theta|G}\in\brc{0,1}}\log \sum_{\theta'\in\releaseparamset} \sup_{\secretrv\in \secretvalueset} \mathbb{P}_{\Theta'|\Theta}\bra{\theta'|\theta_\secretrv},
$$
and denote $\theta^*_\secretrv$ such that $\mathbb{P}^*_{\Theta|G}\bra{\theta^*_{\secretrv}|\secretrv}=1$. For any deterministic mechanism, we have $\mathbb{P}_{\Theta'|\Theta}\in \brc{0,1}$, and we denote $\theta'_\secretrv$ such that $\mathbb{P}_{\Theta'|\Theta}\bra{\theta'_g|\theta^*_\secretrv}=1$. From \cref{prop:sml_calculation}, we can get that $\sml = \log \sum_{\theta'\in\releaseparamset} \sup_{\secretrv\in \secretvalueset} \mathbb{P}_{\Theta'|\Theta}\bra{\theta'|\theta^*_\secretrv}=\log \sum_{\theta'\in\releaseparamset}\mathbbm{1}\brc{\theta'=\theta'_\secretrv, \exists \secretrv\in\secretvalueset}$. Denote $\releaseparamset^{\bra{*}}=\brc{\theta' | \theta'=\theta'_\secretrv, \exists \secretrv\in\secretvalueset}$, and then we have $\sml  = \log \brd{\releaseparamset^{\bra{*}}}$. For any $\theta'\in \releaseparamset^{\bra{*}}$, there is at least one secret, denoted as $\secretrv_{\theta'}$, such that $\theta'=\theta'_{\secretrv_{\theta'}}$. We can construct a feasible network flow where for any $\theta'\in \releaseparamset^{\bra{*}}$, there is a $1$-unit flow in the path src$-\secretrv_{\theta'}-\theta^*_{\secretrv_{\theta'}}-\theta'-$sink. Since we have $\mathbb{P}_{\Theta'|\Theta}\bra{\theta'|\theta^*_{\secretrv_{\theta'}}}=\mathbb{P}_{\Theta'|\Theta}\bra{\theta'_{\secretrv_{\theta'}}|\theta^*_{\secretrv_{\theta'}}}=1$, we can easily get that the total cost of the constructed flow is $-\brd{\releaseparamset^{\bra{*}}}\geq \mincost$. Therefore, we have $\log\bra{-\mincost}\geq \sml$.

We then prove that $\log\bra{-\mincost}\leq \sml$ as follows. We can construct a min-cost network flow such that the flow of each edge is either $1$ or $0$ \cite{ford2015flows}. Denote $\releaseparamset^{\bra{+}}$ as the set of nodes in $\Theta'$-column that this min-cost flow goes through. Therefore, for any $\theta'\in\releaseparamset^{\bra{+}}$, there is one $1$-unit flow goes through it, and we denote this flow path as src$-\secretrv_{\langle\theta'\rangle}-\theta^{\secretrv_{\langle\theta'\rangle}}-\theta'-$sink. Since the capacity of each edge is $1$, $\theta^{\secretrv_{\langle\theta'\rangle}}$ has a one-on-one correspondence to $\secretrv_{\langle\theta'\rangle}$ and $\theta'$. We can construct a $\mathbb{P}^{+}_{\Theta|G}$ such that $\mathbb{P}^{+}_{\Theta|G}\bra{\theta^{\secretrv_{\langle\theta'\rangle}}|\secretrv_{\langle\theta'\rangle}}=1$, i.e., $\theta^{\secretrv_{\langle\theta'\rangle}}=\theta^{+}_{\secretrv_{\langle\theta'\rangle}}$ under $\mathbb{P}^{+}_{\Theta|G}$. %
Then we can get that $\log\bra{-\mincost} = \log\sum_{\theta'\in\releaseparamset^{\bra{+}}}\mathbb{P}_{\Theta'|\Theta}\bra{\theta'|\theta^{\secretrv_{\langle\theta'\rangle}}}=\log\sum_{\theta'\in\releaseparamset^{\bra{+}}}\mathbb{P}_{\Theta'|\Theta}\bra{\theta'|\theta^{+}_{\secretrv_{\langle\theta'\rangle}}}\leq \sup_{\mathbb{P}_{\Theta|G}\in\brc{0,1}}\log \sum_{\theta'\in\releaseparamset} \sup_{\secretrv\in \secretvalueset} \mathbb{P}_{\Theta'|\Theta}\bra{\theta'|\theta_\secretrv}=\sml$. 

Hence, we have $\log\bra{-\mincost}= \sml$.
\end{proof}
\section{Proof of \cref{prop:np-hard}}
\label{proof:np-hard}
\hardness*

\begin{proof}

We first prove the NP-hardness of the SML computation. 
Consider a 3-set cover decision problem where given a set $\mathcal{U}$, a collection $\mathcal{T}$ of $m$ 3-size subsets whose union is $\mathcal{U}$, and an integer $k$, decide whether there are $k$ subsets in $\mathcal{T}$ whose union equals to $\mathcal{U}$. It is well-known that 3-set cover decision problem is NP-complete \cite{karp2010reducibility}. 
We will prove the NP-hardness of SML computation by reducing the 3-Set Cover decision problem to it.

Let $\mathcal{T}=\brc{T_1, \ldots, T_m}$, where $T_i=\brc{t^{(i)}_1, t^{(i)}_3, t^{(i)}_3}, \forall i\in\brb{m}$, and let $\secretset=\brc{\theta^{\bra{\secretrv}}_1, \ldots, \theta^{\bra{\secretrv}}_{\brd{\secretset}}}, \forall \secretrv\in\secretvalueset$. For any $\mathcal{U}, \mathcal{T}$ and $k$, we can construct the support of distribution parameter $\paramset$, the secret function $\secretnotation$, and the mechanism $\mech$ such that (i) $\forall \secretrv\in\secretvalueset: \brd{\secretset}=m$; (ii) $\brd{\secretvalueset}=k$, i.e., $\secretvalueset=\brc{\secretrv_1,\ldots, \secretrv_k}$; (iii) $\releaseparamset=\mathcal{U}$; (iv) $\forall \secretrv\in\secretvalueset, i\in\brb{m}: \mathbb{P}_{\Theta'|\Theta}\bra{\theta'|\theta^{\bra{\secretrv}}_{i}} = \frac{1}{3}$ if $\theta'\in T_i$, otherwise, $0$.

We show that solving the 3-set cover decision problem with sets $\mathcal{U}, \mathcal{T}$ and integer $k$ is equivalent to deciding whether the SML of the mechanism constructed above is equal to $\log\bra{\brd{\mathcal{U}}/3}$ as follows. 

If there is a sub-collection $\tilde{\mathcal{T}}=\brc{T_{j_1}, \ldots, T_{j_k}}$ whose union is $\mathcal{U}$, we have $\theta'\in \bigcup_{i\in\brb{k}}T_{j_k}, \forall \theta'\in\releaseparamset$. Let $\theta_{\secretrv_i}=\theta^{\bra{\secretrv_i}}_{j_i}, \forall i\in \brb{k}$ ($\theta_\secretrv$ is defined above \cref{prop:sml_calculation}), and we can get that
\begin{align*}
\sml &= \sup_{\mathbb{P}_{\Theta|G}\in\brc{0,1}}\log \sum_{\theta'\in\releaseparamset} \sup_{\secretrv\in \secretvalueset} \mathbb{P}_{\Theta'|\Theta}\bra{\theta'|\theta_\secretrv}\\
&\geq \log \sum_{\theta'\in\releaseparamset} \sup_{i\in\brb{k}} \mathbb{P}_{\Theta'|\Theta}\bra{\theta'|\theta^{\bra{\secretrv_i}}_{j_i}}\\
&= \log \sum_{\theta'\in\releaseparamset} \frac{1}{3}\\
&= \log\frac{\brd{\mathcal{U}}}{3}.
\end{align*}
Additionally, we have
\begin{align*}
\sml &= \sup_{\mathbb{P}_{\Theta|G}\in\brc{0,1}}\log \sum_{\theta'\in\releaseparamset} \sup_{\secretrv\in \secretvalueset} \mathbb{P}_{\Theta'|\Theta}\bra{\theta'|\theta_\secretrv}\\
&\leq\sup_{\mathbb{P}_{\Theta|G}\in\brc{0,1}}\log \sum_{\theta'\in\releaseparamset} \frac{1}{3}\\
&= \log\frac{\brd{\mathcal{U}}}{3}.
\end{align*}
Therefore, we have $\sml = \log\frac{\brd{\mathcal{U}}}{3}$.

If there is no sub-collection of $\mathcal{T}$ with size $k$ that covers $\mathcal{U}$, then $\sml < \log\frac{\brd{\mathcal{U}}}{3}$, which can be proved by contradiction as follows. Suppose $\sml = \log\frac{\brd{\mathcal{U}}}{3}$, then we can construct a prior distribution of parameter with $\mathbb{P}_{\Theta|G}\in\brc{0,1}$, such that $\forall \theta'\in\releaseparamset, \exists i \in \brb{k}: \mathbb{P}_{\Theta'|\Theta}\bra{\theta'|\theta_{\secretrv_i}}=\frac{1}{3}$. For any $i\in\brb{k}$, there exists $j_i\in \brb{m}$ such that $\theta_{\secretrv_i}=\theta^{\bra{\secretrv_i}}_{j_i}$, and we construct a sub-collection of $\mathcal{T}$ as $\tilde{\mathcal{T}}= \brc{T_{j_1}, \ldots, T_{j_k}}$. Since $\releaseparamset=\mathcal{U}$ and $\mathbb{P}_{\Theta'|\Theta}\bra{\theta'|\theta_{\secretrv_i}} = \frac{1}{3}$ if and only if $\theta'\in T_{j_i} \bra{i\in\brb{k}}$, we can get that $\forall u\in \mathcal{U}, \exists i\in\brb{k}: u \in T_{j_i}$. Therefore, $\bigcup_{i\in\brb{k}}T_{j_i} = \mathcal{U}$, which contradicts with the assumption that there is no sub-collection of $\mathcal{T}$ with size $k$ that covers $\mathcal{U}$. Since $\sml\leq \log\frac{\brd{\mathcal{U}}}{3}$, we have $\sml< \log\frac{\brd{\mathcal{U}}}{3}$.

Hence, we can reduce 3-set cover decision problem to SML computation, %
and therefore, SML computation is NP-hard.

Next, we convert the computation of SML to network flow problems, for which approximation or exact algorithms have been proposed.

\begin{figure}[htbp]
\centering
\includegraphics[width=0.8\linewidth]{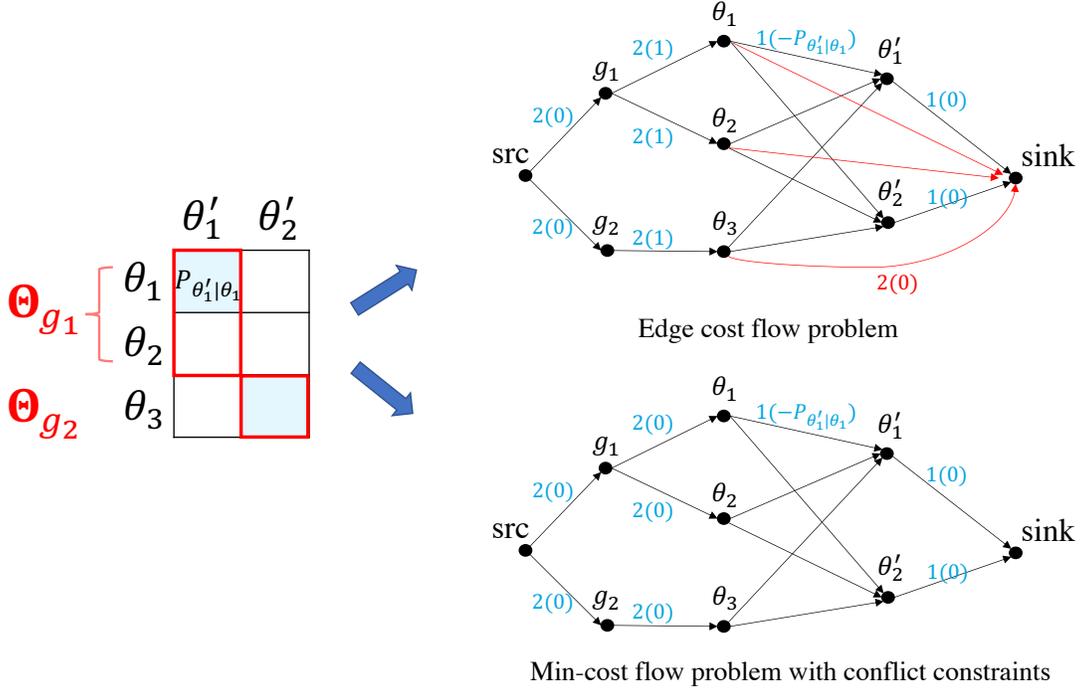}
\caption{Given a mechanism $\mech=\mathbb{P}_{\Theta'|\Theta}$, the left subfigure shows a policy matrix. For each column $j$, the red outlined region indicates rows of parameters with secret $\secretrv$ maximizing $\mathbb{P}_{\Theta'|\Theta}\bra{\theta'_j|\theta_\secretrv}$. The blue cell lies in the row of $\theta_\secretrv$. SML calculation can be converted to an edge cost flow problem (top right) or a min-cost flow problem with conflict constraints (bottom right). Both constructed directed graphs contains three columns of nodes (representing $G, \Theta, \Theta'$ respectively) between the source and sink nodes.}
\label{fig:network_flow}
\end{figure}

\emph{Edge cost flow problem.} 
Edge cost flow problem is a type of min-cost-max-flow problem, where the cost $c$ of an edge satisfies
\begin{align*}
c = 
\begin{cases}
    b & f > 0\\
    0 & f = 0
\end{cases}.
\end{align*}
$f$ is the flow through the edge and $b$ is a constant. This cost function is different from the most common min cost max flow problem, where $c = bf$. 

Edge cost flow problem is shown to be an NP-hard problem \cite[Problem ND32]{lewis1983michael}, %
and calculating SML can be converted as solving an edge cost flow problem.
We illustrate the conversion process in \cref{fig:network_flow}. To construct the network, we introduce the source and sink nodes, and create three columns of nodes between them. 

The first column is $\secretrv$-column, which contains all potential secret values ($\secretrv_1$ and $\secretrv_2$). The capacity of the edge between the source and the $\secretrv$-column node is $\brd{\releaseparamset}$ ($\brd{\releaseparamset}=2$ in \cref{fig:network_flow}), and the cost is $0$. %
Setting the each edge capacity as $\brd{\releaseparamset}$ ensures us to freely select the secret $\secretrv$ for each released parameter value $\theta'\in \releaseparamset$.

The second column is $\theta$-column, which contains all potential parameter value ($\theta_1$, $\theta_2$ and $\theta_3$ in \cref{fig:network_flow}). There will be an edge between node $N_{\secretrv_i}$ and $N_{\theta_j}$ if $\theta_j \in \paramset_{\secretrv_i}$. The capacity of each edge is $\brd{\releaseparamset}$ and the cost is $1$. This ensures that to achieve min-cost goal, all flows from a certain secret $\secretrv$ go to a same node $N_\theta$, which is in line with the SML calculation, i.e., $\mathbb{P}_{\Theta|G}\in\brc{0,1}$.

The third column is the $\theta'$-column, which contains all potential released parameter value ($\theta'_1$ and $\theta'_2$ in \cref{fig:network_flow}).
This column is fully connected with the second column as well as the sink node. To be in line with the privacy calculation, the capacity of the edge between $N_{\theta_i}$ and $N_{\theta'_j}$ column is set as $1$ and the cost is set as $-\mathbb{P}\bra{\theta'_{j}|\theta_i}$. The capacity of the edge between the sink and the third column is $1$ and the cost is $0$, which ensures that the flow that goes through a node in $\theta'$-column is from a single node in $\secretrv$-column. %

Additionally, we fully connect the $\theta$ column and the sink, and set the capacity and cost as $\brd{\releaseparamset}$ and $0$ respectively, which are shown in red edges in \cref{fig:network_flow}.

The max flow of the network is $\brd{\secretvalueset}\cdot \brd{\releaseparamset}$. To achieve the min-cost-max-flow goal, every node in $\theta'$-column has non-zero flow going through it, since only the edge cost between $\theta$ and $\theta'$ columns can be negative. Recall that the flow going through a node in $\theta'$-column is from a single node in $\secretrv$-column, and to achieve the min-cost goal, all flows from a node in $\secretrv$-column will go through only one node in $\theta$-column, denoted as $N_{\theta_{\bra{\secretrv}}}$. Let $\mincost$ be the total cost of the min-cost-max-flow under the network we construct. Similar to the proof of \cref{prop:min-cost-flow}, we can get that
\begin{align*}
    \mincost = \inf_{N_{\theta_{\bra{\secretrv}}}: \theta_{\bra{\secretrv}}\in\secretset, \secretrv\in\secretvalueset} \sum_{\theta'\in\releaseparamset} \inf_{\secretrv\in \secretvalueset} \bra{1-\mathbb{P}\bra{\theta'|\theta_{\bra{\secretrv}}}} = 
    \brd{\releaseparamset}-\sup_{\mathbb{P}_{\Theta|G}\in\brc{0,1}} \sum_{\theta'\in\releaseparamset} \sup_{\secretrv\in \secretvalueset} \mathbb{P}_{\Theta'|\Theta}\bra{\theta'|\theta_\secretrv}= \brd{\releaseparamset}-exp\bra{\sml},
\end{align*}
i.e., $\log\bra{\brd{\releaseparamset}-\mincost}=\sml$. Therefore, we can convert the calculation of SML as the edge cost flow problem.

\cite{krumke1999flow} provides an approximation algorithm for the edge cost flow problem, with approximation ratio $1+\rho$ ($\rho>0$) and running time polynomial in $m$ and $1/\rho$, where $m$ represents the number of edges in the network. Therefore, there exists an approximation algorithm for SML calculation with approximation ratio $1+\rho$ ($\rho>0$) and running time polynomial in $\brd{\Theta}\cdot\brd{\Theta'}$ and $1/\rho$.

\emph{Min-cost flow problem with conflict constraints.} 
We can also convert the calculation process as the min-cost flow problem with conflict constraints, which is also proven to be NP-hard \cite{csuvak2021minimum}. %
The difference of the constructed network with the previous one is that (1) the edge cost between $g$-column and $\theta$-column is $0$. (2) A conflict constraint applies to edges between $g$-column and $\theta$-column: each $g$-node can only send flow to one $\theta$-node. (3) There is no red edges between $\theta$-column and the sink. Similarly, we can also conclude that the calculation of SML can be converted as the min-cost flow problem with conflict constraints.
\cite{csuvak2021minimum} provides an exact solution (with running time exponential to the number of edges) for the min-cost flow problem with conflict constraints.
\end{proof}

\section{Proof of \cref{thm:composition}}
\label{section:proof_composition}
\composition*

\begin{proof}
We first provide \Cref{prop:sml_composed}, which demonstrates how to compute the \privname of an adaptive composition of mechanisms. 

\begin{lemma}
\label{prop:sml_composed}
Suppose there are $\mechnum$ adaptive \datamechanisms, %
and each mechanism $\mech_\indexi$ is defined by $\adaptivemech{\indexi}$, which is adaptively designed given the original data distribution, previous mechanisms $\brc{\mech_k}_{k\in\brb{\indexi-1}}$, and previous outputs $\brc{\Theta'^{(k)}}_{k\in\brb{\indexi-1}}$. Let $\releaseparametervector$ be the sequence of the released distribution parameter, i.e., $\releaseparametervector\in \releaseparamset^{(1)}\times \releaseparamset^{(2)}\times\cdots \releaseparamset^{(\mechnum)}$. The statistic maximal leakage of the adaptively composed mechanism $\boldsymbol{\mech}=\mech_1 \circ \mech_2\circ  \ldots \circ \mech_m$ is
\begin{align*}
&\privacynotation_{\boldsymbol{\mech},\secretnotation} = \sup_{\mathbb{P}_{\Theta|G}\in\brc{0,1}}\log\sum_{\releaseparametervector} \sup_{\secretrv\in \secretvalueset} \prod_{\indexi\in [\mechnum]} \adaptivemech{\indexi}\bra{\theta'^{(\indexi)}|\theta_\secretrv, \brc{\mech_k, \theta'^{(k)}}_{k\in\brb{\indexi-1}}},
\end{align*}
\normalsize
where $\theta_\secretrv$ satisfies $\mathbb{P}_{\Theta|G}\bra{\theta_\secretrv|\secretrv}=1, \forall \secretrv\in\secretvalueset,$ under a prior with $\mathbb{P}_{\Theta|G}\in\brc{0,1}$.
\end{lemma}

\begin{proof}
Based on \cref{eqn:sml_variations}, we have
\begin{align*}
\privacynotation_{\boldsymbol{\mech},\secretnotation} &= \sup_{\paramdistribution} \log\sum_{\releaseparametervector} \sup_{\secretrv\in \secretvalueset}\mathbb{P}_{\theta'^{(1)},\cdots,\theta'^{(\mechnum)}|G}\bra{\boldsymbol{\theta'}|\secretrv}\\
&= \sup_{\paramdistribution}\log \sum_{\releaseparametervector} \sup_{\secretrv\in \secretvalueset}
    \sum_{\theta\in \mathbf{\Theta}}
    \mathbb{P}_{\theta'^{(1)},\cdots,\theta'^{(\mechnum)}|G\Theta}\bra{\releaseparametervector|\secretrv, \theta}\cdot\mathbb{P}_{\Theta|G}\bra{\theta|\secretrv}\\
&= \sup_{\mathbb{P}_{\Theta|G}\in\brc{0,1}} \log\sum_{\releaseparametervector} \sup_{\secretrv\in \secretvalueset} \prod_{\indexi\in [\mechnum]}\adaptivemech{\indexi}\bra{\theta'^{(\indexi)}|\theta_\secretrv, \brc{\mech_k, \theta'^{(k)}}_{k\in\brb{\indexi-1}}}.
\end{align*}
\normalsize

\end{proof}

From that result, we know that the calculation of the \privname{} of adaptively composed mechanism is similar to the process illustrated in \cref{fig:calculation}, except that the policy matrix is changed as $\prod_{\indexi\in [\mechnum]}\adaptivemech{\indexi}$ and the released set of parameters is changed as $\releaseparamset^{(1)}\times \releaseparamset^{(2)}\times\cdots \releaseparamset^{(\mechnum)}$.
Using \cref{prop:sml_composed}, we get that
\begin{align*}
\privacynotation_{\boldsymbol{\mech},\secretnotation} &= \sup_{\mathbb{P}_{\Theta|G}\in\brc{0,1}}\log\sum_{\releaseparametervector} \sup_{\secretrv\in \secretvalueset} \prod_{\indexi\in [\mechnum]}\\
&\quad\quad\quad\adaptivemech{\indexi}\bra{\theta'^{(\indexi)}|\theta_\secretrv, \brc{\mech_k, \theta'^{(k)}}_{k\in\brb{\indexi-1}}}\\
&\quad\leq \sup_{\mathbb{P}_{\Theta|G}\in\brc{0,1}}\log\sum_{\releaseparametervector} \prod_{\indexi\in [\mechnum]}\sup_{\secretrv\in \secretvalueset} \\
&\quad\quad\quad\adaptivemech{\indexi}\bra{\theta'^{(\indexi)}|\theta_\secretrv, \brc{\mech_k, \theta'^{(k)}}_{k\in\brb{\indexi-1}}}\\
&\quad\leq \sup_{\mathbb{P}_{\Theta|G}\in\brc{0,1}}\log\prod_{\indexi\in [\mechnum]}\sum_{\brb{\parameterdistribution'^{(1)},\cdots,\parameterdistribution'^{(i)}}\in \releaseparamset^{(1)}\times\cdots\releaseparamset^{(i)}}\sup_{\secretrv\in \secretvalueset}\\
&\quad\quad\quad\adaptivemech{\indexi}\bra{\theta'^{(\indexi)}|\theta_\secretrv, \brc{\mech_k, \theta'^{(k)}}_{k\in\brb{\indexi-1}}}\\
&\quad\leq \sum_{\indexi\in [\mechnum]}\sup_{\mathbb{P}_{\Theta|G}\in\brc{0,1}}\log \sum_{\brb{\parameterdistribution'^{(1)},\cdots,\parameterdistribution'^{(i)}}\in \releaseparamset^{(1)}\times\cdots\releaseparamset^{(i)}} \sup_{\secretrv\in \secretvalueset}\\
&\quad\quad\quad\adaptivemech{\indexi}\bra{\theta'^{(\indexi)}|\theta_\secretrv, \brc{\mech_k, \theta'^{(k)}}_{k\in\brb{\indexi-1}}}\\
&\quad= \sum_{\indexi \in [\mechnum]} \privacynotation_{\mech_\indexi,\secretnotation}.
\end{align*}
\normalsize
\end{proof}

\section{Proof of \cref{thm:post-processing}}
\label{section:proof_post-processing}
\post*

\begin{proof}
Based on \cref{prop:sml_calculation}, we can get that
\begin{align*}
\privacynotation_{\widetilde{\mech}\circ\mech,\secretnotation}
&= \sup_{\mathbb{P}_{\Theta|G}\in\brc{0,1}}\sum_{\theta''\in\postprocessparamset} \sup_{\secretrv\in \secretvalueset} 
\mathbb{P}_{\Theta''|\Theta}\bra{\theta''|\theta_\secretrv}\\
&= \hspace{-1mm} \sup_{\mathbb{P}_{\Theta|G}\in\brc{0,1}}\sum_{\theta''\in\postprocessparamset} \sup_{\secretrv\in \secretvalueset}\sum_{\parameterdistribution'\in\releaseparamset} 
\mathbb{P}_{\Theta''|\Theta'}\bra{\theta''|\theta'}\cdot\mathbb{P}_{\Theta'|\Theta}\bra{\theta'|\theta_\secretrv}\\
&\leq \hspace{-1mm} \sup_{\mathbb{P}_{\Theta|G}\in\brc{0,1}}\sum_{\theta''\in\postprocessparamset} \sum_{\parameterdistribution'\in\releaseparamset}\mathbb{P}_{\Theta''|\Theta'}\bra{\theta''|\theta'}\cdot\sup_{\secretrv\in \secretvalueset} 
\mathbb{P}_{\Theta'|\Theta}\bra{\theta'|\theta_\secretrv}\\
&= \hspace{-1mm} \sup_{\mathbb{P}_{\Theta|G}\in\brc{0,1}}\sum_{\parameterdistribution'\in\releaseparamset}\sum_{\theta''\in\postprocessparamset}\mathbb{P}_{\Theta''|\Theta'}\bra{\theta''|\theta'}\cdot\sup_{\secretrv\in \secretvalueset} 
\mathbb{P}_{\Theta'|\Theta}\bra{\theta'|\theta_\secretrv}\\
&= \sup_{\mathbb{P}_{\Theta|G}\in\brc{0,1}}\sum_{\parameterdistribution'\in\releaseparamset}\sup_{\secretrv\in \secretvalueset} 
\mathbb{P}_{\Theta'|\Theta}\bra{\theta'|\theta_\secretrv}\\
&= \privacynotation_{\mech,\secretnotation}.
\end{align*}
\normalsize
\end{proof}
\section{Proof of \cref{prop:mech_tradeoff}}
\label{sec:proof_mech_tradeoff}

\privacydistortionrr*

\begin{proof}

 (Privacy and Distortion of Randomized Response)
We first analyze the \privname{} of the Randomized Response with hyperparameter $\epsilon$. %
Based on \cref{prop:sml_calculation}, we can get that
\begin{align*}
\privacyrr &= \sup_{\mathbb{P}_{\Theta|G}\in\brc{0,1}}\log \sum_{\theta'\in\releaseparamset} \sup_{\secretrv\in \secretvalueset} \mathbb{P}_{\Theta'|\Theta}\bra{\theta'|\theta_\secretrv}\\
&= \sup_{\mathbb{P}_{\Theta|G}\in\brc{0,1}}\log \Bigg(\sum_{\theta'\in\brc{\theta_\secretrv}_{\secretrv\in \secretvalueset}} \sup_{\secretrv\in \secretvalueset} \mathbb{P}_{\Theta'|\Theta}\bra{\theta'|\theta_{\secretrv}} +\sum_{\theta'\in 
\releaseparamset\setminus\brc{\theta_\secretrv}_{\secretrv\in \secretvalueset}} \sup_{\secretrv\in \secretvalueset} \mathbb{P}_{\Theta'|\Theta}\bra{\theta'|\theta_{\secretrv}}\Bigg)\\
&= \log\bra{\brd{\secretvalueset}\cdot\frac{e^{\epsilon}}{\brd{\releaseparamset}+e^{\epsilon}-1} + \bra{\brd{\releaseparamset}-\brd{\secretvalueset}}\cdot\frac{1}{\brd{\releaseparamset}+e^{\epsilon}-1}}. 
\end{align*}
\normalsize

Denote $\rrratio = \frac{e^{\epsilon}-1}{\brd{\releaseparamset}}$.
Since $\brd{\secretvalueset}=\secretnum$ and $\brd{\releaseparamset}=\brd{\paramset} = \binom{\samplenum+\categoryleast-1}{\categoryleast-1}$, we have 
$$\privacyrr = \log\bra{\frac{\binom{\samplenum+\categoryleast-1}{\categoryleast-1}+\secretnum\bra{e^{\epsilon}-1}}{\binom{\samplenum+\categoryleast-1}{\categoryleast-1}+e^{\epsilon}-1}} = \log\bra{\frac{1+\secretnum\rrratio}{1+\rrratio}}.
$$
\normalsize

We then analyze the distortion of Randomized Response. For the categorical distribution with parameter $\parameterdistribution$, denote $\mass{\parameterdistribution_{i}}$ as the probability mass of the $i$-th category, where $i\in \brb{\categoryleast}$. 
From \cref{lemma:distortion}, we have $\distortionnotation = \sup_{\theta} \mathbb{E}_{\theta'}[\distanceof{\privatedistribution} {\releasedistribution}]$. Let $\parameterdistribution^* = \arg\sup_{\theta} \mathbb{E}_{\theta'}[\distanceof{\privatedistribution} {\releasedistribution}]$, we first prove by contradiction that $\parameterdistribution^*$ satisfies that $\exists i\in \brb{\categoryleast}:~ \mass{\parameterdistribution^*_i} = 1$ as follows.

Suppose $\max_{i\in\brb{\categoryleast}}\mass{\parameterdistribution^*_i}= \mass{\parameterdistribution^*_j}<1$, then we have $\mass{\parameterdistribution^*_j}\leq 1-\frac{1}{\samplenum}$ and there exist another category $k\not= j$ such that $\mass{\parameterdistribution^*_j}\geq \mass{\parameterdistribution^*_k}\geq\frac{1}{\samplenum}$. We can construct another categorical distribution parameter $\tilde{\parameterdistribution}$ where $\mass{\tilde{\parameterdistribution}_j} = \mass{\parameterdistribution^*_j} + \frac{1}{\samplenum}, \mass{\tilde{\parameterdistribution}_k} = \mass{\parameterdistribution^*_k} - \frac{1}{\samplenum}$, and $\mass{\tilde{\parameterdistribution}_i} = \mass{\parameterdistribution^*_i}, \forall i\in\brb{\samplenum}\setminus\brc{j,k}$. We can get that 
\begin{align*}
&\distanceof{\distributionof{X_{\parameterdistribution^*}}} {\releasedistribution} = 
\begin{cases}
\distanceof{\distributionof{X_{\tilde{\parameterdistribution}}}} {\releasedistribution}+\frac{2}{\samplenum}, & \mass{\parameterdistribution'_{j}} \geq \mass{\tilde{\parameterdistribution}_j}, \mass{\parameterdistribution'_{k}} \leq \mass{\tilde{\parameterdistribution}_k},\\
\distanceof{\distributionof{X_{\tilde{\parameterdistribution}}}} {\releasedistribution}-\frac{2}{\samplenum}, & \mass{\parameterdistribution'_{j}} < \mass{\tilde{\parameterdistribution}_j}, \mass{\parameterdistribution'_{k}} > \mass{\tilde{\parameterdistribution}_k},\\
\distanceof{\distributionof{X_{\tilde{\parameterdistribution}}}} {\releasedistribution}, & \text{otherwise}.\\
\end{cases}
\end{align*}
For every $\parameterdistribution' \in \releaseparamset$ satisfying $\mass{\parameterdistribution'_{j}} \geq \mass{\tilde{\parameterdistribution}_j}, \mass{\parameterdistribution'_{k}} \leq \mass{\tilde{\parameterdistribution}_k}$, we can construct a unique $\parameterdistribution{''}$ as $\mass{\parameterdistribution''_{j}} = \mass{\parameterdistribution'_{k}}, \mass{\parameterdistribution''_{k}} = \mass{\parameterdistribution'_{j}}$, and $\mass{\parameterdistribution''_{i}} = \mass{\parameterdistribution'_{i}}, \forall i\in\brb{\samplenum}\setminus\brc{j,k}$. Since $\mass{\parameterdistribution''_{j}} = \mass{\parameterdistribution'_{k}}\leq \mass{\tilde{\parameterdistribution}_k} < \mass{\parameterdistribution^*_j}$ and $\mass{\parameterdistribution''_{k}} = \mass{\parameterdistribution'_{j}}\geq \mass{\tilde{\parameterdistribution}_j} > \mass{\parameterdistribution^*_k}$, we have $\distanceof{\distributionof{X_{\parameterdistribution^*}}} {\distributionof{Y_{\parameterdistribution''}}} = \distanceof{\distributionof{X_{\tilde{\parameterdistribution}}}} {\distributionof{Y_{\parameterdistribution''}}}-\frac{2}{\samplenum}$. Therefore, we can get that $\mathbb{E}_{\theta'}[\distanceof{\distributionof{X_{\parameterdistribution^*}}} {\releasedistribution}] \leq \mathbb{E}_{\theta'}[\distanceof{\distributionof{X_{\tilde{\parameterdistribution}}}} {\releasedistribution}]$, which contradicts with the definition of $\parameterdistribution^*: \parameterdistribution^*=\arg\sup_{\theta} \mathbb{E}_{\theta'}[\distanceof{\privatedistribution} {\releasedistribution}]$. Therefore, we can get that $\exists i\in \brb{\categoryleast}:~ \mass{\parameterdistribution^*_i} = 1$. Without loss of generality, let $\mass{\parameterdistribution^*_1} = 1$.

We can calculate the distortion as
\begin{align*}
\distortionrr 
&= \mathbb{E}_{\theta'=\mechrr\bra{\parameterdistribution^*}}[\distanceof{\distributionof{X_{\parameterdistribution^*}}} {\releasedistribution}]\\
&= \frac{\brd{\releaseparamset}}{\brd{\releaseparamset}+e^{\epsilon}-1}\cdot\sum_{i\in\brb{\samplenum}} \frac{i}{\samplenum}\cdot \probof{\mass{\parameterdistribution'_{1}}=1-\frac{i}{\samplenum}}\\
&= \frac{\brd{\releaseparamset}}{\brd{\releaseparamset}+e^{\epsilon}-1}\cdot \sum_{i\in\brb{\samplenum}} \frac{i}{\samplenum}\cdot \frac{\binom{i+\categoryleast-2}{\categoryleast-2}}{\brd{\releaseparamset}}\\
&= \frac{1}{\samplenum\binom{\samplenum+\categoryleast-1}{\categoryleast-1}+\samplenum\bra{e^{\epsilon}-1}}\cdot\sum_{i\in\brb{\samplenum}} i\cdot\binom{i+\categoryleast-2}{\categoryleast-2}\\
&=\frac{1}{\samplenum\binom{\samplenum+\categoryleast-1}{\categoryleast-1}+\samplenum\bra{e^{\epsilon}-1}}\cdot\sum_{i\in\brb{\samplenum}} \bra{\categoryleast-1}\cdot\binom{i+\categoryleast-2}{\categoryleast-1}\\
&= \frac{\bra{\categoryleast-1}\cdot\binom{\samplenum+\categoryleast-1}{\categoryleast}}{\samplenum\binom{\samplenum+\categoryleast-1}{\categoryleast-1}+\samplenum\bra{e^{\epsilon}-1}}\\
&= \frac{\categoryleast-1}{\categoryleast\bra{1+\rrratio}}.
\end{align*}

(Privacy and Distortion of Quantization Mechanism) For the privacy of the quantization mechanism, %
we can get that
\begin{align*}
\privacyquan &= \sup_{\mathbb{P}_{\Theta|G}\in\brc{0,1}}\log \sum_{\theta'\in\releaseparamset} \sup_{\secretrv\in \secretvalueset} \mathbb{P}_{\Theta'|\Theta}\bra{\theta'|\theta_\secretrv}\\
&= \sup_{\mathbb{P}_{\Theta|G}\in\brc{0,1}} \hspace{-1mm} \log\hspace{-1mm} \sum_{k \in \brc{0, 1, \cdots, \left\lceil \frac{\secretnum}{\intervallen}\right\rceil-1}}\sum_{\theta'\in\releaseparamset_{\midpoint{k}}} \sup_{\secretrv\in \secretvalueset} \mathbb{P}_{\Theta'|\Theta}\bra{\theta'|\theta_{\secretrv}}\\
&= \sup_{\mathbb{P}_{\Theta|G}\in\brc{0,1}} \log\sum_{k \in \brc{0, 1, \cdots, \left\lceil \frac{\secretnum}{\intervallen}\right\rceil-1}}\sum_{\theta'\in\releaseparamset_{\midpoint{k}}} \frac{1}{\brd{\releaseparamset_{\midpoint{k}}}}\\
&= \log\ceil{ \frac{\secretnum}{\intervallen}}.
\end{align*}

For the distortion analysis, we consider the case where the secret is the fraction of a category, i.e., $\secretofparam =\mass{\parameterdistribution_i}, i\in \brb{\categoryleast}$. Without loss of generality, let $\secretofparam = \mass{\parameterdistribution_1}$. The number of possible secret values is $\secretnum = \samplenum+1$ as $\secretnum \in \brc{0, 1/\samplenum, 2/\samplenum, \cdots, 1}$, and let $\secretrv_{l}=\frac{l-1}{\samplenum}, \forall l\in\brb{\samplenum+1}$. Let $\parameterdistribution^* = \arg\sup_{\theta} \mathbb{E}_{\theta'}[\distanceof{\privatedistribution} {\releasedistribution}]$, with the same distortion analysis of Randomized Response, we can get that $\parameterdistribution^*$ satisfies that $\exists j\in\brb{\categoryleast}\setminus\brc{1}:~ \mass{\parameterdistribution^*_j} = 1-\mass{\parameterdistribution^*_1}$. Without loss of generality, let $\mass{\parameterdistribution^*_2} = 1-\mass{\parameterdistribution^*_1}$. %
Let $\tilde{\parameterdistribution}$ be the distribution parameter satisfying $\mass{\tilde{\parameterdistribution}_1} = \secretrv_{k\intervallen+j}$, where $k \in \brc{0, 1, \cdots, \left\lceil \frac{\secretnum}{\intervallen}\right\rceil-1}, j\in \brb{\intervallen}$, and $\mass{\tilde{\parameterdistribution}_2} = 1 - \mass{\tilde{\parameterdistribution}_1}$. We can get that
\begin{align*}
\mathbb{E}_{\theta'=\mechquan\bra{\tilde{\parameterdistribution}}}&[\distanceof{\distributionof{X_{\tilde{\parameterdistribution}}}} {\releasedistribution}]\\
&= \frac{1}{2}\brb{\frac{\brd{\midpoint{k}-k\intervallen-j}}{\samplenum} +  \sum_{i\in\brb{\samplenum+1}\setminus\brb{\midpoint{k}}}\frac{\brd{i-k\intervallen-j}+i-\midpoint{k}-1}{\samplenum}\probof{\mass{\parameterdistribution'_2}=1-\frac{i-1}{\samplenum}}}\\
&\overset{1}{\leq} \frac{1}{2}\brb{\frac{\brd{\midpoint{k}-k\intervallen-1}}{\samplenum} +  \sum_{i\in\brb{\samplenum+1}\setminus\brb{\midpoint{k}}}\frac{\brd{i-k\intervallen-1}+i-\midpoint{k}-1}{\samplenum}\probof{\mass{\parameterdistribution'_2}=1-\frac{i-1}{\samplenum}}}\\
&= \frac{\brd{\midpoint{k}-k\intervallen-1}}{2\samplenum} + \sum_{i\in\brb{\samplenum+1}\setminus\brb{\midpoint{k}}}\frac{\brd{i-k\intervallen-1}+i-\midpoint{k}-1}{2\samplenum}\frac{\binom{i-\midpoint{k}+\categoryleast-3}{\categoryleast-3}}{\binom{\samplenum+1-\midpoint{k}+\categoryleast-2}{\categoryleast-2}}\\
&\overset{2}{\leq} \frac{\brd{\midpoint{0}-1}}{2\samplenum} + \sum_{i\in\brb{\samplenum+1}\setminus\brb{\midpoint{0}}}\frac{\brd{i-1}+i-\midpoint{0}-1}{2\samplenum}\frac{\binom{i-\midpoint{0}+\categoryleast-3}{\categoryleast-3}}{\binom{\samplenum+1-\midpoint{0}+\categoryleast-2}{\categoryleast-2}}\\
&= \frac{1}{2\samplenum}\floor{\frac{\intervallen}{2}} + \frac{1}{2\samplenum\binom{\samplenum-\floor{\frac{\intervallen}{2}}+\categoryleast-2}{\categoryleast-2}}\sum_{l\in\brb{n-\floor{\frac{\intervallen}{2}}}\cup\brc{0}} \bra{\floor{\frac{\intervallen}{2}}+2l-1}\binom{l+\categoryleast-3}{\categoryleast-3}\\
&= \frac{1}{2\samplenum}\floor{\frac{\intervallen}{2}} + \frac{1}{2\samplenum\binom{\samplenum-\floor{\frac{\intervallen}{2}}+\categoryleast-2}{\categoryleast-2}}\cdot \Bigg[\floor{\frac{\intervallen}{2}}\cdot \binom{\samplenum-\floor{\frac{\intervallen}{2}}+\categoryleast-2}{\categoryleast-2}+\sum_{l\in\brb{n-\floor{\frac{\intervallen}{2}}}}\bra{\categoryleast-2}\binom{l+\categoryleast-3}{\categoryleast-2}\Bigg]\\
&= \frac{1}{\samplenum}\cdot\floor{\frac{\intervallen}{2}} + \frac{\bra{\categoryleast-2}\binom{\samplenum-\floor{\frac{\intervallen}{2}}+\categoryleast-2}{\categoryleast-1}}{2\samplenum\binom{\samplenum-\floor{\frac{\intervallen}{2}}+\categoryleast-2}{\categoryleast-2}}\\
&= \frac{1}{2}+\frac{\categoryleast\floor{\frac{\intervallen}{2}}-\samplenum}{2\samplenum\bra{\categoryleast-1}},
\end{align*}
where $\overset{1}{\leq}$ achieves `$=$' when $j=1$, and $\overset{2}{\leq}$ achieves `$=$' when $k=0$. Therefore, we can get that $$\distortionquan\hspace{-1mm}  
= \mathbb{E}_{\theta'=\mechquan\bra{\parameterdistribution^*}}[\distanceof{\distributionof{X_{\parameterdistribution^*}}} {\releasedistribution}] = \frac{1}{2}+\frac{\categoryleast\floor{\frac{\intervallen}{2}}-\samplenum}{2\samplenum\bra{\categoryleast-1}}.$$

(Mechanism Comparison) Under the case where the secret is the fraction of a category, we compare the distortion of quantization and Randomized Response when they achieve the same non-trivial \privname{}, i.e., $\privacyquan=\privacyrr<\log \secretnum$. Under the same \privname{}, we have
$
\log\frac{1+\secretnum\rrratio}{1+\rrratio} = \log\ceil{ \frac{\secretnum}{\intervallen}}
$, and therefore, 
\begin{align*}
\rrratio = \frac{\secretnum-1}{\secretnum-\ceil{ \frac{\secretnum}{\intervallen}}}-1 \leq \frac{\secretnum-1}{\secretnum-\frac{\secretnum}{\intervallen}-1}-1.
\end{align*}
For the distortion of Randomized Response, we have 
\begin{align*}
\distortionrr = \frac{\categoryleast-1}{\categoryleast\bra{1+\rrratio}} \geq \frac{\bra{\intervallen\secretnum-\secretnum-\intervallen}\bra{\categoryleast-1}}{\intervallen\categoryleast\bra{\secretnum-1}}.
\end{align*}

Since $\secretnum=\samplenum+1$, we can get that 
\begin{align*}
\frac{\distortionrr}{\distortionquan} 
&\geq \frac{\bra{\intervallen\secretnum-\secretnum-\intervallen}\bra{\categoryleast-1}}{\intervallen\categoryleast\bra{\secretnum-1}} \Big / \bra{\frac{1}{2}+\frac{\categoryleast\intervallen-2\samplenum}{4\samplenum\bra{\categoryleast-1}}}\\
&= \frac{4\bra{\categoryleast-1}^2\bra{\intervallen\samplenum-\samplenum-1}}{\brb{\categoryleast\intervallen+2\samplenum\bra{\categoryleast-2}}\intervallen\categoryleast}.
\end{align*}

When $\lim_{\samplenum\rightarrow\infty}\intervallen=\infty$, since $\intervallen\leq \secretnum=\samplenum+1$, we have
\begin{align*}
\lim_{\samplenum\rightarrow\infty} \frac{\distortionrr}{\distortionquan} &\geq \frac{4\samplenum\bra{\categoryleast-1}^2}{\brb{\categoryleast\intervallen+2\samplenum\bra{\categoryleast-2}}\categoryleast}\\
&\geq \frac{4\bra{\categoryleast-1}^2}{\bra{3\categoryleast-4}\categoryleast}\\
&= 1+\frac{\bra{\categoryleast-2}^2}{3\bra{\categoryleast}^2-4\categoryleast}\\
&\geq 1.
\end{align*}

When $\lim_{\samplenum\rightarrow\infty}\intervallen<\infty$, to achieve non-trivial \privname{}, $\intervallen$ should satisfy $\log\ceil{\frac{\secretnum}{\intervallen}}<\log \secretnum$, i.e., $\intervallen>1$. We have
\begin{align*}
\lim_{\samplenum\rightarrow\infty} \frac{\distortionrr}{\distortionquan} &\geq \frac{4\samplenum\bra{\categoryleast-1}^2}{2\samplenum\bra{\categoryleast-2}\categoryleast}\cdot \frac{\intervallen-1}{\intervallen}\\
&\geq \frac{\bra{\categoryleast-1}^2}{\categoryleast\bra{\categoryleast-2}}\\
&= 1+\frac{1}{\categoryleast\bra{\categoryleast-2}}\\
&> 1.
\end{align*}

Hence, we can get that when secret is the fraction of a category, for any non-trivial privacy budget $T<\log\secretnum= \log\bra{\samplenum+1}$, when $\privacyquan=\privacyrr\leq T$, we have 
$$
\lim_{\samplenum\rightarrow\infty} \frac{\distortionrr}{\distortionquan} \geq 1.
$$

\end{proof}
\section{Proof of \cref{prop:robustness_tabular}}
\label{section:proof_robustness_tabular}

\robustnesstabular*

\begin{proof}

Denote $\releaseparamsetall$ as the set that contains all the distribution parameters whose support is or is within $\combinationsetestimate\cup\combinationsetall$, and $\releaseparamsetori$ as the set of parameters whose support is or is within $\combinationsetestimate$. Let $\beta\bra{\parameterdistribution'}$ %
be the number of categories in $\combinationsetall\setminus\combinationsetestimate$ that have non-zero probability mass in $\parameterdistribution'$, i.e., $\beta\bra{\parameterdistribution'} = \brd{\support{\parameterdistribution'}\setminus\combinationsetestimate}$. We analyze the robustness of \rr{} and \qm{} to support mismatch correspondingly as follows.

\subsubsection{Randomized Response}

For Randomized Response with hyperparameter $\epsilon$, from \cref{prop:mech_tradeoff}, we know that
\begin{align*}
\privacyrr^* = \log\frac{\binom{\samplenum+\categoryactual-1}{\categoryactual-1}+\secretnum\bra{e^{\epsilon}-1}}{\binom{\samplenum+\categoryactual-1}{\categoryactual-1}+e^{\epsilon}-1}
\geq \log\frac{\binom{\samplenum+\categoryactual+\categorynoise-1}{\categoryactual+\categorynoise-1}+\secretnum\bra{e^{\epsilon}-1}}{\binom{\samplenum+\categoryactual+\categorynoise-1}{\categoryactual+\categorynoise-1}+e^{\epsilon}-1}.
\end{align*}

For $\privacyrr$, we can get that
\begin{align*}
\privacyrr&= \sup_{\mathbb{P}_{\Theta|G}\in\brc{0,1}}  \log\sum_{\theta'\in\releaseparamsetall} \sup_{\secretrv\in \secretvalueset} \mathbb{P}_{\Theta'|\Theta}\bra{\theta'|\theta_{\secretrv}}\\
&= \sup_{\mathbb{P}_{\Theta|G}\in\brc{0,1}}  \log\Bigg(\sum_{\theta'\in\brc{\theta_\secretrv}_{\secretrv\in \secretvalueset}}\sup_{\secretrv\in \secretvalueset} \mathbb{P}_{\Theta'|\Theta}\bra{\theta'|\theta_{\secretrv}}+\sum_{\theta'\in 
\releaseparamsetall\setminus\brc{\theta_\secretrv}_{\secretrv\in \secretvalueset}} \sup_{\secretrv\in \secretvalueset} \mathbb{P}_{\Theta'|\Theta}\bra{\theta'|\theta_{\secretrv}}\Bigg) \\
&\leq \sup_{\mathbb{P}_{\Theta|G}\in\brc{0,1}}  \log\Bigg(\sum_{\theta'\in \brc{\theta_\secretrv}_{\secretrv\in \secretvalueset}} \frac{ e^{\epsilon}}{\binom{\samplenum+\categoryleast+\categorynoise+\beta\bra{\parameterdistribution'}-1}{\categoryleast+\categorynoise+\beta\bra{\parameterdistribution'}-1}+e^{\epsilon}-1}+\sum_{\theta'\in 
\releaseparamsetall\setminus\brc{\theta_\secretrv}_{\secretrv\in \secretvalueset}} \frac{1}{\binom{\samplenum+\categoryleast+\categorynoise+\beta\bra{\parameterdistribution'}-1}{\categoryleast+\categorynoise+\beta\bra{\parameterdistribution'}-1}+e^{\epsilon}-1}\Bigg)\\
&\leq \sup_{\mathbb{P}_{\Theta|G}\in\brc{0,1}}   \log\Bigg(\frac{\secretnum \bra{e^{\epsilon}-1}}{\binom{\samplenum+\categoryleast+\categorynoise-1}{\categoryleast+\categorynoise-1}+e^{\epsilon}-1}+\sum_{\theta'\in 
\releaseparamsetall} \frac{1}{\binom{\samplenum+\categoryleast+\categorynoise+\beta\bra{\parameterdistribution'}-1}{\categoryleast+\categorynoise+\beta\bra{\parameterdistribution'}-1}+e^{\epsilon}-1}\Bigg)\\
&= \sup_{\mathbb{P}_{\Theta|G}\in\brc{0,1}}   \log\Bigg(\frac{\secretnum \bra{e^{\epsilon}-1}}{\binom{\samplenum+\categoryleast+\categorynoise-1}{\categoryleast+\categorynoise-1}+e^{\epsilon}-1}+\sum_{i\in\brc{0}\cup\brb{\categoryactual-\categoryleast}} \frac{\binom{\categoryactual-\categoryleast}{i}\binom{\samplenum-i+\categoryleast+\categorynoise+i-1}{\categoryleast+\categorynoise+i-1}}{\binom{\samplenum+\categoryleast+\categorynoise+i-1}{\categoryleast+\categorynoise+i-1}+e^{\epsilon}-1}\Bigg)\\
&\leq \sup_{\mathbb{P}_{\Theta|G}\in\brc{0,1}}   \log\Bigg(\frac{\secretnum \bra{e^{\epsilon}-1}}{\binom{\samplenum+\categoryleast+\categorynoise-1}{\categoryleast+\categorynoise-1}+e^{\epsilon}-1}+\frac{\binom{\samplenum+\categoryactual+\categorynoise-1}{\categoryactual+\categorynoise-1}}{\binom{\samplenum+\categoryactual+\categorynoise-1}{\categoryactual+\categorynoise-1}+e^{\epsilon}-1}\sum_{i\in\brc{0}\cup\brb{\categoryactual-\categoryleast}} \frac{\binom{\categoryactual-\categoryleast}{i}\binom{\samplenum-i+\categoryleast+\categorynoise+i-1}{\categoryleast+\categorynoise+i-1}}{\binom{\samplenum+\categoryleast+\categorynoise+i-1}{\categoryleast+\categorynoise+i-1}}\Bigg)\\
&= \sup_{\mathbb{P}_{\Theta|G}\in\brc{0,1}}  \log \Bigg(\frac{\secretnum \bra{e^{\epsilon}-1}}{\binom{\samplenum+\categoryleast+\categorynoise-1}{\categoryleast+\categorynoise-1}+e^{\epsilon}-1}+\hspace{-0.5mm}\frac{\binom{\samplenum+\categoryactual+\categorynoise-1}{\categoryactual+\categorynoise-1}}{\binom{\samplenum+\categoryactual+\categorynoise-1}{\categoryactual+\categorynoise-1}+e^{\epsilon}-1}\sum_{i\in\brc{0}\cup\brb{\categoryactual-\categoryleast}}\hspace{-1mm} \frac{\samplenum!\bra{\samplenum+\categoryleast+\categorynoise-1}!\binom{\categoryactual-\categoryleast}{i}}{\bra{\samplenum-i}!\bra{\samplenum+i+\categoryleast+\categorynoise-1}!}\Bigg)\\
&\leq \sup_{\mathbb{P}_{\Theta|G}\in\brc{0,1}}  \log \Bigg(\frac{\secretnum \bra{e^{\epsilon}-1}}{\binom{\samplenum+\categoryleast+\categorynoise-1}{\categoryleast+\categorynoise-1}+e^{\epsilon}-1}+\frac{\binom{\samplenum+\categoryactual+\categorynoise-1}{\categoryactual+\categorynoise-1}}{\binom{\samplenum+\categoryactual+\categorynoise-1}{\categoryactual+\categorynoise-1}+e^{\epsilon}-1}\sum_{i\in\brc{0}\cup\brb{\categoryactual-\categoryleast}} \bra{\frac{\samplenum}{\samplenum+\categoryleast+\categorynoise-1}}^i \binom{\categoryactual-\categoryleast}{i}\Bigg)\\
&= \log \Bigg(\frac{\secretnum \bra{e^{\epsilon}-1}}{\binom{\samplenum+\categoryleast+\categorynoise-1}{\categoryleast+\categorynoise-1}+e^{\epsilon}-1}+\frac{\binom{\samplenum+\categoryactual+\categorynoise-1}{\categoryactual+\categorynoise-1}}{\binom{\samplenum+\categoryactual+\categorynoise-1}{\categoryactual+\categorynoise-1}+e^{\epsilon}-1} \bra{1+\frac{\samplenum}{\samplenum+\categoryleast+\categorynoise-1}}^{\categoryactual-\categoryleast}\Bigg). 
\end{align*}

When $e^{\epsilon}-1\leq \binom{\samplenum+\categoryleast+\categorynoise-1}{\categoryleast+\categorynoise-1} / \secretnum$, we can get that
\begin{align*}
\privacyrr&-\privacyrr^* \\
&= \log \Bigg(\frac{\binom{\samplenum+\categoryactual+\categorynoise-1}{\categoryactual+\categorynoise-1}+e^{\epsilon}-1}{\binom{\samplenum+\categoryactual+\categorynoise-1}{\categoryactual+\categorynoise-1}+\secretnum\bra{e^{\epsilon}-1}}\frac{\secretnum \bra{e^{\epsilon}-1}}{\binom{\samplenum+\categoryleast+\categorynoise-1}{\categoryleast+\categorynoise-1}+e^{\epsilon}-1}\\
&\hspace{10mm}+\frac{\binom{\samplenum+\categoryactual+\categorynoise-1}{\categoryactual+\categorynoise-1}+e^{\epsilon}-1}{\binom{\samplenum+\categoryactual+\categorynoise-1}{\categoryactual+\categorynoise-1}+\secretnum\bra{e^{\epsilon}-1}}\frac{\binom{\samplenum+\categoryactual+\categorynoise-1}{\categoryactual+\categorynoise-1}}{\binom{\samplenum+\categoryactual+\categorynoise-1}{\categoryactual+\categorynoise-1}+e^{\epsilon}-1} \bra{1+\frac{\samplenum}{\samplenum+\categoryleast+\categorynoise-1}}^{\categoryactual-\categoryleast}\Bigg)\\
&\leq \log \bra{\frac{\binom{\samplenum+\categoryactual+\categorynoise-1}{\categoryactual+\categorynoise-1}+e^{\epsilon}-1}{\binom{\samplenum+\categoryactual+\categorynoise-1}{\categoryactual+\categorynoise-1}+\secretnum\bra{e^{\epsilon}-1}}\frac{\secretnum \bra{e^{\epsilon}-1}}{\binom{\samplenum+\categoryleast+\categorynoise-1}{\categoryleast+\categorynoise-1}+e^{\epsilon}-1}+\bra{1+\frac{\samplenum}{\samplenum+\categoryleast+\categorynoise-1}}^{\categoryactual-\categoryleast}}\\
&\leq \log \bra{1+\bra{1+\frac{\samplenum}{\samplenum+\categoryleast+\categorynoise-1}}^{\categoryactual-\categoryleast}}\\
&< \log 3 \cdot \bra{\categoryactual-\categoryleast},
\end{align*}
which indicates that Randomized Response is $\log 3$-robust when $\epsilon$ satisfies $e^{\epsilon}-1 
\leq \binom{\samplenum+\categoryestimate-1}{\categoryestimate-1} / \secretnum$.

\subsubsection{Quantization Mechanism}

For quantization mechanism with interval length $\intervallen$, from \cref{prop:mech_tradeoff}, we know that 
\begin{align*}
\privacyquan^* = \log\ceil{\frac{\secretnum}{\intervallen}}.
\end{align*}

For convenience, let $\secretrv_{k,\intervallen}$ represent $\secretrv_{\bra{k+\frac{1}{2}}\intervallen}$, the output secret value of the $k$-th interval.
For $\privacyquan$, we have
\begin{align*}
\privacyquan &= \sup_{\mathbb{P}_{\Theta|G}\in\brc{0,1}}  \log\sum_{\theta'\in\releaseparamsetall} \sup_{\secretrv\in \secretvalueset} \mathbb{P}_{\Theta'|\Theta}\bra{\theta'|\theta_{\secretrv}}\\
&= \sup_{\mathbb{P}_{\Theta|G}\in\brc{0,1}}  \log \sum_{k \in \brc{0, 1, \cdots, \left\lceil \frac{\secretnum}{\intervallen}\right\rceil-1}}\sum_{\theta'\in\releaseparamsetall_{\midpoint{k}}} \sup_{\secretrv\in \secretvalueset} \mathbb{P}_{\Theta'|\Theta}\bra{\theta'|\theta_{\secretrv}}\\
&\leq \sup_{\mathbb{P}_{\Theta|G}\in\brc{0,1}}  \log \sum_{k \in \brc{0, 1, \cdots, \left\lceil \frac{\secretnum}{\intervallen}\right\rceil-1}}\sum_{\theta'\in\releaseparamsetall_{\midpoint{k}}} \frac{1}{\binom{\samplenum\bra{1-\secretrv_{k,\intervallen}}+\categoryleast+\categorynoise+\beta\bra{\parameterdistribution'}-2}{\categoryleast+\categorynoise+\beta\bra{\parameterdistribution'}-2}}\\
&= \sup_{\mathbb{P}_{\Theta|G}\in\brc{0,1}}  \log \sum_{k \in \brc{0, 1, \cdots, \left\lceil \frac{\secretnum}{\intervallen}\right\rceil-1}}\sum_{i\in\brc{0,1,\cdots,\categoryactual-\categoryleast}}\frac{\binom{\categoryactual-\categoryleast}{i}\binom{\samplenum\bra{1-\secretrv_{k,\intervallen}}-i+\categoryleast+\categorynoise+i-2}{\categoryleast+\categorynoise+i-2}}{\binom{\samplenum\bra{1-\secretrv_{k,\intervallen}}+\categoryleast+\categorynoise+i-2}{\categoryleast+\categorynoise+i-2}}\\
&= \sup_{\mathbb{P}_{\Theta|G}\in\brc{0,1}}  \log \sum_{k \in \brc{0, 1, \cdots, \left\lceil \frac{\secretnum}{\intervallen}\right\rceil-1}}\sum_{i\in\brc{0,1,\cdots,\categoryactual-\categoryleast}}\\
&\hspace{15mm}\frac{\bra{\samplenum\bra{1-\secretrv_{k,\intervallen}}}!\bra{\samplenum\bra{1-\secretrv_{k,\intervallen}}+\categoryleast+\categorynoise-2}!}{\bra{\samplenum\bra{1-\secretrv_{k,\intervallen}}-i}!\bra{\samplenum\bra{1-\secretrv_{k,\intervallen}}+i+\categoryleast+\categorynoise-2}!}\binom{\categoryactual-\categoryleast}{i}\\
&\leq \sup_{\mathbb{P}_{\Theta|G}\in\brc{0,1}}  \log \sum_{k \in \brc{0, 1, \cdots, \left\lceil \frac{\secretnum}{\intervallen}\right\rceil-1}}\sum_{i\in\brc{0,1,\cdots,\categoryactual-\categoryleast}} \bra{\frac{\samplenum}{\samplenum+\categoryleast+\categorynoise-2}}^i\binom{\categoryactual-\categoryleast}{i}\\
&= \sup_{\mathbb{P}_{\Theta|G}\in\brc{0,1}}  \log \sum_{k \in \brc{0, 1, \cdots, \left\lceil \frac{\secretnum}{\intervallen}\right\rceil-1}}\bra{1+\frac{\samplenum}{\samplenum+\categoryleast+\categorynoise-2}}^{\categoryactual-\categoryleast}\\
&= \log \brb{\ceil{ \frac{\secretnum}{\intervallen}}\bra{1+\frac{\samplenum}{\samplenum+\categoryleast+\categorynoise-2}}^{\categoryactual-\categoryleast}}.
\end{align*}

Therefore, we have
\begin{align*}
\privacyquan-\privacyquan^* &\leq \log \brb{\ceil{ \frac{\secretnum}{\intervallen}}\bra{1+\frac{\samplenum}{\samplenum+\categoryleast+\categorynoise-2}}^{\categoryactual-\categoryleast}} - \log\ceil{\frac{\secretnum}{\intervallen}}\\
&= \bra{\categoryactual-\categoryleast} \log\bra{1+\frac{\samplenum}{\samplenum+\categoryleast+\categorynoise-2}}\\
&< \categoryactual-\categoryleast,
\end{align*}
which indicates that Quantization Mechanism with any interval length $\intervallen$ is $1$-robust when the secret is the fraction of a category.

\end{proof}
\section{Privacy-Utility Tradeoffs Analysis of \rr{} and \qm{} with Misspecified of Feasible Attribute Combinations}
\label{section:proof_bounds}

In this section, we focus on the case where there is mis-specification of feasible attribute combinations, i.e., $\combinationsetestimate\not=\combinationsetall$, and analyze the upper and lower bounds of privacy and distortion for both \rr{} and \qm{} in \cref{subsec:upper,subsec:lower,subsec:distortion}, as well as the tightness of the privacy bounds in \cref{subsec:tightness}.

Similar to \cref{section:proof_robustness_tabular}, we denote $\releaseparamsetall$ as the set that contains all the distribution parameters whose support is or is within $\combinationsetestimate\cup\combinationsetall$, and $\releaseparamsetori$ as the set of parameters whose support is or is within $\combinationsetestimate$. Let $\beta\bra{\parameterdistribution'}$ %
be the number of categories in $\combinationsetall\setminus\combinationsetestimate$ that have non-zero probability mass in $\parameterdistribution'$, i.e., $\beta\bra{\parameterdistribution'} = \brd{\support{\parameterdistribution'}\setminus\combinationsetestimate}$.

\subsection{Privacy Upper Bounds}
\label{subsec:upper}

\begin{proposition}
\label{prop:privacy_upperbounds}
For any secret function $\secretnotation$, the SML upper bound of \rr{} is
\begin{align*}
\privacyrr \leq \min\Bigg\{&\log \Bigg(\frac{\secretnum \bra{e^{\epsilon}-1}}{\binom{\samplenum+\categoryleast+\categorynoise-1}{\categoryleast+\categorynoise-1}+e^{\epsilon}-1}+\frac{\binom{\samplenum+\categoryactual+\categorynoise-1}{\categoryactual+\categorynoise-1}}{\binom{\samplenum+\categoryactual+\categorynoise-1}{\categoryactual+\categorynoise-1}+e^{\epsilon}-1} \bra{1+\frac{\samplenum}{\samplenum+\categoryleast+\categorynoise-1}}^{\categoryactual-\categoryleast}\Bigg),\\
&\log\bra{1+\secretnum-\secretnum\cdot \frac{\bra{\frac{\categoryleast+\categorynoise}{\samplenum+\categoryactual+\categorynoise-1}}^{\categoryactual-\categoryleast}}{1+\bra{e^{\epsilon}-1}/\binom{\samplenum+\categoryactual+\categorynoise-1}{\categoryactual+\categorynoise-1}}}\Bigg\}.
\end{align*}
Specifically, when $e^{\epsilon}-1\leq \binom{\samplenum+\categoryleast+\categorynoise-1}{\categoryleast+\categorynoise-1} / \secretnum$, the upper bound can be written as
\begin{align*}
\privacyrr < \min\brc{\log \Bigg(\frac{\secretnum}{1+\secretnum}+ \bra{1+\frac{\samplenum}{\samplenum+\categoryleast+\categorynoise-1}}^{\categoryactual-\categoryleast}\Bigg), \ 
\log\bra{1+\secretnum-\secretnum\cdot \frac{\bra{\frac{\categoryleast+\categorynoise}{\samplenum+\categoryactual+\categorynoise-1}}^{\categoryactual-\categoryleast}}{1+1/\secretnum}}}.
\end{align*}

For any secret function $\secretnotation$, the SML upper bound of \qm{} is
\begin{align*}
\privacyquan \leq \min\brc{\log \brb{\ceil{ \frac{\secretnum}{\intervallen}}\bra{1+\frac{\samplenum}{\samplenum+\categoryleast+\categorynoise-2}}^{\categoryactual-\categoryleast}}, \ 
\log\bra{\ceil{\frac{\secretnum}{\intervallen}}+\secretnum-\secretnum\cdot\bra{\frac{\categoryleast+\categorynoise-1}{\samplenum+\categoryactual+\categorynoise-2}}^{\categoryactual-\categoryleast}}}.
\end{align*}
\end{proposition}

\begin{proof}

\subsubsection{Randomized Response}

$\privacyrr$ can be upper bounded by
\begin{align*}
\privacyrr &= \sup_{\mathbb{P}_{\Theta|G}\in\brc{0,1}}  \log\sum_{\theta'\in\releaseparamsetall} \sup_{\secretrv\in \secretvalueset} \mathbb{P}_{\Theta'|\Theta}\bra{\theta'|\theta_{\secretrv}}\\
&\leq \sup_{\mathbb{P}_{\Theta|G}\in\brc{0,1}}  \log\bra{\sup_{\secretrv\in \secretvalueset}\sum_{\theta'\in\releaseparamsetori}\mathbb{P}_{\Theta'|\Theta}\bra{\theta'|\theta_{\secretrv}}+\sum_{\secretrv\in \secretvalueset} \sum_{\theta'\in\releaseparamsetall\setminus \releaseparamsetori} \mathbb{P}_{\Theta'|\Theta}\bra{\theta'|\theta_{\secretrv}}}\\
&\leq \sup_{\mathbb{P}_{\Theta|G}\in\brc{0,1}}  \log\bra{1+\sum_{\secretrv\in \secretvalueset} \bra{1-\frac{\brd{\releaseparamsetori}}{\binom{\samplenum+\categoryleast+\categorynoise+\beta\bra{\theta_{\secretrv}}-1}{\categoryleast+\categorynoise+\beta\bra{\theta_{\secretrv}}-1}+e^{\epsilon}-1}}}\\
&\leq   \log\bra{1+\secretnum\cdot\frac{\binom{\samplenum+\categoryactual+\categorynoise-1}{\categoryactual+\categorynoise-1}-\binom{\samplenum+\categoryleast+\categorynoise-1}{\categoryleast+\categorynoise-1}+e^{\epsilon}-1}{\binom{\samplenum+\categoryactual+\categorynoise-1}{\categoryactual+\categorynoise-1}+e^{\epsilon}-1}}.
\end{align*}

Since we have
\begin{align*}
\frac{\binom{\samplenum+\categoryleast+\categorynoise-1}{\categoryleast+\categorynoise-1}}{\binom{\samplenum+\categoryactual+\categorynoise-1}{\categoryactual+\categorynoise-1}}=\frac{\bra{\samplenum+\categoryleast+\categorynoise-1}!\bra{\categoryactual+\categorynoise-1}!}{\bra{\samplenum+\categoryactual+\categorynoise-1}!\bra{\categoryleast+\categorynoise-1}!}\geq \bra{\frac{\categoryleast+\categorynoise}{\samplenum+\categoryactual+\categorynoise-1}}^{\categoryactual-\categoryleast},
\end{align*}
we can get that
\begin{align*}
\privacyrr\leq \log\bra{1+\secretnum-\secretnum\cdot \frac{\bra{\frac{\categoryleast+\categorynoise}{\samplenum+\categoryactual+\categorynoise-1}}^{\categoryactual-\categoryleast}}{1+\bra{e^{\epsilon}-1}/\binom{\samplenum+\categoryactual+\categorynoise-1}{\categoryactual+\categorynoise-1}}}.
\end{align*}

Combining with the privacy upper bound provided in \cref{section:proof_robustness_tabular}, we have
\begin{align*}
\privacyrr \leq \min\Bigg\{&\log \Bigg(\frac{\secretnum \bra{e^{\epsilon}-1}}{\binom{\samplenum+\categoryleast+\categorynoise-1}{\categoryleast+\categorynoise-1}+e^{\epsilon}-1}+\frac{\binom{\samplenum+\categoryactual+\categorynoise-1}{\categoryactual+\categorynoise-1}}{\binom{\samplenum+\categoryactual+\categorynoise-1}{\categoryactual+\categorynoise-1}+e^{\epsilon}-1} \bra{1+\frac{\samplenum}{\samplenum+\categoryleast+\categorynoise-1}}^{\categoryactual-\categoryleast}\Bigg),\\
&\log\bra{1+\secretnum-\secretnum\cdot \frac{\bra{\frac{\categoryleast+\categorynoise}{\samplenum+\categoryactual+\categorynoise-1}}^{\categoryactual-\categoryleast}}{1+\bra{e^{\epsilon}-1}/\binom{\samplenum+\categoryactual+\categorynoise-1}{\categoryactual+\categorynoise-1}}}\Bigg\}.
\end{align*}

Specifically, when $e^{\epsilon}-1\leq \binom{\samplenum+\categoryleast+\categorynoise-1}{\categoryleast+\categorynoise-1} / \secretnum$, we have
\begin{align*}
\privacyrr < \min\brc{\log \Bigg(\frac{\secretnum}{1+\secretnum}+ \bra{1+\frac{\samplenum}{\samplenum+\categoryleast+\categorynoise-1}}^{\categoryactual-\categoryleast}\Bigg), \ 
\log\bra{1+\secretnum-\secretnum\cdot \frac{\bra{\frac{\categoryleast+\categorynoise}{\samplenum+\categoryactual+\categorynoise-1}}^{\categoryactual-\categoryleast}}{1+1/\secretnum}}}.
\end{align*}

\subsubsection{Quantization Mechanism}
$\privacyquan$ can be upper bounded by
\begin{align*}
\privacyquan &= \sup_{\mathbb{P}_{\Theta|G}\in\brc{0,1}}  \log\sum_{\theta'\in\releaseparamsetall} \sup_{\secretrv\in \secretvalueset} \mathbb{P}_{\Theta'|\Theta}\bra{\theta'|\theta_{\secretrv}}\\
&\leq \sup_{\mathbb{P}_{\Theta|G}\in\brc{0,1}}  \log \sum_{k \in \brc{0, 1, \cdots, \left\lceil \frac{\secretnum}{\intervallen}\right\rceil-1}}\bra{\sup_{\secretrv\in \secretvalueset}\sum_{\theta'\in\releaseparamsetori_{\midpoint{k}}}\mathbb{P}_{\Theta'|\Theta}\bra{\theta'|\theta_{\secretrv}}+\sum_{\secretrv\in \secretvalueset} \sum_{\theta'\in\releaseparamsetall_{\midpoint{k}}\setminus \releaseparamsetori_{\midpoint{k}}} \mathbb{P}_{\Theta'|\Theta}\bra{\theta'|\theta_{\secretrv}}}\\
&\leq \sup_{\mathbb{P}_{\Theta|G}\in\brc{0,1}}  \log\bra{\ceil{\frac{\secretnum}{\intervallen}}+\sum_{\secretrv\in \secretvalueset} \sum_{\theta'\in\releaseparamsetall\setminus \releaseparamsetori} \mathbb{P}_{\Theta'|\Theta}\bra{\theta'|\theta_{\secretrv}}}\\
&\leq  \log\bra{\ceil{\frac{\secretnum}{\intervallen}}+\secretnum\cdot \sup_{k \in \brc{0, 1, \cdots, \left\lceil \frac{\secretnum}{\intervallen}\right\rceil-1}} \bra{1-\frac{\brd{\releaseparamsetori_{\midpoint{k}}}}{\binom{\samplenum\bra{1-\secretrv_{k,\intervallen}}+\categoryactual+\categorynoise-2}{\categoryactual+\categorynoise-2}}}}\\
&=   \log\bra{\ceil{\frac{\secretnum}{\intervallen}}+\secretnum-\secretnum\cdot\sup_{k \in \brc{0, 1, \cdots, \left\lceil \frac{\secretnum}{\intervallen}\right\rceil-1}}\frac{\binom{\samplenum\bra{1-\secretrv_{k,\intervallen}}+\categoryleast+\categorynoise-2}{\categoryleast+\categorynoise-2}}{\binom{\samplenum\bra{1-\secretrv_{k,\intervallen}}+\categoryactual+\categorynoise-2}{\categoryactual+\categorynoise-2}}}\\
&\leq \log\bra{\ceil{\frac{\secretnum}{\intervallen}}+\secretnum-\secretnum\cdot\bra{\frac{\categoryleast+\categorynoise-1}{\samplenum+\categoryactual+\categorynoise-2}}^{\categoryactual-\categoryleast}}.
\end{align*}

Combining with the privacy upper bound provided in \cref{section:proof_robustness_tabular}, we have
\begin{align*}
\privacyquan \leq \min\brc{\log \brb{\ceil{ \frac{\secretnum}{\intervallen}}\bra{1+\frac{\samplenum}{\samplenum+\categoryleast+\categorynoise-2}}^{\categoryactual-\categoryleast}}, \ 
\log\bra{\ceil{\frac{\secretnum}{\intervallen}}+\secretnum-\secretnum\cdot\bra{\frac{\categoryleast+\categorynoise-1}{\samplenum+\categoryactual+\categorynoise-2}}^{\categoryactual-\categoryleast}}}.
\end{align*}
\end{proof}

\subsection{Lower Bounds for SML}
\label{subsec:lower}

\begin{proposition}
\label{prop:privacy_lowerbounds}
For any secret function $\secretnotation$, the SML lower bound of \rr{} is
\begin{align*}
\privacyrr \geq 
\begin{cases}
\log\bra{\frac{\secretnum \bra{e^{\epsilon}-1}}{\binom{\samplenum+\categoryactual+\categorynoise-1}{\categoryactual+\categorynoise-1}+e^{\epsilon}-1}+\frac{\binom{\samplenum+\categoryleast+\categorynoise-1}{\categoryleast+\categorynoise-1}}{\binom{\samplenum+\categoryleast+\categorynoise-1}{\categoryleast+\categorynoise-1}+e^{\epsilon}-1} \bra{1+\frac{\samplenum-\bra{\categoryactual-\categoryleast}}{\samplenum+\categoryactual+\categorynoise-1}}^{\categoryactual-\categoryleast}}, & \categoryactual-\categoryleast\leq \log{\secretnum},\\
\log\bra{\frac{\secretnum \bra{e^{\epsilon}-1}}{\binom{\samplenum+\categoryactual+\categorynoise-1}{\categoryactual+\categorynoise-1}+e^{\epsilon}-1}+\frac{\binom{\samplenum+\categoryleast+\categorynoise-1}{\categoryleast+\categorynoise-1}}{\binom{\samplenum+\categoryleast+\categorynoise-1}{\categoryleast+\categorynoise-1}+e^{\epsilon}-1} \bra{1+\frac{\samplenum-{\log{\secretnum}}}{\samplenum+\categoryactual+\categorynoise-1}}^{\log{\secretnum}}}, & \log{\secretnum}<\categoryactual-\categoryleast< {\secretnum},\\
\log\bra{\secretnum-\bra{\secretnum-1}\bra{\frac{\categoryleast+\categorynoise+\zeta-1}{\samplenum+\categoryleast+\categorynoise}}^{\zeta}}, & \categoryactual-\categoryleast\geq {\secretnum}.
\end{cases}
\end{align*}

For any secret function $\secretnotation$, the SML lower bound of \qm{} is
\begin{align*}
\privacyquan \geq 
\begin{cases}
\log \brb{\floor{ \frac{\secretnum}{\intervallen}}\bra{1+\frac{\samplenum\intervallen-2\secretnum\bra{\categoryactual-\categoryleast}}{\samplenum\intervallen+2\secretnum\bra{\categoryactual+\categorynoise-2}}}^{\categoryactual-\categoryleast}}, & \categoryactual-\categoryleast \leq \log{\intervallen},\\
\log \brb{\floor{ \frac{\secretnum}{\intervallen}}\bra{1+\frac{\samplenum\intervallen-2\secretnum\cdot{\log{\intervallen}}}{\samplenum\intervallen+2\secretnum\bra{\categoryactual+\categorynoise-2}}}^{\log{\intervallen}}}, & \log{\intervallen}< \categoryactual-\categoryleast < {\intervallen},\\
\log\bra{\secretnum-\bra{\secretnum-\ceil{\frac{\secretnum}{\intervallen}}}\cdot\bra{\frac{\categoryleast+\categorynoise+\zeta-2}{\samplenum\intervallen/2\secretnum+\categoryleast+\categorynoise-1}}^{\zeta}}, & \categoryactual-\categoryleast \geq {\intervallen}.
\end{cases}
\end{align*}
\end{proposition}

\begin{proof}

\subsubsection{Randomized Response}

When $\categoryactual-\categoryleast\leq \log{\secretnum}$, we can find a prior distribution $\paramdistribution$ of the distribution parameter such that $\mathbb{P}_{\Theta|G}\in\brc{0,1}$ and $\forall \combinationsetsingle \subseteq \combinationsetall\setminus
\combinationsetpre, \ \exists\secretrv\in\secretvalueset: \support{\theta_{\secretrv}}=\combinationsetsingle\cup\combinationsetpre$. Therefore, we can get that 
\begin{align*}
\privacyrr&= \sup_{\mathbb{P}_{\Theta|G}\in\brc{0,1}}  \log\sum_{\theta'\in\releaseparamsetall} \sup_{\secretrv\in \secretvalueset} \mathbb{P}_{\Theta'|\Theta}\bra{\theta'|\theta_{\secretrv}}\\
&= \sup_{\mathbb{P}_{\Theta|G}\in\brc{0,1}}  \log\Bigg(\sum_{\theta'\in\brc{\theta_\secretrv}_{\secretrv\in \secretvalueset}}\sup_{\secretrv\in \secretvalueset} \mathbb{P}_{\Theta'|\Theta}\bra{\theta'|\theta_{\secretrv}}+\sum_{\theta'\in 
\releaseparamsetall\setminus\brc{\theta_\secretrv}_{\secretrv\in \secretvalueset}} \sup_{\secretrv\in \secretvalueset} \mathbb{P}_{\Theta'|\Theta}\bra{\theta'|\theta_{\secretrv}}\Bigg) \\
&\geq \log\bra{\sum_{\theta'\in \brc{\theta_\secretrv}_{\secretrv\in \secretvalueset}} \frac{ e^{\epsilon}}{\binom{\samplenum+\categoryleast+\categorynoise+\beta\bra{\parameterdistribution'}-1}{\categoryleast+\categorynoise+\beta\bra{\parameterdistribution'}-1}+e^{\epsilon}-1}+\sum_{\theta'\in 
\releaseparamsetall\setminus\brc{\theta_\secretrv}_{\secretrv\in \secretvalueset}} \frac{1}{\binom{\samplenum+\categoryleast+\categorynoise+\beta\bra{\parameterdistribution'}-1}{\categoryleast+\categorynoise+\beta\bra{\parameterdistribution'}-1}+e^{\epsilon}-1}}\\
&\geq    \log\bra{\frac{\secretnum \bra{e^{\epsilon}-1}}{\binom{\samplenum+\categoryactual+\categorynoise-1}{\categoryactual+\categorynoise-1}+e^{\epsilon}-1}+\sum_{\theta'\in 
\releaseparamsetall} \frac{1}{\binom{\samplenum+\categoryleast+\categorynoise+\beta\bra{\parameterdistribution'}-1}{\categoryleast+\categorynoise+\beta\bra{\parameterdistribution'}-1}+e^{\epsilon}-1}}\\
&= \log\bra{\frac{\secretnum \bra{e^{\epsilon}-1}}{\binom{\samplenum+\categoryactual+\categorynoise-1}{\categoryactual+\categorynoise-1}+e^{\epsilon}-1}+\sum_{i\in\brc{0}\cup\brb{\categoryactual-\categoryleast}} \frac{\binom{\categoryactual-\categoryleast}{i}\binom{\samplenum-i+\categoryleast+\categorynoise+i-1}{\categoryleast+\categorynoise+i-1}}{\binom{\samplenum+\categoryleast+\categorynoise+i-1}{\categoryleast+\categorynoise+i-1}+e^{\epsilon}-1}}\\
&\geq \log\bra{\frac{\secretnum \bra{e^{\epsilon}-1}}{\binom{\samplenum+\categoryactual+\categorynoise-1}{\categoryactual+\categorynoise-1}+e^{\epsilon}-1}+\frac{\binom{\samplenum+\categoryleast+\categorynoise-1}{\categoryleast+\categorynoise-1}}{\binom{\samplenum+\categoryleast+\categorynoise-1}{\categoryleast+\categorynoise-1}+e^{\epsilon}-1}\sum_{i\in\brc{0}\cup\brb{\categoryactual-\categoryleast}} \frac{\binom{\categoryactual-\categoryleast}{i}\binom{\samplenum-i+\categoryleast+\categorynoise+i-1}{\categoryleast+\categorynoise+i-1}}{\binom{\samplenum+\categoryleast+\categorynoise+i-1}{\categoryleast+\categorynoise+i-1}}}\\
&= \log\bra{\frac{\secretnum \bra{e^{\epsilon}-1}}{\binom{\samplenum+\categoryactual+\categorynoise-1}{\categoryactual+\categorynoise-1}+e^{\epsilon}-1}+\frac{\binom{\samplenum+\categoryleast+\categorynoise-1}{\categoryleast+\categorynoise-1}}{\binom{\samplenum+\categoryleast+\categorynoise-1}{\categoryleast+\categorynoise-1}+e^{\epsilon}-1}\sum_{i\in\brc{0}\cup\brb{\categoryactual-\categoryleast}}\hspace{-1mm} \frac{\samplenum!\bra{\samplenum+\categoryleast+\categorynoise-1}!\binom{\categoryactual-\categoryleast}{i}}{\bra{\samplenum-i}!\bra{\samplenum+i+\categoryleast+\categorynoise-1}!}}\\
&\geq \log\bra{\frac{\secretnum \bra{e^{\epsilon}-1}}{\binom{\samplenum+\categoryactual+\categorynoise-1}{\categoryactual+\categorynoise-1}+e^{\epsilon}-1}+\frac{\binom{\samplenum+\categoryleast+\categorynoise-1}{\categoryleast+\categorynoise-1}}{\binom{\samplenum+\categoryleast+\categorynoise-1}{\categoryleast+\categorynoise-1}+e^{\epsilon}-1}\sum_{i\in\brc{0}\cup\brb{\categoryactual-\categoryleast}} \bra{\frac{\samplenum-\bra{\categoryactual-\categoryleast}}{\samplenum+\categoryactual+\categorynoise-1}}^i \binom{\categoryactual-\categoryleast}{i}}\\
&= \log\bra{\frac{\secretnum \bra{e^{\epsilon}-1}}{\binom{\samplenum+\categoryactual+\categorynoise-1}{\categoryactual+\categorynoise-1}+e^{\epsilon}-1}+\frac{\binom{\samplenum+\categoryleast+\categorynoise-1}{\categoryleast+\categorynoise-1}}{\binom{\samplenum+\categoryleast+\categorynoise-1}{\categoryleast+\categorynoise-1}+e^{\epsilon}-1} \bra{1+\frac{\samplenum-\bra{\categoryactual-\categoryleast}}{\samplenum+\categoryactual+\categorynoise-1}}^{\categoryactual-\categoryleast}}. 
\end{align*}

When $\log{\secretnum} <\categoryactual-\categoryleast < \secretnum$, we can easily get that
\begin{align*}
\privacyrr \geq \log\bra{\frac{\secretnum \bra{e^{\epsilon}-1}}{\binom{\samplenum+\categoryactual+\categorynoise-1}{\categoryactual+\categorynoise-1}+e^{\epsilon}-1}+\frac{\binom{\samplenum+\categoryleast+\categorynoise-1}{\categoryleast+\categorynoise-1}}{\binom{\samplenum+\categoryleast+\categorynoise-1}{\categoryleast+\categorynoise-1}+e^{\epsilon}-1} \bra{1+\frac{\samplenum-{\log{\secretnum}}}{\samplenum+\categoryactual+\categorynoise-1}}^{\log{\secretnum}}}
\end{align*}

When $\categoryactual-\categoryleast \geq \secretnum$, denote $\combinationsetsingle_{\theta}\triangleq \support{\theta}\setminus \combinationsetpre$. We can find a prior distribution $\paramdistribution$ of the distribution parameter such that $\mathbb{P}_{\Theta|G}\in\brc{0,1}$ and $\forall \secretrv_1, \secretrv_2\in \secretvalueset, \secretrv_1\not= \secretrv_2: \combinationsetsingle_{\theta_{\secretrv_1}}\cap \combinationsetsingle_{\theta_{\secretrv_1}} = \emptyset$ and $\brd{\combinationsetsingle_{\theta_{\secretrv_1}}}=\brd{\combinationsetsingle_{\theta_{\secretrv_2}}}=\floor{\frac{\categoryactual-\categoryleast}{\secretnum}}$. Denote $\zeta \triangleq \floor{\frac{\categoryactual-\categoryleast}{\secretnum}}$. Therefore, we can get that
\begin{align*}
\privacyrr &= \sup_{\mathbb{P}_{\Theta|G}\in\brc{0,1}}  \log\sum_{\theta'\in\releaseparamsetall} \sup_{\secretrv\in \secretvalueset} \mathbb{P}_{\Theta'|\Theta}\bra{\theta'|\theta_{\secretrv}}\\
&\geq   \log\bra{\frac{\brd{\releaseparamsetori}}{\binom{\samplenum+\categoryleast+\categorynoise+\zeta-1}{\categoryleast+\categorynoise+\zeta-1}+e^{\epsilon}-1}+\sum_{\secretrv\in \secretvalueset} \bra{1-\frac{\brd{\releaseparamsetori}}{\binom{\samplenum+\categoryleast+\categorynoise+\zeta-1}{\categoryleast+\categorynoise+\zeta-1}+e^{\epsilon}-1}}}\\
&\geq \log\bra{\secretnum-\bra{\secretnum-1}\cdot \frac{\binom{\samplenum+\categoryleast+\categorynoise-1}{\categoryleast+\categorynoise-1}}{\binom{\samplenum+\categoryleast+\categorynoise+\zeta-1}{\categoryleast+\categorynoise+\zeta-1}+e^{\epsilon}-1}}\\
&\geq \log\bra{\secretnum-\bra{\secretnum-1}\bra{\frac{\categoryleast+\categorynoise+\zeta-1}{\samplenum+\categoryleast+\categorynoise}}^{\zeta}}.
\end{align*}

\subsubsection{Quantization Mechanism}

Denote $\secretvalueset_{k}$ as the set of the secret values within the $k$-th interval, i.e., $\secretvalueset_{k}=\brc{\secretrv_{k\intervallen+j}}_{j\in \brb{\intervallen}}$.
When $\categoryactual-\categoryleast\leq \log{\intervallen}$, we can find a prior distribution $\paramdistribution$ of the distribution parameter such that $\mathbb{P}_{\Theta|G}\in\brc{0,1}$ and $\forall \combinationsetsingle \subseteq \combinationsetall\setminus
\combinationsetpre, k \in \brc{0, 1, \cdots,  \floor{\frac{\secretnum}{\intervallen}}-1}, \ \exists \secretrv\in\secretvalueset_{k}: \support{\parameterdistribution_{\secretrv}}=\combinationsetsingle\cup\combinationsetpre$. Therefore, we can get that 
\begin{align*}
\privacyquan &= \sup_{\mathbb{P}_{\Theta|G}\in\brc{0,1}}  \log\sum_{\theta'\in\releaseparamsetall} \sup_{\secretrv\in \secretvalueset} \mathbb{P}_{\Theta'|\Theta}\bra{\theta'|\theta_{\secretrv}}\\
&\geq \sup_{\mathbb{P}_{\Theta|G}\in\brc{0,1}}  \log \sum_{k \in \brc{0, 1, \cdots,  \floor{\frac{\secretnum}{\intervallen}}-1}}\sum_{\theta'\in\releaseparamsetall_{\midpoint{k}}} \sup_{\secretrv\in \secretvalueset} \mathbb{P}_{\Theta'|\Theta}\bra{\theta'|\theta_{\secretrv}}\\
&\geq  \log \sum_{k \in \brc{0, 1, \cdots, \floor{\frac{\secretnum}{\intervallen}}-1}}\sum_{\theta'\in\releaseparamsetall_{\midpoint{k}}} \frac{1}{\binom{\samplenum\bra{1-\secretrv_{k,\intervallen}}+\categoryleast+\categorynoise+\beta\bra{\parameterdistribution'}-2}{\categoryleast+\categorynoise+\beta\bra{\parameterdistribution'}-2}}\\
&=  \log \sum_{k \in \brc{0, 1, \cdots, \floor{\frac{\secretnum}{\intervallen}}-1}}\sum_{i\in\brc{0,1,\cdots,\categoryactual-\categoryleast}}\frac{\binom{\categoryactual-\categoryleast}{i}\binom{\samplenum\bra{1-\secretrv_{k,\intervallen}}-i+\categoryleast+\categorynoise+i-2}{\categoryleast+\categorynoise+i-2}}{\binom{\samplenum\bra{1-\secretrv_{k,\intervallen}}+\categoryleast+\categorynoise+i-2}{\categoryleast+\categorynoise+i-2}}\\
&=  \log \sum_{k \in \brc{0, 1, \cdots, \floor{\frac{\secretnum}{\intervallen}}-1}}\sum_{i\in\brc{0,1,\cdots,\categoryactual-\categoryleast}}\frac{\bra{\samplenum\bra{1-\secretrv_{k,\intervallen}}}!\bra{\samplenum\bra{1-\secretrv_{k,\intervallen}}+\categoryleast+\categorynoise-2}!}{\bra{\samplenum\bra{1-\secretrv_{k,\intervallen}}-i}!\bra{\samplenum\bra{1-\secretrv_{k,\intervallen}}+i+\categoryleast+\categorynoise-2}!}\binom{\categoryactual-\categoryleast}{i}\\
&\geq \log \sum_{k \in \brc{0, 1, \cdots, \floor{\frac{\secretnum}{\intervallen}}-1}}\sum_{i\in\brc{0,1,\cdots,\categoryactual-\categoryleast}} \bra{\frac{\samplenum\intervallen-2\secretnum\bra{\categoryactual-\categoryleast}}{\samplenum\intervallen+2\secretnum\bra{\categoryactual+\categorynoise-2}}}^i\binom{\categoryactual-\categoryleast}{i}\\
&=  \log \sum_{k \in \brc{0, 1, \cdots, \floor{\frac{\secretnum}{\intervallen}}-1}}\bra{1+\frac{\samplenum\intervallen-2\secretnum\bra{\categoryactual-\categoryleast}}{\samplenum\intervallen+2\secretnum\bra{\categoryactual+\categorynoise-2}}}^{\categoryactual-\categoryleast}\\
&= \log \brb{\floor{ \frac{\secretnum}{\intervallen}}\bra{1+\frac{\samplenum\intervallen-2\secretnum\bra{\categoryactual-\categoryleast}}{\samplenum\intervallen+2\secretnum\bra{\categoryactual+\categorynoise-2}}}^{\categoryactual-\categoryleast}}.
\end{align*}

When $\log{\intervallen} <\categoryactual-\categoryleast < \intervallen$, we can easily get that
\begin{align*}
\privacyquan\geq \log \brb{\floor{ \frac{\secretnum}{\intervallen}}\bra{1+\frac{\samplenum\intervallen-2\secretnum\cdot{\log{\intervallen}}}{\samplenum\intervallen+2\secretnum\bra{\categoryactual+\categorynoise-2}}}^{\log{\intervallen}}}.
\end{align*}

When $\categoryactual-\categoryleast \geq \intervallen$, denote $\combinationsetsingle_{\theta}\triangleq \support{\theta}\setminus \combinationsetpre$. We can find a prior distribution $\paramdistribution$ of the distribution parameter such that $\mathbb{P}_{\Theta|G}\in\brc{0,1}$ and $\forall k \in \brc{0, 1, \cdots,  \floor{\frac{\secretnum}{\intervallen}}-1}, \secretrv_1, \secretrv_2\in \secretvalueset_k, \secretrv_1\not= \secretrv_2: \combinationsetsingle_{\theta_{\secretrv_1}}\cap \combinationsetsingle_{\theta_{\secretrv_1}} = \emptyset$ and $\brd{\combinationsetsingle_{\theta_{\secretrv_1}}}=\brd{\combinationsetsingle_{\theta_{\secretrv_2}}}=\floor{\frac{\categoryactual-\categoryleast}{\intervallen}}$. Denote $\zeta \triangleq \floor{\frac{\categoryactual-\categoryleast}{\intervallen}}$. Therefore, we can get that
\begin{align*}
\privacyquan &= \sup_{\mathbb{P}_{\Theta|G}\in\brc{0,1}}  \log\sum_{\theta'\in\releaseparamsetall} \sup_{\secretrv\in \secretvalueset} \mathbb{P}_{\Theta'|\Theta}\bra{\theta'|\theta_{\secretrv}}\\
&\geq  \log\brb{\sum_{k \in \brc{0, 1, \cdots, \left\lceil \frac{\secretnum}{\intervallen}\right\rceil-1}}\bra{\frac{\brd{\releaseparamsetori_{\midpoint{k}}}}{\binom{\samplenum\bra{1-\secretrv_{k,\intervallen}}+\categoryleast+\zeta+\categorynoise-2}{\categoryleast+\zeta+\categorynoise-2}}+\min\brc{\intervallen, \secretnum-k\intervallen}\cdot  \bra{1-\frac{\brd{\releaseparamsetori_{\midpoint{k}}}}{\binom{\samplenum\bra{1-\secretrv_{k,\intervallen}}+\categoryleast+\zeta+\categorynoise-2}{\categoryleast+\zeta+\categorynoise-2}}}}}\\
&\geq   \log\bra{\secretnum-\bra{\secretnum-\ceil{\frac{\secretnum}{\intervallen}}}\cdot\sup_{k \in \brc{0, 1, \cdots, \left\lceil \frac{\secretnum}{\intervallen}\right\rceil-1}}\frac{\binom{\samplenum\bra{1-\secretrv_{k,\intervallen}}+\categoryleast+\categorynoise-2}{\categoryleast+\categorynoise-2}}{\binom{\samplenum\bra{1-\secretrv_{k,\intervallen}}+\categoryleast+\zeta+\categorynoise-2}{\categoryleast+\zeta+\categorynoise-2}}}\\
&\geq   \log\bra{\secretnum-\bra{\secretnum-\ceil{\frac{\secretnum}{\intervallen}}}\cdot\bra{\frac{\categoryleast+\categorynoise+\zeta-2}{\samplenum\intervallen/2\secretnum+\categoryleast+\categorynoise-1}}^{\zeta}}.
\end{align*}

\end{proof}

\subsection{Tightness of Lower and Upper Bounds for SML}
\label{subsec:tightness}

\begin{proposition}
Let $\boundgaprr$ and $\boundgapquan$ be the gaps between SML lower and upper bounds of \rr{} and \qm{} under precision level $\samplenum$ respectively.

For \rr{}, we have
\begin{align*}
&\lim_{\samplenum\rightarrow\infty}\boundgaprr \leq 1, \quad \text{when } \categoryactual-\categoryleast< \secretnum, \ e^{\epsilon}-1\leq \binom{\samplenum+\categoryleast+\categorynoise-1}{\categoryleast+\categorynoise-1} / \secretnum,\\
&\lim_{\samplenum\rightarrow\infty}\boundgaprr = 0, \quad \text{when } \categoryactual-\categoryleast\geq \secretnum.
\end{align*}

For \qm{}, we have
\begin{align*}
&\lim_{\samplenum\rightarrow\infty}\boundgapquan \leq 1, \quad \text{when } \categoryactual-\categoryleast< \intervallen,\\
&\lim_{\samplenum\rightarrow\infty}\boundgapquan = 0, \quad \text{when } \categoryactual-\categoryleast\geq \intervallen.
\end{align*}
Specifically, when $\secretnum$ can be divided by $\intervallen$, $\lim_{\samplenum\rightarrow\infty}\boundgapquan=0$ whatever $\categoryactual-\categoryleast$ is.
\end{proposition}

\begin{proof}
    
\subsubsection{Randomized Response}

When $\categoryactual-\categoryleast\leq \log \secretnum$ and $e^{\epsilon}-1\leq \binom{\samplenum+\categoryleast+\categorynoise-1}{\categoryleast+\categorynoise-1} / \secretnum$, we have
\begin{align*}
\lim_{\samplenum\rightarrow\infty}\boundgaprr &\leq \lim_{\samplenum\rightarrow\infty}\log \Bigg(\frac{\secretnum \bra{e^{\epsilon}-1}}{\binom{\samplenum+\categoryleast+\categorynoise-1}{\categoryleast+\categorynoise-1}+e^{\epsilon}-1}+\frac{\binom{\samplenum+\categoryactual+\categorynoise-1}{\categoryactual+\categorynoise-1}}{\binom{\samplenum+\categoryactual+\categorynoise-1}{\categoryactual+\categorynoise-1}+e^{\epsilon}-1} \bra{1+\frac{\samplenum}{\samplenum+\categoryleast+\categorynoise-1}}^{\categoryactual-\categoryleast}\Bigg)\\ 
&\quad - \lim_{\samplenum\rightarrow\infty}\log\bra{\frac{\secretnum \bra{e^{\epsilon}-1}}{\binom{\samplenum+\categoryactual+\categorynoise-1}{\categoryactual+\categorynoise-1}+e^{\epsilon}-1}+\frac{\binom{\samplenum+\categoryleast+\categorynoise-1}{\categoryleast+\categorynoise-1}}{\binom{\samplenum+\categoryleast+\categorynoise-1}{\categoryleast+\categorynoise-1}+e^{\epsilon}-1} \bra{1+\frac{\samplenum-\bra{\categoryactual-\categoryleast}}{\samplenum+\categoryactual+\categorynoise-1}}^{\categoryactual-\categoryleast}}\\
&= \lim_{\samplenum\rightarrow\infty} \log\bra{\frac{\frac{\secretnum \bra{e^{\epsilon}-1}}{\binom{\samplenum+\categoryleast+\categorynoise-1}{\categoryleast+\categorynoise-1}+e^{\epsilon}-1}+\frac{\binom{\samplenum+\categoryactual+\categorynoise-1}{\categoryactual+\categorynoise-1}}{\binom{\samplenum+\categoryactual+\categorynoise-1}{\categoryactual+\categorynoise-1}+e^{\epsilon}-1}\cdot 2^{\categoryactual-\categoryleast}}{\frac{\secretnum \bra{e^{\epsilon}-1}}{\binom{\samplenum+\categoryactual+\categorynoise-1}{\categoryactual+\categorynoise-1}+e^{\epsilon}-1}+\frac{\binom{\samplenum+\categoryleast+\categorynoise-1}{\categoryleast+\categorynoise-1}}{\binom{\samplenum+\categoryleast+\categorynoise-1}{\categoryleast+\categorynoise-1}+e^{\epsilon}-1}\cdot 2^{\categoryactual-\categoryleast}}}\\
&\leq \log\bra{\frac{\frac{1}{1+1/\secretnum}+2^{\categoryactual-\categoryleast}}{\frac{1}{1+1/\secretnum}\cdot2^{\categoryactual-\categoryleast}}}\\
&= \log\bra{1+\frac{1}{\secretnum}+2^{-\bra{\categoryactual-\categoryleast}}}.
\end{align*}
Since we focus on the nontrivial case where $\secretnum>1$ and $\categoryactual-\categoryleast>0$, we have $\lim_{\samplenum\rightarrow\infty}\boundgaprr\leq 1$.

When $\log \secretnum< \categoryactual-\categoryleast< \secretnum$ and $e^{\epsilon}-1\leq \binom{\samplenum+\categoryleast+\categorynoise-1}{\categoryleast+\categorynoise-1} / \secretnum$, we have
\begin{align*}
\lim_{\samplenum\rightarrow\infty}\boundgaprr &\leq \log \secretnum - \lim_{\samplenum\rightarrow\infty}\log\bra{\frac{\secretnum \bra{e^{\epsilon}-1}}{\binom{\samplenum+\categoryactual+\categorynoise-1}{\categoryactual+\categorynoise-1}+e^{\epsilon}-1}+\frac{\binom{\samplenum+\categoryleast+\categorynoise-1}{\categoryleast+\categorynoise-1}}{\binom{\samplenum+\categoryleast+\categorynoise-1}{\categoryleast+\categorynoise-1}+e^{\epsilon}-1} \bra{1+\frac{\samplenum-\log \secretnum}{\samplenum+\categoryactual+\categorynoise-1}}^{\log\secretnum}}\\
&\leq \log\bra{\frac{\secretnum}{\frac{1}{1+1/\secretnum}\cdot2^{\log \secretnum}}}\\
&= \log\bra{1+\frac{1}{\secretnum}}\\
&<1.
\end{align*}

When $\categoryactual-\categoryleast\geq \secretnum$, denote $\zeta \triangleq \floor{\frac{\categoryactual-\categoryleast}{\secretnum}}$ and we have
\begin{align*}
\lim_{\samplenum\rightarrow\infty}\boundgaprr &\leq \log \secretnum - \lim_{\samplenum\rightarrow\infty}\log\bra{\secretnum-\bra{\secretnum-1}\bra{\frac{\categoryleast+\categorynoise+\zeta-1}{\samplenum+\categoryleast+\categorynoise}}^{\zeta}}
= \log \secretnum - \log \secretnum = 0.
\end{align*}

Hence, we can get that $\lim_{\samplenum\rightarrow\infty}\boundgaprr\leq 1$ when $\categoryactual-\categoryleast\leq \secretnum$ and $e^{\epsilon}-1\leq \binom{\samplenum+\categoryleast+\categorynoise-1}{\categoryleast+\categorynoise-1} / \secretnum$, and $\lim_{\samplenum\rightarrow\infty}\boundgaprr=0$ when $\categoryactual-\categoryleast> \secretnum$.

\subsubsection{Quantization Mechanism}

Let $\boundgapquan$ be the gap of lower and upper bounds of the privacy of randomized response under a dataset with $\samplenum$ samples.
When $\categoryactual-\categoryleast\leq \log\intervallen$, we have
\begin{align*}
\lim_{\samplenum\rightarrow\infty}\boundgapquan &\leq \lim_{\samplenum\rightarrow\infty}\log \brb{\ceil{ \frac{\secretnum}{\intervallen}}\bra{1+\frac{\samplenum}{\samplenum+\categoryleast+\categorynoise-2}}^{\categoryactual-\categoryleast}} - \lim_{\samplenum\rightarrow\infty}\log \brb{\floor{ \frac{\secretnum}{\intervallen}}\bra{1+\frac{\samplenum\intervallen-2\secretnum\bra{\categoryactual-\categoryleast}}{\samplenum\intervallen+2\secretnum\bra{\categoryactual+\categorynoise-2}}}^{\categoryactual-\categoryleast}}\\
& = \log \ceil{ \frac{\secretnum}{\intervallen}} - \log \floor{ \frac{\secretnum}{\intervallen}}\\
&\leq 1.
\end{align*}
Specifically, when $\secretnum$ can be divided by $\intervallen$, $\lim_{\samplenum\rightarrow\infty}\boundgapquan = 0$.

When $\log\intervallen<\categoryactual-\categoryleast< \intervallen$, we have
\begin{align*}
\lim_{\samplenum\rightarrow\infty}\boundgapquan &\leq \log\secretnum - \lim_{\samplenum\rightarrow\infty}\log \brb{\floor{ \frac{\secretnum}{\intervallen}}\bra{1+\frac{\samplenum\intervallen-2\secretnum\log\intervallen}{\samplenum\intervallen+2\secretnum\bra{\categoryactual+\categorynoise-2}}}^{\log\intervallen}}\\
&=\log\secretnum - \log\bra{\intervallen\cdot\floor{\frac{\secretnum}{\intervallen}}}\\
&<1.
\end{align*}
Specifically, when $\secretnum$ can be divided by $\intervallen$, $\lim_{\samplenum\rightarrow\infty}\boundgapquan = 0$.

When $\categoryactual-\categoryleast\geq \intervallen$, denote $\combinationsetsingle_{\theta}\triangleq \support{\theta}\setminus \combinationsetpre$ and we have
\begin{align*}
\lim_{\samplenum\rightarrow\infty}\boundgapquan &\leq \log\secretnum - \lim_{\samplenum\rightarrow\infty}\log\bra{\secretnum-\bra{\secretnum-\ceil{\frac{\secretnum}{\intervallen}}}\cdot\bra{\frac{\categoryleast+\categorynoise+\zeta-2}{\samplenum\intervallen/2\secretnum+\categoryleast+\categorynoise-1}}^{\zeta}} = \log\secretnum - \log\secretnum = 0.
\end{align*}

Hence, we can get that $\lim_{\samplenum\rightarrow\infty}\boundgapquan\leq 1$ when $\categoryactual-\categoryleast\leq \intervallen$, and $\lim_{\samplenum\rightarrow\infty}\boundgapquan=0$ otherwise. Specifically, when $\secretnum$ can be divided by $\intervallen$, $\lim_{\samplenum\rightarrow\infty}\boundgapquan=0$ whatever $\categoryactual-\categoryleast$ is.

\end{proof}

\subsection{Lower and Upper Bounds for Distortion}
\label{subsec:distortion}

\subsubsection{Randomized Response}

Similar to the proof of \cref{prop:mech_tradeoff}, we can bound the distortion of Randomized Response as
\begin{align*}
\frac{{\categoryleast+\categorynoise-1}}{\bra{\categoryleast+\categorynoise}\bra{1+\rrratio_1}}\leq \distortionrr \leq \frac{{\categoryactual+\categorynoise-1}}{\bra{\categoryactual+\categorynoise}\bra{1+\rrratio_2}},
\end{align*}
where $\rrratio_1 \triangleq \frac{e^{\epsilon}-1}{\binom{\samplenum+\categoryleast+\categorynoise-1}{\categoryleast+\categorynoise-1}}$ and $\rrratio_2 \triangleq \frac{e^{\epsilon}-1}{\binom{\samplenum+\categoryactual+\categorynoise-1}{\categoryactual+\categorynoise-1}}$.

\subsubsection{Quantization Mechanism} 
When secret is the fraction of a category, similar to the proof of \cref{prop:mech_tradeoff},  we can bound the distortion of Quantization Mechanism as 
\begin{align*}
\frac{1}{2}+\frac{\bra{\categoryleast+\categorynoise}\floor{\frac{\intervallen}{2}}-\samplenum}{2\samplenum\bra{\categoryleast+\categorynoise-1}} \leq \distortionquan \leq \frac{1}{2}+\frac{\bra{\categoryactual+\categorynoise}\floor{\frac{\intervallen}{2}}-\samplenum}{2\samplenum\bra{\categoryactual+\categorynoise-1}}.
\end{align*}
\section{Analysis of privacy performance with the increase of $\categorynoise$}
\label{sec:curve_explain}

In this section, we provide theoretical analysis on the patterns of decrease in privacy value under different numbers of missing feasible attribute combinations $\categoryactual-\categoryleast$, that is, with increasing $\categorynoise$, the privacy value drops linearly when $\categoryactual-\categoryleast=5$, drops in a concave curve when $\categoryactual-\categoryleast=20$, and always achieves its trivial upper bound $\log \secretnum$ when $\categoryactual-\categoryleast=100$. %

From \cref{prop:privacy_upperbounds} we know that the privacy of Quantization mechanism can be upper bounded by 
\begin{align*}
\privacyrr \leq \min\brc{\log \brb{\ceil{ \frac{\secretnum}{\intervallen}}\bra{1+\frac{\samplenum}{\samplenum+\categoryleast+\categorynoise-2}}^{\categoryactual-\categoryleast}}, \ 
\log\bra{\ceil{\frac{\secretnum}{\intervallen}}+\secretnum-\secretnum\cdot\bra{\frac{\categoryleast+\categorynoise-1}{\samplenum+\categoryactual+\categorynoise-2}}^{\categoryactual-\categoryleast}}}.
\end{align*}

For the Census Income Dataset we adopt in our experiment, we have $\samplenum=48,842, \categoryleast=22,381$, and $\log \secretnum=4.09$. When $\categoryactual-\categoryleast=5$ and $\categorynoise\ll\samplenum+\categoryleast$, based on simple calculation, we can get that the privacy is upper bounded by 
$
\privacyrr \leq \log \brb{\ceil{ \frac{\secretnum}{\intervallen}}\bra{1+\frac{\samplenum}{\samplenum+\categoryleast+\categorynoise-2}}^5}
$. Then we have
\begin{align*}
\privacyrr &\leq \log \brb{\ceil{ \frac{\secretnum}{\intervallen}}\bra{1+\frac{\samplenum}{\samplenum+\categoryleast+\categorynoise-2}}^5}\\
&= \log \ceil{\frac{\secretnum}{\intervallen}} + 5\brb{\log \bra{2\samplenum+\categoryleast+\categorynoise-2}-\log\bra{\samplenum+\categoryleast+\categorynoise-2}}\\
&= \log \ceil{\frac{\secretnum}{\intervallen}} + 5\brb{\log\frac{2\samplenum+\categoryleast-2}{\samplenum+\categoryleast-2} + \log \bra{1 + \frac{\categorynoise}{2\samplenum+\categoryleast-2}}-\log \bra{1 + \frac{\categorynoise}{\samplenum+\categoryleast-2}}}\\
&= \log \ceil{\frac{\secretnum}{\intervallen}} + 5\brb{\log\frac{2\samplenum+\categoryleast-2}{\samplenum+\categoryleast-2} - \frac{\samplenum}{\bra{2\samplenum+\categoryleast-2}\bra{\samplenum+\categoryleast-2}}\cdot \categorynoise} + o\bra{\frac{\categorynoise}{\samplenum+\categoryleast}},
\end{align*}
which indicates a linear function of $\categorynoise$ dominate the privacy upper bound. Therefore, when $\categoryactual-\categoryleast=5$ and $\categorynoise\ll\samplenum+\categoryleast$, the privacy upper bound of the Quantization mechanism decrease linearly with the increase of $\categorynoise$.

When $\categoryactual-\categoryleast=20$ and $\samplenum\leq \categorynoise<5\samplenum$, we can get that the privacy is upper bounded by 
$\privacyrr \leq 
\log\bra{\ceil{\frac{\secretnum}{\intervallen}}+\secretnum-\secretnum\cdot\bra{\frac{\categoryleast+\categorynoise-1}{\samplenum+\categoryactual+\categorynoise-2}}^{20}}$. Let
$
f\bra{\categorynoise} \triangleq 
\bra{\frac{\categoryleast+\categorynoise-1}{\samplenum+\categoryactual+\categorynoise-2}}^{20}
$ and we have
\begin{align*}
\frac{\mathrm{d}^2 f\bra{\categorynoise}}{\mathrm{d} \bra{\categorynoise}^2} = \frac{20\bra{\samplenum-1}\bra{\frac{\categoryleast+\categorynoise-1}{\samplenum+\categoryactual+\categorynoise-2}}^{18}}{\bra{\samplenum+\categoryactual+\categorynoise-2}^4}\cdot \brb{19\bra{\samplenum-1}-2\bra{\categoryleast+\categorynoise-1}}>0.
\end{align*}
Therefore, $g\bra{\categorynoise}=\ceil{\frac{\secretnum}{\intervallen}}+\secretnum-\secretnum\cdot\bra{\frac{\categoryleast+\categorynoise-1}{\samplenum+\categoryactual+\categorynoise-2}}^{20}=\ceil{\frac{\secretnum}{\intervallen}}+\secretnum-\secretnum\cdot f\bra{\categorynoise}$ is a concave function. Since $l(x)=\log\bra{x}$ is a concave and monotonically increasing function, we have $h\bra{\categorynoise} = \log\bra{\ceil{\frac{\secretnum}{\intervallen}}+\secretnum-\secretnum\cdot\bra{\frac{\categoryleast+\categorynoise-1}{\samplenum+\categoryactual+\categorynoise-2}}^{20}} = l\circ g\bra{\categorynoise}$ is a concave function. Therefore, when $\categoryactual-\categoryleast=20$ and $\categorynoise<5\samplenum$, with the increase of $\categorynoise$, the privacy upper bound of the Quantization mechanism decreases in a concave curve.

When $\categoryactual-\categoryleast=100$ and $\categorynoise<5\samplenum$, based on simple calculation, we can get that the privacy upper bound is always $\log \secretnum$.

More generally, we analyze the patterns of decrease in privacy value with increasing $\categorynoise$ under the case where $\categorynoise\ll \categoryleast$ and $\categoryactual-\categoryleast \ll \samplenum+\categoryleast$. Under this setting, we can get that
$\frac{\samplenum}{\samplenum+\categoryleast+\categorynoise-2} \approx \frac{\samplenum}{\samplenum+\categoryleast}\triangleq r$ and $\frac{\categoryleast+\categorynoise-1}{\samplenum+\categoryactual+\categorynoise-2} \approx \frac{\categoryleast}{\samplenum+\categoryleast}=1-r$. Therefore, when $\log \brb{\ceil{ \frac{\secretnum}{\intervallen}}\bra{1+\frac{\samplenum}{\samplenum+\categoryleast+\categorynoise-2}}^{\categoryactual-\categoryleast}} \leq 
\log\bra{\ceil{\frac{\secretnum}{\intervallen}}+\secretnum-\secretnum\cdot\bra{\frac{\categoryleast+\categorynoise-1}{\samplenum+\categoryactual+\categorynoise-2}}^{\categoryactual-\categoryleast}}$, we have
\begin{align*}
\ceil{\frac{\secretnum}{\intervallen}}\bra{1+r}^{\categoryactual-\categoryleast} + \secretnum \bra{1-r}^{\categoryactual-\categoryleast} \leq \ceil{\frac{\secretnum}{\intervallen}} + \secretnum.
\end{align*}
Let $y\bra{\categoryactual-\categoryleast} = \ceil{\frac{\secretnum}{\intervallen}}\bra{1+r}^{\categoryactual-\categoryleast} + \secretnum \bra{1-r}^{\categoryactual-\categoryleast}$, and we can get that 
$\frac{\mathrm{d}^2 y\bra{\categoryactual-\categoryleast}}{\mathrm{d} \bra{\categoryactual-\categoryleast}^2}\geq 0$ when $\categoryactual-\categoryleast \geq 0$. Let $\alpha = \frac{\log \secretnum - \log \ceil{ \frac{\secretnum}{\intervallen}}}{\log \bra{1+r}}$, since $\bra{1-r}^\alpha \leq 1 / \bra{1+r}^\alpha$, we have
\begin{align*}
\ceil{\frac{\secretnum}{\intervallen}}\bra{1+r}^{\alpha} + \secretnum \bra{1-r}^{\alpha} \leq \ceil{\frac{\secretnum}{\intervallen}} + \secretnum.
\end{align*}
Since $y\bra{\categoryactual-\categoryleast}$ is a convex function and $y\bra{0} = \ceil{\frac{\secretnum}{\intervallen}} + \secretnum$, we can get that 
$\ceil{\frac{\secretnum}{\intervallen}}\bra{1+r}^{\categoryactual-\categoryleast} + \secretnum \bra{1-r}^{\categoryactual-\categoryleast} \leq \ceil{\frac{\secretnum}{\intervallen}} + \secretnum$ when $0\leq \categoryactual-\categoryleast\leq \alpha$.
Additionally, we can easily get that when $\categoryactual-\categoryleast > \alpha$,  \begin{align*}
\min\brc{\log \brb{\ceil{ \frac{\secretnum}{\intervallen}}\bra{1+\frac{\samplenum}{\samplenum+\categoryleast+\categorynoise-2}}^{\categoryactual-\categoryleast}}, \ 
\log\bra{\ceil{\frac{\secretnum}{\intervallen}}+\secretnum-\secretnum\cdot\bra{\frac{\categoryleast+\categorynoise-1}{\samplenum+\categoryactual+\categorynoise-2}}^{\categoryactual-\categoryleast}}} > \log \secretnum.
\end{align*}
Therefore, the privacy value is upper bounded by 
\begin{align*}
\privacyrr \leq
\begin{cases}
\log \brb{\ceil{ \frac{\secretnum}{\intervallen}}\bra{1+\frac{\samplenum}{\samplenum+\categoryleast+\categorynoise-2}}^{\categoryactual-\categoryleast}}, & 0\leq \categoryactual-\categoryleast< \alpha,\\
\log \secretnum, &\beta\leq\categoryactual-\categoryleast.
\end{cases}
\end{align*}

Similar to the analyses above, we can get that with increasing $\categorynoise$, the privacy value drops linearly when $0\leq \categoryactual-\categoryleast< \alpha$, and always achieves its trivial upper bound $\log \secretnum$ when $\categoryactual-\categoryleast \geq \alpha$.
\section{Privacy-Distortion Tradeoff Analysis when $\combinationset\subseteq\combinationsetall\not=\combinationsetestimate$}
\label{app:exp_more_res}

\begin{figure}[htbp]
    \centering
    \begin{subfigure}{0.46\textwidth}
         \centering
        \includegraphics[width=1\linewidth]{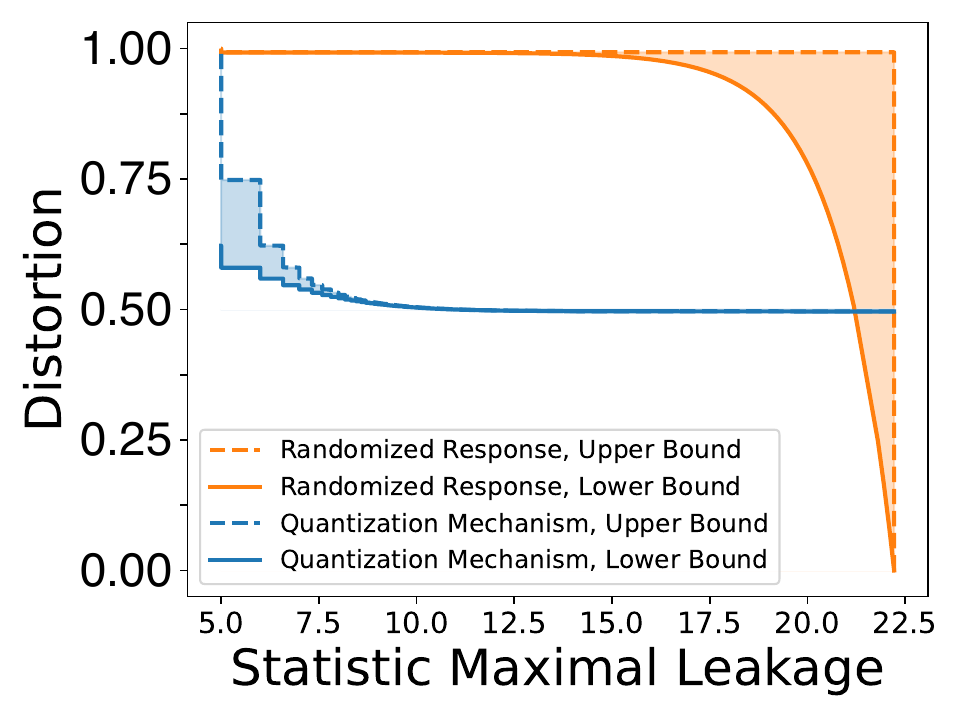}
         \caption{$\categoryactual-\categoryleast=5$.}
     \end{subfigure}
     \begin{subfigure}{0.46\textwidth}
         \centering
        \includegraphics[width=1\linewidth]{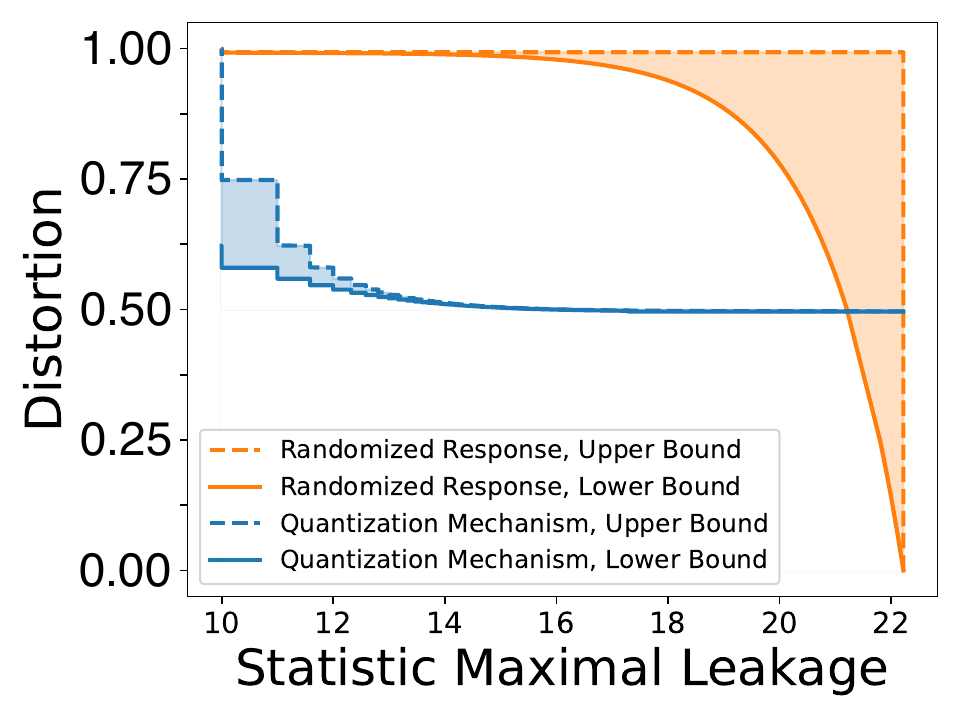}
         \caption{$\categoryactual-\categoryleast=10$.}
     \end{subfigure}

    \begin{subfigure}{0.46\textwidth}
    \centering
    \includegraphics[width=1\linewidth]{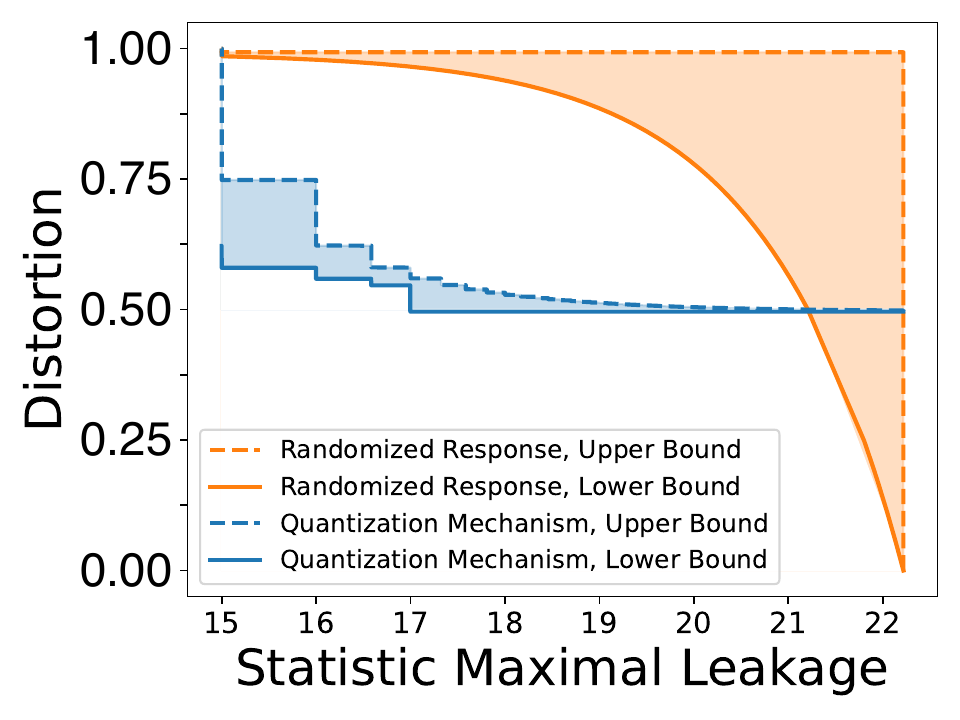}
        \caption{$\categoryactual-\categoryleast=15$.}
     \end{subfigure}
     \begin{subfigure}{0.46\textwidth}
    \centering
    \includegraphics[width=1\linewidth]{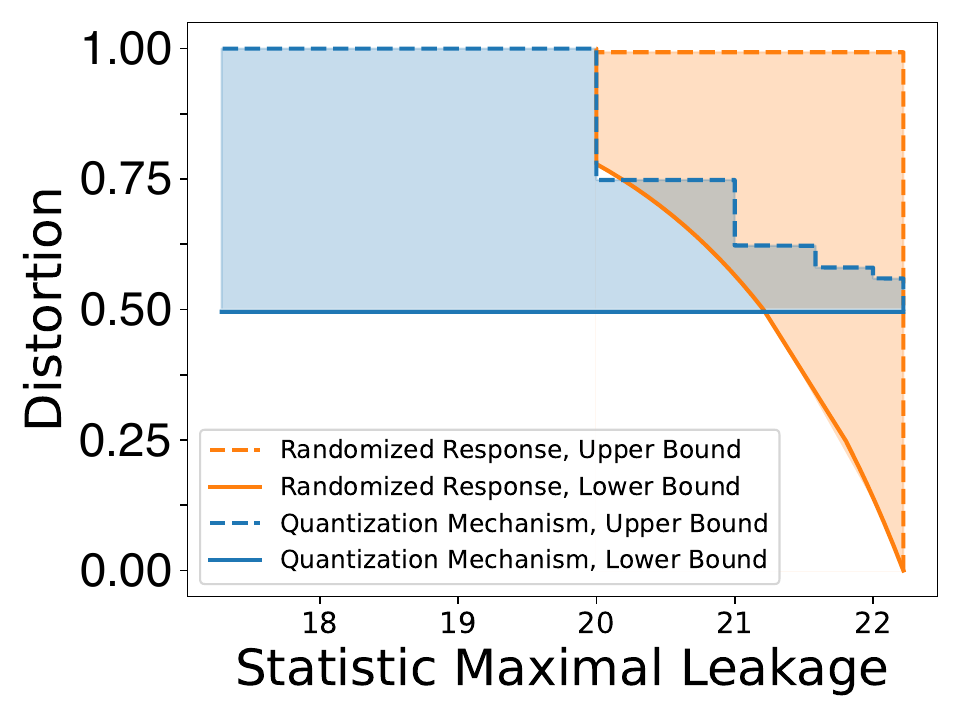}
        \caption{$\categoryactual-\categoryleast=20$.}
     \end{subfigure}
     \caption{The lower and upper bounds of the privacy-distortion performance of quantization mechanism and Randomized Response with different number of feasible attribute combinations the data holder misses, i.e., $\categoryactual-\categoryleast$.}
\end{figure}

We conduct numerical experiments to illustrate the lower and upper bounds of the privacy-distortion performance of our proposed quantization mechanism and Randomized Response. We consider a dataset with 5M samples and 100 feasible attribute combinations in total, i.e., $\categoryactual=100$. We vary the number of actually feasible attribute combinations that the data holder does not consider to be valid, i.e., $\categoryactual-\categoryleast$, from $5$ to $20$. We assume the data holder introduces $40$ invalid attribute combinations during the data release process, i.e., $\categorynoise=40$. We can observe that with the increase of the number of missed feasible attribute combinations, the minimal achievable statistic maximal leakage of both mechanisms increase. The reason is that correlation may exists between the secret value and the appearance of those missed feasible attribute combinations in the original dataset. The existence of certain missed feasible attribute combinations in the released dataset attacker may serve as the indication of certain secret values. We can also observe that with smaller number of missed feasible attribute combinations, the gaps between the lower and upper bounds of the privacy-distortion performance for both mechanism are smaller. With different number of missed feasible attribute combination, quantization mechanism consistently outperforms Randomized Response under most leakage regimes (statistic maximal leakage smaller than $21$).
\section{Alternative Privacy Measures}
\label{sec:metric_comparison}

Recall that our proposed \privname{} can be written as \begin{align*}
\sml = \sup_{\paramdistribution, \mathbb{P}_{\hat{G}|\Theta'}}\log \frac{\probof{\hat{G}=G}}{\sup_{\secretrv\in \secretvalueset} \mathbb{P}_G\bra{\secretrv}},
\end{align*}
which aims to measure the information gain of the attacker’s knowledge about
what the secret value is most likely to be after data release under the worst case prior $\paramdistribution$ and attack strategy $\mathbb{P}_{\hat{G}|\Theta'}$. 

In scenarios related to social justice and fairness, where knowing what the secret is most unlikely to be also reveals sensitive information, we propose an alternative privacy metric to measure the information gain of the attacker’s knowledge about
what the secret value is most unlikely to be:
\begin{align*}
\smlinf=
\begin{cases}
\infty, & \exists \paramdistribution, \mathbb{P}_{\hat{G}|\Theta'}:\ \mathbb{P}_G>0 \ \&\  \probof{\hat{G}=G}=0,\\
\sup_{\paramdistribution:\mathbb{P}_G>0; \mathbb{P}_{\hat{G}|\Theta'}}\log \frac{\inf_{\secretrv\in \secretvalueset} \mathbb{P}_G\bra{\secretrv}}{\probof{\hat{G}=G}}, & \text{otherwise}.
\end{cases}
\end{align*}
Similarly to \privname{}, $\smlinf$ is also defined under the worst-case prior $\paramdistribution$ and the attack strategy $\mathbb{P}_{\hat{G}|\Theta'}$, but consider the lowest prior probability of the secret value in the numerator of the ratio. A higher value of $\smlinf$ indicates that the attacker gains more information about the least likely secret. We set $\smlinf$ to infinity when the adversary can successfully infer what the secret value is not after the data release.

Furthermore, statistic maximal leakage measures the attack success probability over the randomness of the
released data. In a more conservative scenario, however, one may prefer to measure the attack success probability under the worst-case output. For this scenario, we propose alternative privacy measures, $\smlworse$ and $\smlworseinf$, which quantify the information leakage of the most likely and most unlikely secret values, respectively:
\begin{align*}
\smlworse&=\sup_{\paramdistribution, \theta'\in\Theta',\mathbb{P}_{\hat{G}|\theta'}} \log\frac{\probof{\hat{G}=G}}{\sup_{\secretrv\in \secretvalueset} \mathbb{P}_G\bra{\secretrv}},\\
\smlworseinf&=
\begin{cases}
\infty, & \exists \paramdistribution, \theta'\in\releaseparamset, \mathbb{P}_{\hat{G}|\theta'}:\ \mathbb{P}_G>0 \ \&\  \probof{\hat{G}=G}=0,\\
\sup_{\paramdistribution:\mathbb{P}_G>0;\theta'\in\releaseparamset; \mathbb{P}_{\hat{G}|\theta'}}\log \frac{\inf_{\secretrv\in \secretvalueset} \mathbb{P}_G\bra{\secretrv}}{\probof{\hat{G}=G}}, & \text{otherwise}.
\end{cases}
\end{align*}
In addition to the worst-case prior and attack strategy, $\smlworse$ and $\smlworseinf$ are also defined under the worst-case data output $\theta'$; thus, the attack strategy is represented as $\mathbb{P}_{\hat{G}|\theta'}$.

Additionally, inspired by differential privacy, we design the indistinguishability-based privacy measure $\ipmetric$, and we say that a data release mechanism $\mech$ is $\ip$-$\ipmetric$ if
\begin{align*} 
\probof{\mathcal{M}\bra{\theta_1}\in \releaseparamset_0} \leq e^{\ip}\cdot\probof{\mathcal{M}\bra{\theta_2}\in \releaseparamset_0}, \quad \forall \theta_1, \theta_2\in \paramset, \secretnotation\bra{\theta_1}\not= \secretnotation\bra{\theta_2}; \forall \releaseparamset_0 \subseteq \releaseparamset.
\end{align*}

We analyze the connections between those privacy measures in \cref{prop:alternative_metric}.

\begin{proposition}
\label{prop:alternative_metric}
If a data release mechanism $\mech$ is $\ip$-$\ipmetric$, then it satisfies $\sml\leq \ip$ and $\smlinf\leq \ip$.

If a data release mechanism $\mech$ satisfies $\smlworse\leq \ip_1$ and $\smlworseinf\leq \ip_2$, then it is $\bra{\ip_1+\ip_2}$-$\ipmetric$.
\end{proposition}

\begin{proof}
We first prove that $\sml\leq \ip$ and $\smlinf\leq \ip$ if a data release mechanism $\mech$ is $\ip$-$\mathcal{L}_{DP}$ as follows.

When a mechanism $\mech$ is $\ip$-$\ipmetric$, i.e.,
\begin{align*}
\probof{\mathcal{M}\bra{\theta_1}\in \releaseparamset_0} \leq e^{\ip}\cdot\probof{\mathcal{M}\bra{\theta_2}\in \releaseparamset_0}, \quad \forall \theta_1, \theta_2\in \paramset, \secretnotation\bra{\theta_1}\not= \secretnotation\bra{\theta_2}; \forall \releaseparamset_0 \subseteq \releaseparamset,
\end{align*}
we have
\begin{align}
\label{eqn:ip_detail}
\mathbb{P}_{\Theta'|\Theta}\bra{\theta'|\theta_1} \leq e^{\ip} \cdot \mathbb{P}_{\Theta'|\Theta}\bra{\theta'|\theta_2}, \quad \forall \theta_1, \theta_2\in \paramset, \secretnotation\bra{\theta_1}\not= \secretnotation\bra{\theta_2}; \forall \theta'\in \releaseparamset.
\end{align}

From \cref{prop:sml_calculation}, we know that
\begin{align*}
\sml &= \sup_{\mathbb{P}_{\Theta|G}\in\brc{0,1}}\log \sum_{\theta'\in\releaseparamset} \sup_{\secretrv\in \secretvalueset} \mathbb{P}_{\Theta'|\Theta}\bra{\theta'|\theta_\secretrv}.
\end{align*}
We can find a $\theta_0\in \paramset$ that satisfies $\exists \secretrv\in\secretvalueset, \mathbb{P}^*_{\Theta|G} \in \arg\sup_{\mathbb{P}_{\Theta|G}}\sml: \theta_0 = \theta_{\secretrv}$. 
Under the prior distribution $\mathbb{P}^*_{\Theta|G}$, we can get that
\begin{align*}
\sml & = \log \sum_{\theta'\in\releaseparamset} \sup_{\secretrv\in \secretvalueset} \mathbb{P}_{\Theta'|\Theta}\bra{\theta'|\theta_\secretrv}
\leq \log \sum_{\theta'\in\releaseparamset} e^{\ip} \cdot  \mathbb{P}_{\Theta'|\Theta}\bra{\theta'|\theta_0}
= \ip.
\end{align*}

From \cref{eqn:ip_detail}, we know that $\forall \theta\in\paramset, \theta'\in \releaseparamset: \mathbb{P}_{\Theta'|\Theta}\bra{\theta'|\theta}>0$. For the value of $\smlworse$, since $\forall \paramset: \mathbb{P}_G>0, \forall \mathbb{P}_{\hat{G}|\Theta'},$
\begin{align*}
\probof{\hat{G}=G} &\geq \sum_{\theta'\in\releaseparamset} \inf_{\secretrv\in\secretvalueset}\mathbb{P}_{G|\Theta'}\bra{\secretrv|\theta'}\mathbb{P}_{\Theta'}\bra{\theta'}\\
& = \sum_{\theta'\in\releaseparamset} \inf_{\secretrv\in\secretvalueset}\mathbb{P}_{\Theta'|G}\bra{\theta'|\secretrv}\mathbb{P}_{G}\bra{g}\\
&= \sum_{\theta'\in\releaseparamset} \inf_{\secretrv\in\secretvalueset}\sum_{\theta\in\secretset}\mathbb{P}_{\Theta'|\Theta}\bra{\theta'|\theta}\mathbb{P}_{\Theta|G}\bra{\theta|g}\mathbb{P}_{G}\bra{g}\\
&> 0,
\end{align*}
we can get that $\smlworse<\infty$.
Similar to the proof of \cref{prop:sml_calculation}, we can get that when $\smlworse<\infty$,
\begin{align*}
\smlworse &= \sup_{\mathbb{P}_{\Theta|G}\in\brc{0,1}}\log \frac{1}{\sum_{\theta'\in\releaseparamset} \inf_{\secretrv\in \secretvalueset} \mathbb{P}_{\Theta'|\Theta}\bra{\theta'|\theta_\secretrv}}.
\end{align*}
We can find a $\theta_1\in \paramset$ that satisfies $\exists \secretrv\in\secretvalueset, \mathbb{P}^*_{\Theta|G} \in \arg\sup_{\mathbb{P}_{\Theta|G}}\smlworse: \theta_1 = \theta_{\secretrv}$. 
Under the prior distribution $\mathbb{P}^*_{\Theta|G}$, we can get that
\begin{align*}
\smlworse & = \log \frac{1}{\sum_{\theta'\in\releaseparamset} \inf_{\secretrv\in \secretvalueset} \mathbb{P}_{\Theta'|\Theta}\bra{\theta'|\theta_\secretrv}}
\leq \log \frac{1}{\sum_{\theta'\in\releaseparamset} e^{-\ip} \cdot  \mathbb{P}_{\Theta'|\Theta}\bra{\theta'|\theta_0}}
= \ip.
\end{align*}

Above all, we have $\sml\leq \ip$ and $\smlinf\leq \ip$ when a data release mechanism $\mech$ is $\ip$-$\mathcal{L}_{DP}$.

If $\mech$ satisfies $\smlworse\leq \ip_1$ and $\smlworseinf\leq \ip_2$, we have
\begin{align*}
\ip_1+\ip_2
&\geq
\smlworse + \smlworseinf\\ 
&= \sup_{\paramdistribution, \theta'\in\Theta',\mathbb{P}_{\hat{G}|\theta'}} \log\frac{\probof{\hat{G}=G}}{\sup_{\secretrv\in \secretvalueset} \mathbb{P}_G\bra{\secretrv}} + \sup_{\paramdistribution:\mathbb{P}_G>0;\theta'\in\releaseparamset; \mathbb{P}_{\hat{G}|\theta'}}\log \frac{\inf_{\secretrv\in \secretvalueset} \mathbb{P}_G\bra{\secretrv}}{\probof{\hat{G}=G}}\\
&\geq \sup_{\paramdistribution:\mathbb{P}_G=1/\secretnum,\atop\theta'\in\releaseparamset} \bra{\log\frac{\sup_{\mathbb{P}_{\hat{G}|\theta'}}\probof{\hat{G}=G}}{\sup_{\secretrv\in \secretvalueset} \mathbb{P}_G\bra{\secretrv}}+\log \frac{\inf_{\secretrv\in \secretvalueset} \mathbb{P}_G\bra{\secretrv}}{\inf_{\mathbb{P}_{\hat{G}|\theta'}}\probof{\hat{G}=G}}}\\
&= \sup_{\paramdistribution:\mathbb{P}_G=1/\secretnum,\atop\theta'\in\releaseparamset}\log \frac{\sup_{\mathbb{P}_{\hat{G}|\theta'}}\probof{\hat{G}=G}}{\inf_{\mathbb{P}_{\hat{G}|\theta'}}\probof{\hat{G}=G}}\\
&= \sup_{\paramdistribution:\mathbb{P}_G=1/\secretnum,\atop\theta'\in\releaseparamset}\log\frac{\sup_{\secretrv\in\secretvalueset}\mathbb{P}\bra{\secretrv|\theta'}}{\inf_{\secretrv\in\secretvalueset}\mathbb{P}\bra{\secretrv|\theta'}}\\
&= \sup_{\paramdistribution:\mathbb{P}_G=1/\secretnum,\atop\theta'\in\releaseparamset}\log\frac{\sup_{\secretrv\in\secretvalueset}\mathbb{P}\bra{\theta'|\secretrv}\cdot\mathbb{P}_G\bra{\secretrv}}{\inf_{\secretrv\in\secretvalueset}\mathbb{P}\bra{\theta'|\secretrv}\cdot\mathbb{P}_G\bra{\secretrv}}\\
&\geq \sup_{\paramdistribution:\mathbb{P}_G=1/\secretnum,\mathbb{P}_{\Theta|G}\in\brc{0,1},\atop\theta'\in\releaseparamset}\log\frac{\sup_{\secretrv\in\secretvalueset}\mathbb{P}\bra{\theta'|\theta_\secretrv}}{\inf_{\secretrv\in\secretvalueset}\mathbb{P}\bra{\theta'|\theta_\secretrv}},
\end{align*}
where $\theta_{\secretrv}$ satisfies $\mathbb{P}_{\Theta|G}\bra{\theta_\secretrv|g}=1$. Therefore, we can get that 
\begin{align*}
\mathbb{P}_{\Theta'|\Theta}\bra{\theta'|\theta_1} \leq e^{\ip_1+\ip_2} \cdot \mathbb{P}_{\Theta'|\Theta}\bra{\theta'|\theta_2}, \quad \forall \theta_1, \theta_2\in \paramset, \secretnotation\bra{\theta_1}\not= \secretnotation\bra{\theta_2}; \forall \theta'\in \releaseparamset,
\end{align*}
i.e., $\mech$ is $\bra{\ip_1+\ip_2}$-$\ipmetric$.
\end{proof}

\section{\Privname{} under continuous parameter space}
\label{section:proof_sml_calculation_continuous}

\begin{proposition}
\label{prop:sml_calculation_continuous}
Under continuous parameter space, \privname{} satisfies
\begin{align*}
\sml = \sup_{\pdf{\Theta}: \pdfof{\Theta}{\theta}\in\brc{0, \pdfof{G}{\secretrv}}, \forall \theta\in \secretset, \secretrv\in\secretvalueset}\log \int_{\theta'\in\releaseparamset} \sup_{\secretrv\in \secretvalueset} \pdfof{\Theta'|\Theta}{\theta'|\theta_\secretrv}\mathrm{d}\theta',
\end{align*}
where $\theta_\secretrv$ satisfies $\pdfof{\Theta}{\theta_{\secretrv}}=\pdfof{G}{\secretrv}$.
\end{proposition}

\begin{proof}
Let $\secretset$ be the set of distribution parameters whose secret value is $\secretrv$, i.e, $\secretset = \brc{\theta\in \mathbf{\Theta} | \secretofparam = \secretrv}$.
For \privname $\sml$, we have $\sml = \sup_{\pdf{\Theta}}\log \frac{\int_{\theta'\in\releaseparamset} \sup_{\secretrv\in \secretvalueset} \pdf{G\Theta'}\bra{\secretrv,\theta'}\mathrm{d}\theta'}{\sup_{\secretrv\in \secretvalueset} \pdf{G}\bra{\secretrv}}$. Under a fixed distribution $\pdf{\Theta}$, we can get that 
\begin{equation}
\label{eqn:pi_continuous}
\begin{aligned}
\log \frac{\int_{\theta'\in\releaseparamset} \sup_{\secretrv\in \secretvalueset} \pdf{G\Theta'}\bra{\secretrv,\theta'}\mathrm{d}\theta'}{\sup_{\secretrv\in \secretvalueset} \pdf{G}\bra{\secretrv}}
&= \log \frac{\int_{\theta'\in\releaseparamset} \sup_{\secretrv\in \secretvalueset} \pdf{\Theta'|G}\bra{\theta'|\secretrv}\pdf{G}\bra{\secretrv}\mathrm{d}\theta'}{\sup_{\secretrv\in \secretvalueset} \pdf{G}\bra{\secretrv}}\\
&
\leq\log \frac{\int_{\theta'\in\releaseparamset} \sup_{\secretrv\in \secretvalueset} \pdf{\Theta'|G}\bra{\theta'|\secretrv}\cdot\sup_{\secretrv\in \secretvalueset}\pdf{G}\bra{\secretrv}\mathrm{d}\theta'}{\sup_{\secretrv\in \secretvalueset} \pdf{G}\bra{\secretrv}}\\
&= \log \int_{\theta'\in\releaseparamset} \sup_{\secretrv\in \secretvalueset}\pdf{\Theta'|G}\bra{\theta'|\secretrv}\mathrm{d}\theta'\\
&= \log\int_{\theta'\in\releaseparamset} \sup_{\secretrv\in \secretvalueset}
\int_{\theta\in \mathbf{\Theta}}
\pdf{\Theta'|G\Theta}\bra{\theta'|\secretrv, \theta}\cdot\pdf{\Theta|G}\bra{\theta|\secretrv}\mathrm{d}\theta \mathrm{d}\theta'\\
&= \log \hspace{-1mm} \int_{\theta'\in\releaseparamset} \sup_{\secretrv\in \secretvalueset}
\int_{\theta\in \mathbf{\Theta}}
\pdf{\Theta'|G\Theta}\bra{\theta'|\secretrv, \theta} \frac{\pdf{G|\Theta}\bra{\secretrv|\theta}\cdot\pdf{\Theta}\bra{\theta}}{\int_{\theta\in\mathbf{\Theta}}\pdf{G|\Theta}\bra{\secretrv|\theta}\cdot\pdf{\Theta}\bra{\theta}\mathrm{d}\theta}\mathrm{d}\theta\mathrm{d}\theta'
\\
&= \log \int_{\theta'\in\releaseparamset} \sup_{\secretrv\in \secretvalueset} \frac{\int_{\theta\in \mathbf{\Theta}_{\secretrv}}\pdf{\Theta'|\Theta}\bra{\theta'|\theta}\cdot\pdf{\Theta}\bra{\theta}\mathrm{d}\theta}{\int_{\theta\in\mathbf{\Theta}_{\secretrv}}\pdf{\Theta}\bra{\theta}\mathrm{d}\theta}\mathrm{d}\theta'\\
&=\log \int_{\theta'\in\releaseparamset} \sup_{\secretrv\in \secretvalueset} \mathbb{E}_{\Theta}\brb{\pdf{\Theta'|\Theta}\bra{\theta'|\theta}\big|\theta\in \secretset}\mathrm{d}\theta'\\
&\triangleq V_{\pdf{\Theta}}.
\end{aligned}
\end{equation}
\normalsize
This upper bound can be achieved when $\pdf{\Theta}$ satisfies 
\begin{equation}
\label{condition:pri1_continuous}
\begin{aligned}
    \forall \theta'\in\releaseparamset: \  \arg\sup_{\secretrv\in \secretvalueset} \pdf{\Theta'|G}\bra{\theta'|\secretrv} \cap \arg\sup_{\secretrv\in \secretvalueset}\pdf{G}\bra{\secretrv} \not= \emptyset,
\end{aligned}
\end{equation}
where $\pdf{\Theta'|G}\bra{\theta'|\secretrv}=\mathbb{E}_{\Theta}\brb{\pdf{\Theta'|\Theta}\bra{\theta'|\theta}\big| \theta\in\secretset}$, $\pdf{G}\bra{\secretrv}=\int_{\theta\in\secretset}\pdf{\Theta}\bra{\theta}\mathrm{d}\theta$.

Let ${\pdf{\Theta}^+}\in \arg\sup_{\pdf{\Theta}} V_{\pdf{\Theta}}$ and $\releaseparamset(\secretrv)$ be the set of released parameters satisfying $\secretrv=\arg\sup_{\secretrv\in \secretvalueset} \mathbb{E}_{\Theta}\brb{\pdf{\Theta'|\Theta}\bra{\theta'|\theta}\big|\theta\in \secretset}$
under ${\pdf{\Theta}^+}$. We can construct a prior $\pdf{\Theta}^*$ that satisfies  
\begin{itemize}
\item $\pdfof{G}{\secretrv_1} = \pdfof{G}{\secretrv_2}, \forall \secretrv_1, \secretrv_2\in \secretvalueset$; 
\item $\pdf{\Theta}\bra{\theta}\in\brc{0,\pdfof{G}{\secretrv}}, \forall \secretrv\in\secretvalueset, \theta\in\secretset$; 
\item $\pdf{\Theta}\bra{\theta} = 0$ if $\theta \not\in \arg\sup_{\theta\in\secretset}  \sum_{\theta'\in\releaseparamset(\secretrv)}\mathbb{P}_{\Theta'|\Theta}\bra{\theta'|\theta}, \forall \secretrv\in\secretvalueset$.
\end{itemize}
\normalsize

We can get that
\begin{align*}
V_{{\pdf{\Theta}^+}} &= \log \int_{\theta'\in\releaseparamset} \sup_{\secretrv\in \secretvalueset} \mathbb{E}_{\Theta}\brb{\pdf{\Theta'|\Theta}\bra{\theta'|\theta}\big|\theta\in \secretset}\mathrm{d}\theta'\\
&= \log \int_{\secretrv\in\secretvalueset}\int_{\theta'\in\releaseparamset(\secretrv)}\mathbb{E}_{\Theta}\brb{\pdf{\Theta'|\Theta}\bra{\theta'|\theta}\big|\theta\in \secretset}\mathrm{d}\theta'\mathrm{d}\secretrv\\
&= \log \int_{\secretrv\in\secretvalueset}\mathbb{E}_{\Theta}\brb{\int_{\theta'\in\releaseparamset(\secretrv)}\pdf{\Theta'|\Theta}\bra{\theta'|\theta}\mathrm{d}\theta'\bigg|\theta\in \secretset}\mathrm{d}\secretrv\\
&\leq \log \int_{\secretrv\in\secretvalueset} \sup_{\theta\in\secretset}\int_{\theta'\in\releaseparamset(\secretrv)}\mathbb{P}_{\Theta'|\Theta}\bra{\theta'|\theta}\mathrm{d}\theta'\mathrm{d}\secretrv\\
&= V_{\pdf{\Theta}^*}.
\end{align*}
Since ${\pdf{\Theta}^+}\in \arg\sup_{\pdf{\Theta}} V_{\pdf{\Theta}}$, we can get that $\pdf{\Theta}^*\in \arg\sup_{\pdf{\Theta}} V_{\pdf{\Theta}}$. Since $\pdf{\Theta}^*$ satisfies condition \ref{condition:pri1_continuous} and $\pdfof{\Theta}{\theta}\in\brc{0, \pdfof{G}{\secretrv}}, \forall \secretrv\in\secretvalueset, \theta\in \secretset$, we have
\begin{equation}
\begin{aligned}
\label{eqn:sml_variations}
\sml &= \sup_{\pdf{\Theta}}\log \frac{\int_{\theta'\in\releaseparamset} \sup_{\secretrv\in \secretvalueset} \pdf{G\Theta'}\bra{\secretrv,\theta'}\mathrm{d}\theta'}{\sup_{\secretrv\in \secretvalueset} \pdf{G}\bra{\secretrv}}\\
&= \sup_{\pdf{\Theta}}\log \int_{\theta'\in\releaseparamset} \sup_{\secretrv\in \secretvalueset}\pdf{\Theta'|G}\bra{\theta'|\secretrv}\mathrm{d}\theta'\\
&= \sup_{\pdf{\Theta}}\log \int_{\theta'\in\releaseparamset} \sup_{\secretrv\in \secretvalueset} \frac{\int_{\theta\in \mathbf{\Theta}_{\secretrv}}\pdf{\Theta'|\Theta}\bra{\theta'|\theta}\cdot\pdf{\Theta}\bra{\theta}\mathrm{d}\theta}{\int_{\theta\in\mathbf{\Theta}_{\secretrv}}\pdf{\Theta}\bra{\theta}\mathrm{d}\theta}\mathrm{d}\theta'\\
&= \sup_{\pdf{\Theta}}\log \int_{\theta'\in\releaseparamset} \sup_{\secretrv\in \secretvalueset} \mathbb{E}_{\Theta}\brb{\pdf{\Theta'|\Theta}\bra{\theta'|\theta}\big|\theta\in \secretset}\mathrm{d}\theta'\\
&= \sup_{\pdf{\Theta}: \pdfof{\Theta}{\theta}\in\brc{0, \pdfof{G}{\secretrv}}, \forall \secretrv\in\secretvalueset, \theta\in \secretset}\log \int_{\theta'\in\releaseparamset} \sup_{\secretrv\in \secretvalueset} \pdfof{\Theta'|\Theta}{\theta'|\theta_\secretrv}\mathrm{d}\theta',
\end{aligned}
\end{equation}
where $\theta_\secretrv$ satisfies $\pdfof{\Theta}{\theta_{\secretrv}}=\pdfof{G}{\secretrv}$.
\end{proof}

\end{appendices}

\end{document}